\documentclass[prb,twocolumn,aps,amssymb,amsmath,superscriptaddress,showpacs]{revtex4-2}

\usepackage{amsmath,amssymb,amsfonts,amsthm,esint,bm,xcolor}
\usepackage{moreverb,rotating,graphicx,siunitx}
\usepackage[T1]{fontenc}
\usepackage[utf8]{inputenc}
\usepackage{graphicx}
\newsavebox{\bigimage}
\usepackage{subcaption}
\usepackage{epsfig}
\usepackage[english]{babel} 
\usepackage{dsfont,siunitx}
\usepackage[colorlinks, linkcolor={blue}, citecolor={red}]{hyperref}
\usepackage{float} 

\DeclareSIUnit{\angstrom}{\textup{\AA}}

\newtheorem{remark}{Remark}
\newtheorem{lemma}{Lemma}

\def\L{{\mathbb L}}

\def\bbI{{\mathbb I}}
\def\bbV{{\mathbb V}}

\def\bbW{{\mathbb W}}
\def\1{{\mathds{1}}}

\def\rd{{\rm d}}
\def\re{{\rm e}}
\def\ri{{\rm i}}

\def\ba{{\mathbf a}}
\def\bb{{\mathbf b}}
\def\be{{\mathbf e}}
\def\bk{{\mathbf k}}

\def\bq{{\mathbf q}}
\def\br{{\mathbf r}}
\def\bx{{\mathbf x}}
\def\by{{\mathbf y}}

\def\bX{{\mathbf X}}
\def\bR{{\mathbf R}}
\def\bG{{\mathbf G}}

\def\bK{{\mathbf K}}
\def\bR{{\mathbf R}}

\def\bmm{{\mathbf m}}

\def\eps{{\varepsilon}}

\def\b0{{\bold 0}}

\def\per{{\rm{per}}}

\let\C\relax\newcommand{\C}{\mathbb{C}}\newcommand{\Z}{\mathbb{Z}}

\newcommand{\R}{\mathbb{R}}

\newcommand\cC{\mathcal{C}}\newcommand\cH{\mathcal{H}}\newcommand\cM{\mathcal{M}}\newcommand\cP{\mathcal{P}}\newcommand\cR{\mathcal{R}}\newcommand\cS{\mathcal{S}}\newcommand\cV{\mathcal{V}}\newcommand\cX{\mathcal{X}}

\DeclareMathOperator{\Ker}{Ker}
\def\d{{\rm d}}

\renewcommand{\ge}{\geqslant}\renewcommand{\le}{\leqslant}

\newcommand{\pa}[1]{\left( #1 \right)} 
\newcommand{\ab}[1]{\left|#1\right|} 
\newcommand{\nor}[2]{ \left| \! \left| #1 \right| \! \right|_{#2} } 

\def\eps{\varepsilon} 
\newcommand{\na}{\nabla} 

\newcommand{\f}[2]{\frac{#1}{#2}} 
\newcommand{\ind}[1]{_{\textup{#1}}} 

\newcommand{\mat}[1]{\begin{pmatrix} #1 \end{pmatrix}} 

\newcommand{\dd}{\tfrac{d}{2}}

\def\bdelta{{\boldsymbol\delta}}
\def\bsigma{{\boldsymbol\sigma}}

\def\lAngle{\langle\!\langle}
\def\rAngle{\rangle\!\rangle}
\newcommand{\ppa}[1]{\pa{\!\pa{#1}\!}_d}

\def\Lat{\mathbb{L}}
\def\RLat{{\mathbb{L}^*}}

\def\WS{{\Omega}}

\def\fR{{\mathfrak R}}
\def\fC{{\mathfrak C}}

\newcommand{\U}{U_{d,\theta}}
\newcommand{\ept}{\varepsilon_\theta}

\newcommand{\tbm}{\bm{V}}
\newcommand{\wAAd}{w\ind{AA}^{d=6.45}}

\begin{document}

\title{A simple derivation of moiré-scale continuous models \\ for twisted bilayer graphene}

\author{Éric Cancès}
\author{Louis Garrigue}
\affiliation{CERMICS, \'Ecole des Ponts and Inria Paris, 6 and 8 av. Pascal, 77455 Marne-la-Vallée, France.  (eric.cances@enpc.fr, louis.garrigue@enpc.fr)}

\author{David Gontier}
\affiliation{CEREMADE, University of Paris-Dauphine, PSL University, 75016 Paris, France \& ENS/PSL University, Département de Mathématiques et Applications, F-75005, Paris, France. (gontier@ceremade.dauphine.fr)}

\begin{abstract} We provide a formal derivation of a reduced model for twisted bilayer graphene (TBG) from Density Functional Theory. Our derivation is based on a variational approximation of the TBG Kohn-Sham Hamiltonian and asymptotic limit techniques. In contrast with other approaches, it does not require the introduction of an intermediate tight-binding model.  The so-obtained model is similar to that of the Bistritzer-MacDonald (BM) model but contains additional terms. Its parameters can be easily computed from Kohn-Sham calculations on single-layer graphene and untwisted bilayer graphene with different stackings. It allows one in particular to estimate the parameters $w_{\rm AA}$ and $w_{\rm AB}$ of the BM model from first-principles. The resulting numerical values, namely $w_{\rm AA}= w_{\rm AB} \simeq 126$ meV for the experimental interlayer mean distance are in good agreement with the empirical values $w_{\rm AA}= w_{\rm AB}=110$ meV obtained by fitting to experimental data. We also show that if the BM parameters are set to $w_{\rm AA}= w_{\rm AB} \simeq 126$ meV, the BM model is an accurate approximation of our reduced model. 
\end{abstract}

\maketitle

Moiré materials~\cite{Andrei21,Carr20} have attracted a lot of interest in condensed matter physics since, notably, the experimental discovery of Mott insulating and nonconventional superconducting phases \cite{Cao18} in twisted bilayer graphene (TBG) for specific small twist angles~$\theta$. In particular, the experiments reported in \cite{Cao18} were done with $\theta  \simeq 1.1^\circ$. For such small twist angles, the moiré pattern is quite large and typically contains of the order of 11,000 carbon atoms. In addition, the system is aperiodic (incommensurate), except for a countable set of values of $\theta$. All this makes brute force first-principle calculations extremely challenging. 

Most theoretical investigations on TBG rely on continuous models~\cite{Mele10,BM11_2,BM11,Lopes12} such as the Bistritzer-MacDonald (BM) model~\cite{BM11}, an effective continuous periodic model describing low-energy excitations in TBG close to half-filling. The BM Hamiltonian is a self-adjoint operator on $L^2(\R^2;\C^4)$ given by
\begin{equation} \label{eq:BM_Hamiltonian}
    H^{\rm BM}_\theta = \left( \begin{array}{cc} v_{\rm F} \bm\sigma_{-\theta/2} \cdot (-i \nabla) & \bm V(k_\theta \bx) \\ \bm V(k_\theta \bx)^* & v_{\rm F} \bm\sigma_{\theta/2} \cdot (-i \nabla) \end{array} \right),
\end{equation}
where $v_{\rm F}$ is the Fermi velocity in single-layer graphene, $\bm\sigma_{\pm\theta/2} = e^{\mp i \frac \theta 4 \sigma_3}(\sigma_1,\sigma_2) e^{\pm i \frac \theta 4 \sigma_3}$ are rotated Pauli matrices, and $k_\theta = \frac{8\pi}{3a_0} \sin(\theta/2)$, with $a_0$ the single-layer graphene lattice constant. The function $\bm V : \R^2 \to \C^2$ is quasiperiodic at the so-called moiré scale (see Section~\ref{sec:comparison_BM} for details) and depends on two empirical parameters $w_{\rm AA}$ and $w_{\rm AB}$ describing interlayer coupling in AA and AB stacking respectively. A rigorous mathematical derivation of the BM model from a tight-binding Hamiltonian whose parameters satisfy suitable scaling laws,  was recently proposed in \cite{Watson22}.
A simplified chiral BM model, obtained by setting $w_{\rm AA}=0$, was notably used in~\cite{TKV19,Becker21,WL21}  to prove the existence of perfectly flat bands at the Fermi level for a sequence of so-called ``magic'' angles. From a physical point of view, the existence of partially occupied almost flat bands in the single-particle picture may reflect the presence of localized strongly correlated electrons, and provide a possible explanation of the experimentally observed superconducting behavior~\cite{Po18}. An alternative approach to using effective models at the moiré scale is to develop atomic-scale models and efficient computational methods adapted to incommensurate 2D heterostructures~\cite{TdLMM12,JM14,Carr17,Carr18,Le18}.

This article is concerned with the derivation of BM-like effective models directly from Density Functional Theory (DFT). In contrast with other approaches~\cite{BM11,Moon13,Carr19,Fang19,Koshino20,Bernevig21,Watson22}, our derivation does not involve an intermediate tight-binding model and is based on real-space (not momentum-space) computations.


\section{Approximation of the Kohn-Sham Hamiltonian for TBG}
\label{sec:approx_KS}

\subsection{Single-layer graphene}

We denote the position variable by $\br = (\bx,z) \in \R^3$ where $\bx=(x_1,x_2) \in \R^2$ and $z \in \R$ are respectively the longitudinal (in-plane) and transverse (out-of-plane) position variables. Let $V$ be the Kohn-Sham potential for a single-layer graphene in the horizontal plane $z=0$. The space group of graphene is ${\rm Dg80} = {\rm D}_{6h} \ltimes \Lat$ (p6/mmm), so $V$ has the honeycomb symmetry and is $\L$-periodic, where $\Lat =  \Z  \ba_1 + \Z  \ba_2$ and
\[
\ba_1 = a_0 \begin{pmatrix}
    1/2 \\ - \sqrt{3}/2
\end{pmatrix}
\quad \text{and} \quad 
\ba_2 = a_0 \begin{pmatrix}
    1/2 \\ \sqrt{3}/2
\end{pmatrix}.
\]
Here, $a_0=\sqrt 3 a$ is the graphene lattice constant ($a \simeq 0.142$ nm $\simeq 2.68$ bohr is the carbon-carbon nearest neighbor distance). We set $\Omega := \R^2 / \Lat$ the Wigner-Seitz cell.

 \begin{figure}
     \centering

    \includegraphics[width=0.25\textwidth,trim={0cm 1.5cm 0cm 1.5cm},clip]{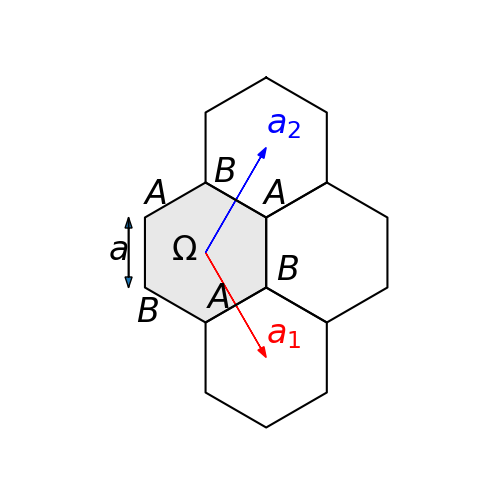}
    \includegraphics[width=0.22\textwidth,trim={0cm 2cm 0cm 2cm},clip]{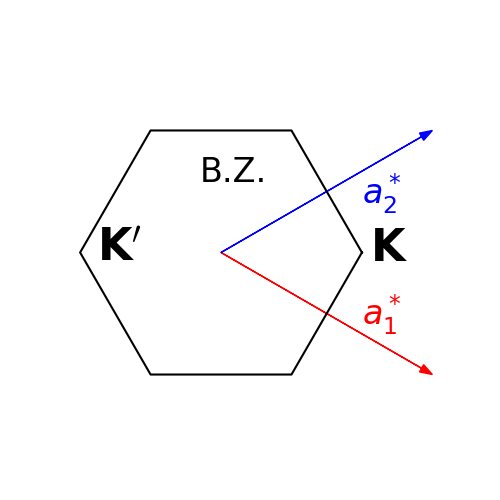} 
   \caption{Single-layer graphene. Above: atomic positions of the carbon atoms (A and B sublattices) in the physical space, and lattice vectors $\ba_1$ and $\ba_2$. Below: reciprocal lattice vectors $\ba_1^*$ and $\ba_2^*$, and first Brillouin zone in momentum space; the positions of the $\bK$ and $\bK'$ Dirac points are also indicated.}
    \label{fig:graphene}
 \end{figure}

The single-layer graphene Kohn-Sham Hamiltonian reads
\begin{equation}\label{eq:monolayer_Hamiltonian}
    H^{(1)} := - \frac12 \Delta + V.
\end{equation}
We denote by $H_\bk^{(1)}$ its Bloch fibers. Recall that these operators have domains representing the $\bk$--quasiperiodic boundary conditions $\Phi(\bx - \bR, z) = e^{i\bk \cdot \bR} \Phi(\bx, z)$ for all $\bR \in \Lat$. The map $\bk \mapsto H_\bk^{(1)}$ is $\RLat$-periodic, where $\RLat$ is the dual lattice of $\Lat$. Explicitly, $\RLat = \ba_1^* \Z + \ba_2^* \Z $ with
\[
\ba_1^* = \sqrt{3} k_D  \begin{pmatrix}
    \sqrt{3}/2 \\  -1/2
\end{pmatrix}
\quad \text{and} \quad 
\ba_2^* = \sqrt{3} k_D \begin{pmatrix}
    \sqrt{3}/2 \\ 1/2
\end{pmatrix},
\]
where $k_D := \frac{4\pi}{3 \, a_0}$. At the special Dirac point 
\begin{equation} \label{eq:def:K}
    \bK := \tfrac13 (\ba_1^* + \ba_2^*) = k_D (1, 0)^T,
\end{equation}
$H^{(1)}_\bK$ has an eigenvalue of multiplicity $2$ at the Fermi level $\mu_F$. We denote by $\Phi_1$ and $\Phi_2$ two corresponding eigenvectors, oriented so that $ \langle \Phi_1, (- \ri\nabla_\bx) \Phi_1 \rangle = \langle \Phi_2, (- \ri\nabla_\bx) \Phi_2 \rangle = (0,0)^T$, and
\begin{align} \label{eq:pre_Dirac}
    \langle \Phi_1,  (-\ri \nabla_\bx) \Phi_2 \rangle 
    = v_{\rm F}  \begin{pmatrix} 1 \\ -\ri \end{pmatrix},
\end{align}
with $v_F > 0$ the Fermi velocity. We refer to the Supplementary Material~\ref{supp:graphene} for an explanation of how to achieve such an orientation. Finally, we denote by $u_1$ and $u_2$ the periodic parts of $\Phi_1$ and $\Phi_2$, {\em i.e.} $u_j(\br) := \Phi_j(\br) \re^{ - \ri \bK \cdot \bx}$.

\subsection{Kohn-Sham model for TBG}

We now consider two parallel layers of graphene, separated by a distance $d > 0$, and with a twist angle $\theta$ between the two. More precisely, we first place the top layer in the plane $z=\dd$ and the bottom layer in the plane $z=-\dd$ in AA stacking, and then rotate counterclockwise the top layer by $-\theta/2$ and the bottom layer by $\theta/2$ around the $z$-axis, placing the origin at the center of a carbon hexagon. We set 
\begin{equation*}
c_\theta := \cos \frac{\theta}{2} 
\quad \text{and} \quad 
\ept :=  2 \sin \frac{\theta}{2}.
\end{equation*}
Note that $\ept \sim \theta$ in the small-angle limit. We denote by $R_\theta$ the 2D rotation matrix of angle $\theta \in \R$, specifically
\[
R_\theta := c_\theta \; \bbI_2 - \ept J, \quad \text{with} \quad J := \begin{pmatrix}
    0 & 1 \\ -1 & 0
\end{pmatrix},
\] 
and we introduce the unitary operator
\begin{equation*}
    (\U f)(\bx, z) := f(R_{-\theta/2}^* \bx, z) = f \left( c_\theta \bx - \tfrac12 \ept J \bx, z - \dd \right). 
\end{equation*}
The inverse of $\U$ is $\U^{-1}= \U^* = U_{-d,-\theta}$. 

\medskip

For twist angles $\theta$ giving rise to a periodic structure at the moiré scale,
the TBG Kohn-Sham potential is a well-defined moiré-periodic function. It is unclear how to define a mean-field potential for incommensurate twist angles. This problem actually occurs for all infinite aperiodic systems (see~\cite{CLL13} for a mathematical definition of a mean-field model in an ergodic setting). Here, we assume that this potential exists, and can be approximated using the procedure in~\cite{TSKCLPC16}. We consider an approximate Kohn-Sham potential for the TBG of the form
\begin{equation*}
     V^{(2)}_{d, \theta}(\bx, z) :=  (\U V)(\bx, z) + (\U^{-1} V)(\bx, z) + V_{{\rm int}, d}(z).
\end{equation*}
Each component $\U^{\pm 1} V$ represents a layer of graphene shifted by $\pm\dd \be_z$, and twisted by an angle $\mp \theta/2$. The last term $V_{{\rm int}, d}$ is a correction which takes into account the relaxation of the Kohn-Sham potential due to interlayer coupling. This term is constructed as follows. For each disregistry vector $\by \in \WS$, we denote by $V^{(2)}_{d, \by}$ the mean-field Kohn-Sham potential for the configuration where the two layers are aligned (no twist angle) with the top one shifted by $\by$ in the longitudinal direction. The potential $V_{{\rm int}, d}(z)$ is defined as the average 
\[
    V_{{\rm int}, d}(z) = \fint_{\WS} V_{{\rm int}, d, \by}(z) \rd \by
\]
where $\fint_\Omega := |\Omega|^{-1} \int_\Omega$, and with $V_{{\rm int}, d, \by}(z)$ defined by
\begin{equation*}
      \fint_{\WS}  \left( V^{(2)}_{d, \by}(\bx, z) - V(\bx, z + \dd) - V(\bx - \by, z -\dd)   \right) \rd \bx.
\end{equation*}
 In other words, $V_{\rm int, d}$ is the mismatch between the interacting Kohn-Sham potential and the non-interacting one, averaged over all possible disregistries $\by \in \Omega$. Note that $V_{\rm int, d}(z)$ only depends on the $z$--variable and satisfies $V_{{\rm int}, d}(-z) = V_{{\rm int}, d}(z)$.

\medskip

In what follows, we study the approximate TBG Kohn-Sham Hamiltonian
\begin{equation} \label{eq:Hdtheta2}
    H_{d, \theta}^{(2)} := - \frac12 \Delta + V^{(2)}_{d, \theta}(\bx, z).
\end{equation}
Our goal is to derive a 2D reduced model describing the low-energy band structure around the Fermi level $\mu_{\rm F}$ in the limit of small twist angles.


\section{Reduced model}

The potential $V^{(2)}_{d, \theta}$ is of the form $V^{(2)}_{d, \theta}(\bx, z) = v_{d}\left( c_\theta   \bx,   \ept   \bx , z\right)$, with 
\begin{align*}
&v_d(\bx, \bX, z):= \\
&\; V(\bx - \tfrac12 J \bX, z- \dd) + V(\bx + \tfrac12 J \bX, z +  \dd) + V_{{\rm int}, d}(z).
\end{align*}
The potential $v_d$ is $\Lat$--periodic in the first variable $\bx$, and $2 J \Lat$-periodic in its second variable~$\bX$. In the limit $\theta \to 0$, $v_{d}\left( c_\theta   \bx,   \ept   \bx , z\right)$ has a natural two-scale structure, so that $H_{d, \theta}^{(2)}$ could be studied using adiabatic theory, semiclassical analysis, and/or homogenization theory. In this article however, we take a different route, and present a simple approximation scheme to derive an effective Hamiltonian describing electronic transport around the Fermi level.

\subsection{Variational approximation of low-energy wavepacket dynamics}
\label{sec:VA}

The main idea of our approach is to project the time-dependent Schr\"odinger equation 
\begin{equation}\label{eq:3D-TDSE-TBG}
i \partial_t \Psi(t,\br) = (H^{(2)}_{d,\theta}-\mu_{\rm F}) \Psi(t,\br), \quad \Psi(0,\br)=\Psi_0(\br)
\end{equation}
onto the two-scale approximation space
\begin{equation}\label{eq:approx_space_TBG}
\cX_{d,\theta} := \left\{ (\alpha : \Phi)_{d,\theta}, \; \alpha : \R^2 \to \C^4 \right\},
\end{equation}
where we set
\[
    \left( \alpha : \Phi \right)_{d, \theta}(\bx, z) := \sum_{\eta \in \{ \pm 1\} \atop j \in \{ 1, 2\}} \alpha_{\eta, j} \left(\ept \bx \right) (\U^\eta \Phi_j)(\bx, z).
\]
The trial functions in $\cX_{d, \theta}$ are linear combinations of four wave-packets, each one consisting of an envelope function $\alpha$ oscillating at the moiré scale $\ept^{-1}$ multiplied by one of the two (translated and rotated) Bloch waves $\Phi_1$ or $\Phi_2$ of one of the two layers. Note that both the TBG approximation subspace $\cX_{d, \theta}$ and Hamiltonian $H_{d, \theta}^{(2)}$ depend on the small parameter $\theta$.

\medskip

Given an initial state of the form $\Psi_0 = (\alpha^0 : \Phi)_{d, \theta}$, the true solution $\Psi(t)$ of~\eqref{eq:3D-TDSE-TBG} is expected to be close to $(\alpha(\ept t) : \Phi)_{d, \theta}$ up to times of order $\ept^{-1}$, where $\alpha(t)$ satisfies $\alpha(t = 0) = \alpha^0$, and solves the projected equation
\begin{align} \label{eq:proj_Schrodinger}
        & i \partial_t \left\langle \left( \widetilde{\alpha} : \Phi \right)_{d, \theta}, \left( \alpha(\ept t) : \Phi \right)_{d, \theta} \right\rangle \nonumber \\
        & \quad =  \left\langle \left( \widetilde{\alpha} : \Phi \right)_{d, \theta}, \left( H_{d, \theta}^{(2)} - \mu_F \right) \left( {\alpha}(\ept t) : \Phi \right)_{d, \theta} \right\rangle
\end{align}
for all $\widetilde{\alpha} : \R^2 \to \C^4$. The time variable has to be rescaled as $\tau := \ept t$ in order to obtain wave-packet propagation with finite velocity at the moir\'e scale.

\subsection{Formulation of the reduced model}
\label{sec:reduced_model}

It follows from tedious calculations detailed in Appendix~\ref{sec:LOT} that if we let $\theta$ go to zero in~\eqref{eq:proj_Schrodinger} for fixed trial smooth functions $\widetilde \alpha$ and $\alpha$, we obtain the asymptotic equality 
\begin{equation} \label{eq:asymptotic_VF}
\ri \partial_\tau \langle \widetilde \alpha, \cS_d \alpha (\tau) \rangle = \langle \widetilde \alpha, \cH_{d,\theta} \alpha(\tau) \rangle
+ O(\ept^\infty),
\end{equation}
where the overlap operator $\cS_d$ and the Hamiltonian operator $\cH_{d,\theta}$ act on $L^2(\R^2;\C^4)$, and are defined by
\begin{equation*}
  \cS_d :=  \begin{pmatrix}
        \bbI_2 & \Sigma_d(\bX) \\ \Sigma_d^*(\bX) & \bbI_2
    \end{pmatrix}
\end{equation*}
and
\begin{widetext}
    \begin{align}
        \cH_{d, \theta} & := \begin{pmatrix}
            v_{\rm F} \bm\sigma_{-\theta/2} \cdot (- \ri \nabla_\bX)  & \ept^{-1} \bbV_d(\bX) \\ \ept^{-1} \bbV_d (\bX)^*  & v_{\rm F} \bm\sigma_{\theta/2} \cdot (- \ri \nabla_\bX)  
        \end{pmatrix} +
        \begin{pmatrix}
            \ept^{-1} \bbW_d^+ & 0 \\
            0 & \ept^{-1} \bbW_d^-
        \end{pmatrix} \nonumber \\
        & \quad  +
        \begin{pmatrix}
            0 & c_\theta J( - \ri \nabla \Sigma_d(\bX)) \cdot (- \ri \nabla) \\
            c_\theta J( - \ri \nabla \Sigma_d^*)(\bX) \cdot (- \ri \nabla) & 0
        \end{pmatrix} 
        - \frac{\ept}{2} \nabla \cdot \left(  \begin{pmatrix}
            \bbI_2 & \Sigma_d(\bX) \\
            \Sigma_d^*(\bX) & \bbI_2
        \end{pmatrix} \nabla \bullet\right). \label{eq:ourH}
    \end{align}
\end{widetext}
The $2 \times 2$ matrix-valued functions $\Sigma_d(\bX)$, $\bbV_d(\bX)$ and $\bbW_d(\bX)$ are given by 
\begin{align*}
\left[ \Sigma_d(\bX)\right]_{jj'} &:= \re^{-\ri J \bK \cdot \bX}  (\!( u_{j}, u_{j'} )\!)^{+-}_d(\bX),  \\ 
       \left[ \bbV_d(\bX) \right]_{jj'}  &:= \re^{ -\ri J \bK \cdot \bX} (\!( \pa{V+V_{\text{int},d}(\cdot+ \dd)} u_j, u_{j'} )\!)_d^{+-}(\bX),  \\ 
             \left[ \bbW_d^\pm(\bX) \right]_{j j'} 
         &:= (\!(u_{j} \overline{u_{j'}}, V )\!)_d^{\pm \mp}(\bX) + \pa{W^{\pm}_{\text{int},d}}_{jj'} , 
\end{align*}
where for all $\Lat$-periodic functions $f$ and $g$, 
\begin{align} 
&\ppa{f,g}^{\eta\eta'}(\bX) :=  \nonumber \\
  &\;   \int_{\Omega \times \R} \overline{f\left(\bx-\eta \tfrac12 J\bX,z-\eta \dd\right)} \, g \left(\bx-\eta' \tfrac12 J\bX,z-\eta' \dd \right)\, \d\bx \, \d z, \label{eq:def:ppa}
\end{align}
and where
\begin{align*}
 \pa{W^{\pm}_{\text{int},d}}_{jj'} :=  \int_{\Omega \times \R} (\overline{u_{j}}u_{j'})(\bx, z \mp \dd) V_{{\rm int}, d}(z) \rd \bx \, \rd z. \label{eq:def:Wintd}
\end{align*}

All these quantities can be computed from the single-layer potential $V$, the Bloch wave-functions $u_1$ and $u_2$, and the Kohn-Sham correction potential $V_{{\rm int}, d}$.

It is not obvious that the operator $\cH_{d,\theta}$ is Hermitian. It is however the case. For instance, one can check that $\overline{(\!( u, v )\!)^{+-}_d}(\bX) = (\!( \overline{u}, \overline{v} )\!)^{+-}_d(\bX)$, which proves that the matrices $\bbW^+_d(\bX)$ and $\bbW_d^-(\bX)$ are Hermitian. In addition, from the equality, ${\rm div} J \nabla  = - \partial_{x_1x_2}^2  + \partial_{x_2x_1}^2 = 0$, we see that the third matrix in the definition \eqref{eq:ourH} of $\cH_{d, \theta} $ also defines an Hermitian operator.

\medskip

The asymptotic equality \eqref{eq:asymptotic_VF} leads us to introduce the following effective model for the propagation of low-energy wave-packets in TBG:
\begin{equation} \label{eq:def:reduced_eqt}
    \ri \partial_\tau \left( \cS_d \alpha \right) (\tau, \bX) = \left( \cH_{d, \theta} \alpha \right)(\tau, \bX).
\end{equation}

At this stage, we have kept all the terms in $\cS_d$ and $\cH_{d, \theta}$, and only thrown away the remainders of order $\ept^\infty$. Indeed, $\ept$ should not be considered as the only small parameter in this problem. The interlayer coupling is also small, hence the operators $\bbW_d^\pm$ can be considered to be small as well. The interplay between the two parameters $\ept$ and $w$ (the interlayer characteristic interaction energy) is subtle and will be explored in a future work.

\begin{remark} A similar approach can be used to derive an effective model for wave-packet propagation in monolayer graphene in an energy range close to the Fermi level and localized in the $K$-valley in momentum space. One obtains the effective Hamiltonian    
    \[
        \cH^{(1)}_\eps = v_F \bsigma \cdot (- \ri \nabla ) + \frac12 \eps (- \ri \nabla)^2
    \]
   acting on $L^2(\R^2; \C^2)$. Neglecting the second term, we recover the usual massless Dirac operator. This effective model was rigorously derived in~\cite{FefWei13} using other methods.
\end{remark}

\subsection{Translational covariance} 
For any $\Lat$-periodic functions $f$ and $g$, the maps $ \ppa{f,g}^{\pm\mp}(\bX)$ defined in~\eqref{eq:def:ppa} are $J \L$-periodic.
This suggests to introduce the rescaled moiré lattice 
\[
\L_{\rm M} :=  J \L
\]
which satisfies $\Lat_{\rm M} = \ba_{1, \rm M} \Z + \ba_{2, \rm M} \Z$ with lattice vectors $\ba_{1,\rm M} =  J \ba_1$ and $\ba_{2,\rm M} =  J \ba_2$. Explicitly,
$$
\ba_{1,\rm M} = a_0 \begin{pmatrix}
    - \sqrt{3}/2 \\  - 1/2 
\end{pmatrix}
\quad \text{and} \quad 
\ba_{2,\rm M} = a_0 \begin{pmatrix}
    \sqrt{3}/2 \\  -1/2 
\end{pmatrix}.
$$
The corrresponding Wigner--Seitz cell is $\Omega_{\rm M} := \R^2 / \Lat_{\rm M}$, and its dual basis is given by $\ba_{1,\rm M}^* = J \ba_1^*$ and $\ba_{2,\rm M}^* = J \ba_2^*$, that is
$$
\ba_{1,\rm M}^* = \sqrt{3} k_{\rm D}  \begin{pmatrix}
    - 1/2 \\ - \sqrt{3}/2 
\end{pmatrix}
\quad \text{and} \quad 
\ba_{2,\rm M}^* =  \sqrt{3} k_{\rm D} \begin{pmatrix}
    1/2 \\ -\sqrt{3}/2 
\end{pmatrix}.
$$
We also introduce the vectors (see Fig.~\ref{fig:path})
\begin{align*}
    & \bq_1 := k\ind{D} \mat{0 \\ -1} = \frac 13\pa{ \ba_{1,\rm M}^* + \ba_{2,\rm M}^*}, \\ 
    & \bq_{2} := k\ind{D}\mat{\f{\sqrt 3}{2} \\ \f 12} =  \frac 13 \pa{-2 \ba_{1,\rm M}^* + \ba_{2,\rm M}^*} = R_{\frac{2 \pi}{3}} \bq_1, \\  
    & \bq_{3} := k\ind{D}\mat{-\f{\sqrt 3}{2} \\ \f 12}= \frac 13 \pa{\ba_{1,\rm M}^* - 2 \ba_{2,\rm M}^*} = R_{\frac{2 \pi}{3}} \bq_2. 
\end{align*}
These vectors satisfy $\bq_1 + \bq_2 + \bq_3 = \b0$ and correspond to the $\bK$--valley of the moiré Brillouin zone. Actually, we have $\bq_1 = J \bK$.

Going back to our reduced model, we see that the diagonal elements $\bbW_d^{\pm}$ are $\Lat_M$-periodic, while, for $\bR_M \in \Lat_M$,
\begin{equation} \label{eq:quasi-periodicity}
    \bbV_d(\bX - \bR_M) = \re^{ \ri \bq_1 \cdot \bR_M} \bbV_d(\bX), 
\end{equation}
and similarly for $\Sigma_d$. Writing $\bR_M = m_1 \ba_{1, M} + m_2 \ba_{2, M}$, we have $\re^{\ri \bq_1 \cdot \bR_M} = \omega^{(m_1 + m_2)}$, where we set $\omega := \re^{\ri \frac{2 \pi}{3}}$. Since $\omega^3=1$, our model is $3\L_{\rm M}$--periodic.

Thus, although the true moiré pattern generated by the superposition of two twisted honeycomb lattices is not periodic for a generic twist angle $\theta$, our reduced model is. In some sense, the moiré pattern looks $\ept^{-1}\L_{\rm M}$-periodic at the mesoscopic scale $\ept^{-1}$ for $\theta \ll 1$ (see Fig.~\ref{fig:moire}).

\begin{figure}[H]
    \begin{center}
        \includegraphics[height=5cm,trim={5cm 5cm 5cm 8cm},clip]{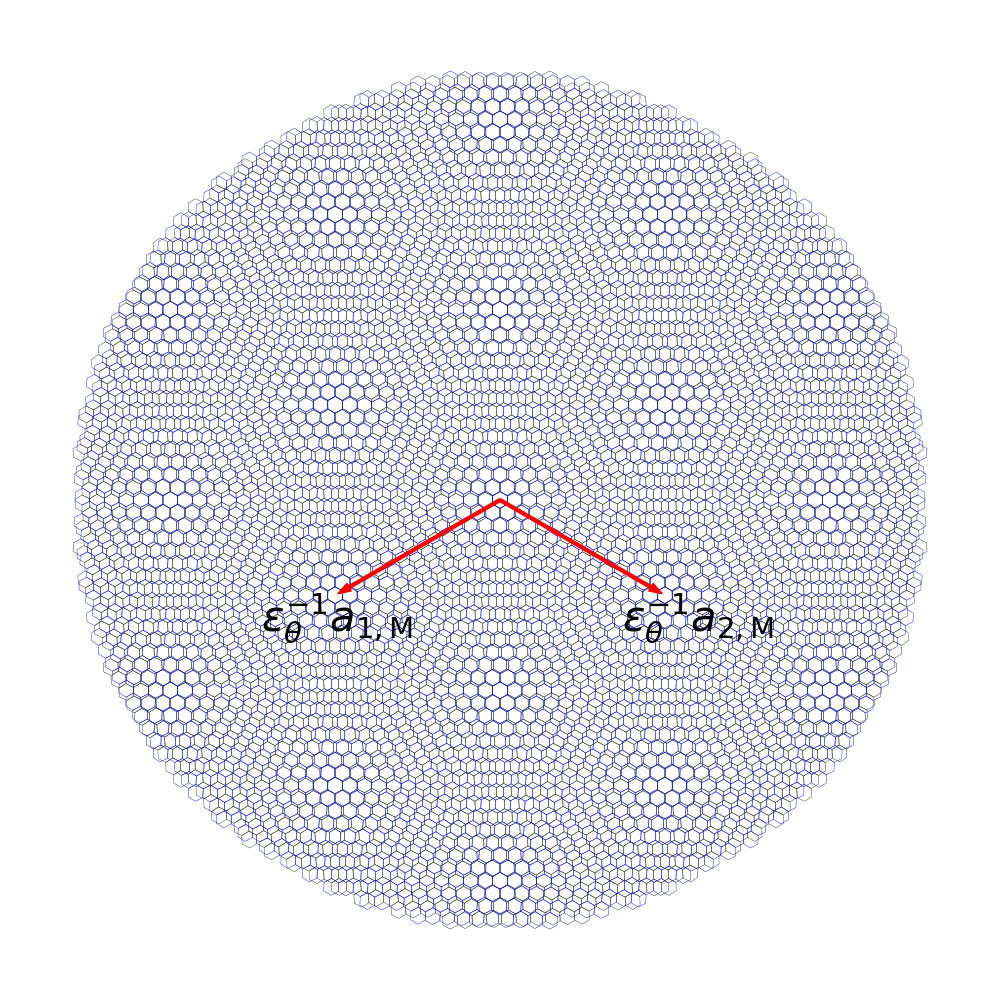}
        \caption{Moiré lattice vectors for twist angle $\theta \simeq 2.2^\circ$.}\label{fig:moire}
    \end{center}
\end{figure}



\subsection{Comparison with the Bistritzer-MacDonald model}%
\label{sec:comparison_BM}

The BM Hamiltonian \eqref{eq:BM_Hamiltonian} can be written more explicitly~\cite{BM11,TKV19, Watson22} as
\begin{equation*} 
H^{\rm BM}_\theta = \left( \begin{array}{cc} v_{\rm F} \bm\sigma_{-\theta/2} \cdot (- \ri \nabla_\bx)  & \tbm(\ept \bx) \\ \tbm (\ept \bx)^*  & v_{\rm F} \bm\sigma_{\theta/2} \cdot (- \ri \nabla_\bx)   \end{array} \right),
\end{equation*} 
with
\begin{equation} \label{eq:BM_pot_V}
\tbm(\bX) :=   \mat{w_{\rm AA} \, G(\bX) & w_{\rm AB} \,  \overline{F(-\bX)}  \\ w_{\rm AB}  \, F(\bX)  & w_{\rm AA} \, G(\bX)},
\end{equation}
where $w_{\rm AA}$ and $w_{\rm AB}$ are the two real parameters describing the interlayer coupling in AA and AB stacking, and
\begin{align}  \label{eq:defFG}
F(\bX) & := \re^{-\ri  \bq_1 \cdot \bX} + \omega \re^{-\ri  \bq_2 \cdot \bX} + \omega^2 \re^{-\ri  \bq_3 \cdot \bX}, \\
G(\bX) & := \re^{-\ri  \bq_1 \cdot \bX} + \re^{-\ri  \bq_2 \cdot \bX} +  \re^{-\ri  \bq_3 \cdot \bX}.
\end{align} 
The BM potential satisfies, for all $\bR_M \in \Lat_M$,
\begin{equation}\label{eq:translation_V}
    \bm V(\bX-\bR_M) = \re^{\ri \bq_1 \cdot \bR_M} V(\bX) =  \omega^{(m_1+m_2)} \bm V(\bX),
\end{equation}
so that the BM potential $\tbm$ and our reduced potential $\bbV$ have the same covariance symmetries. They actually share many other symmetries (see Section~\ref{app:symmetries} in the Supplementary Materials for a comprehensive analysis of the symmetries of our reduced model).

Rescaling lengths and energies as $\bX=\ept \bx$ and ${\cal E}=\ept^{-1} E$ (in the BM model, the Fermi level is set to zero), we obtain the rescaled BM Hamiltonian 
\begin{align}\label{eq:BM_rescaled}
\cH^{\rm BM}_\theta = \begin{pmatrix}
    v_{\rm F} \bm\sigma_{-\theta/2} \cdot (- \ri \nabla_\bX)  & \ept^{-1} \tbm(\bX) \\ \ept^{-1} \tbm (\bX)^*  & v_{\rm F} \bm\sigma_{\theta/2} \cdot (- \ri \nabla_\bX)  
\end{pmatrix}.
\end{align}

Going back to our model in~\eqref{eq:ourH}, we see that the two models are similar under the following assumptions (which will be justified numerically in Section~\ref{sec:numerics}):
\begin{enumerate}
\item the matrix $\Sigma_d(\bX)$ and its gradient can be neglected;
\item the term $- \frac 12 \ept \Delta$ can be neglected, which is the case if the oscillations of the envelope functions $\alpha_{\eta,j}$ at the moiré scale contribute more than those at the atomic scale;
\item the functions $\bbW_d^\pm$ are almost proportional to the identity matrix (and thus only induce a global energy shift);
\item the function $\bbV_d$ is close to~\eqref{eq:BM_pot_V} for some well-chosen parameters $w_{AA}$ and $w_{AB}$.
\end{enumerate}

As for the last point, first-principle values of the BM parameters can be inferred from our reduced model by setting
\begin{align}
w_{\rm AA}^d &:= \frac{1}{3|\Omega_{\rm M}|} \int_{\Omega_{\rm M}} [\bbV_d]_{11}(\bX) \overline{G(\bX)} \, \rd \bX,  \label{eq:value_wAA} \\
w_{\rm AB}^d &:= \frac{1}{3|\Omega_{\rm M}|} \int_{\Omega_{\rm M}} [\bbV_d]_{21}(\bX) \overline{F(\bX)} \, \rd \bX, \label{eq:value_wAB}
\end{align}
where we used the fact that $\int_{\Omega_{\rm M}} |F|^2= \int_{\Omega_{\rm M}} |G|^2= 3|\Omega_{\rm M}|$. Note that from~\eqref{eq:quasi-periodicity} and~\eqref{eq:translation_V}, the integrands are $\L_{\rm M}$-periodic functions, and that $| \Omega_{\rm M} | = | \Omega| = \frac{\sqrt 3}2 a_0^2$. We prove in Appendix~\ref{app:wAA} that $w_{AA}^d = w_{AB}^d$.

\section{Numerical results}
\label{sec:numerics}

In this section, we numerically study the generalized spectral problem associated with the operators $(\cH_{d, \theta},\cS_d)$. Due to the energy shift in \eqref{eq:3D-TDSE-TBG} and the rescaling by $\ept$, the spectrum of the operator $H^{(2)}_{d,\theta}$ close to $\mu_{\rm F}$ is related to the spectrum of $(\cH_{d, \theta},\cS_d)$ around $0$ by
$$
\sigma(H^{(2)}_{d,\theta}) \simeq \mu_{\rm F} + \ept \sigma(\cH_{d, \theta},\cS_d) = \mu_{\rm F} + \ept \sigma(\widetilde\cH_{d, \theta},\widetilde\cS_d),
$$
where $\widetilde\cH_{d, \theta}$ and $\widetilde\cS_d$ are $\L_{\rm M}$-periodic operators obtained from $\cH_{d, \theta}$ and $\cS_d$ by the gauge transformation specified in the next section.


\subsection{Gauge transformation}
\label{sec:gauge}

First, we perform a gauge transformation in order to remove the phase factors in \eqref{eq:quasi-periodicity} and end up with an $\L_{\rm M}$-periodic model. The same arguments can be used for the BM model. Let $\bK_1$ and $\bK_2$ be two vectors such that $\bK_1 - \bK_2 = \bq_1$, {\em e.g.} $\bK_2 = \bq_3$, and $\bK_1 = -\bq_2$ (recall that $\bq_1 + \bq_2 + \bq_3 = \b0$). We introduce the unitary multiplication operator
\[
P(\bX) := \begin{pmatrix}
    \re^{ \ri \bK_1 \cdot \bX} \bbI_2 & \b0 \\
    \b0 & \re^{ \ri \bK_2 \cdot \bX} \bbI_2
\end{pmatrix}
\]
with inverse $P^{-1} (\bX) = P^*(\bX) = P(-\bX)$. First, we have
\[
    \widetilde{\cS}_{d}:= P\cS_{d}P^* = \begin{pmatrix}
        \bbI_2 & \widetilde{\Sigma_d}(\bX) \\ \widetilde{\Sigma_d}^*(\bX) & \bbI_2
    \end{pmatrix}, 
\]
with $\widetilde{\Sigma_d}(\bX) = \re^{\ri (\bK_1 - \bK_2)\cdot \bX} \Sigma_d(\bX)$. Using the definition of $\Sigma_d$ and the fact that $\bq_1 = J \bK$, we obtain
\[
    \widetilde{\Sigma_d}(\bX) =  \re^{\ri \bq_1 \cdot \bX} \Sigma_d(\bX) = \ppa{ u_{j}, u_{j'} }^{+-}(\bX).
\]
The $\widetilde{\cS_d}$ matrix-valued function is now $\Lat_M$-periodic. Similarly, with components given by 
\begin{align*}
\left[ \widetilde{\bbV}_d(\bX) \right]_{jj'}  &:= \ppa{ \left(V+V_{\text{int},d}(\cdot+ \dd) \right) u_j, u_{j'}}^{+-}(\bX),   \\
\left[ \widetilde{\bbW}_d^\pm(\bX) \right]_{j j'} &:= \ppa{u_{j} \overline{u_{j'}}, V}^{\pm \mp}(\bX) + \pa{W^{\pm}_{\text{int},d}}_{jj'}, \\
\widetilde A_d & := e^{i\bq_1 \bx}(-i\na \Sigma_d) = (-i\nabla_{\bq_1}) \widetilde \Sigma_d,  
\end{align*}
we find, using the notation $\nabla_\bk := \nabla - \ri \bk$,
{\small
\begin{widetext}
    \begin{align}
        \widetilde{\cH}_{d, \theta} & = \begin{pmatrix}
            v_{\rm F} \bm\sigma_{-\theta/2} \cdot (- \ri \nabla_{\bK_1})  & \ept^{-1} \widetilde{\bbV}_d(\bX) \\ \ept^{-1} \widetilde{\bbV}_d (\bX)^*  & v_{\rm F} \bm\sigma_{\theta/2} \cdot (- \ri \nabla_{\bK_2})  
        \end{pmatrix} +
        \begin{pmatrix}
            \ept^{-1} \widetilde{\bbW}_d^+ & 0 \\
            0 & \ept^{-1} \widetilde{\bbW}_d^-
        \end{pmatrix} \nonumber \\
        & \quad  +
        \begin{pmatrix}
            0 & c_\theta J  \widetilde{A}_d(\bX) \cdot (- \ri \nabla_{\bK_2}) \\
            c_\theta J \widetilde{A}_d^*(\bX) \cdot (- \ri \nabla_{\bK_1}) & 0
        \end{pmatrix} 
        + \frac{\ept}{2} \begin{pmatrix}
            ( - \ri \nabla_{\bK_1})^2 & (- \ri \nabla_{\bK_1} ) \cdot \left[ \widetilde{\Sigma_d} (- \ri \nabla_{\bK_2} ) \bullet \right] \\
            (- \ri \nabla_{\bK_2}) \cdot \left[ \widetilde{\Sigma_d}^* (- \ri \nabla_{\bK_1} ) \bullet \right] & ( - \ri \nabla_{\bK_2} )^2 
        \end{pmatrix}.
    \end{align}
\end{widetext}
}
In this gauge, the model is $\Lat_{\rm M}$ periodic, and we can apply the usual Bloch transform to compute its band diagram. For the sake of illustration, we display on Fig.~\ref{fig:band_diag} and \ref{fig:bands_comparisions} the band diagrams of the BM model (black) and of our continuous model (red). The path in momentum space used to produce the bands diagrams is displayed in Fig.~\ref{fig:path}.

\begin{figure}[H]
\centering
  \includegraphics[width=0.22\textwidth,trim={5cm 6cm 8cm 6cm},clip]{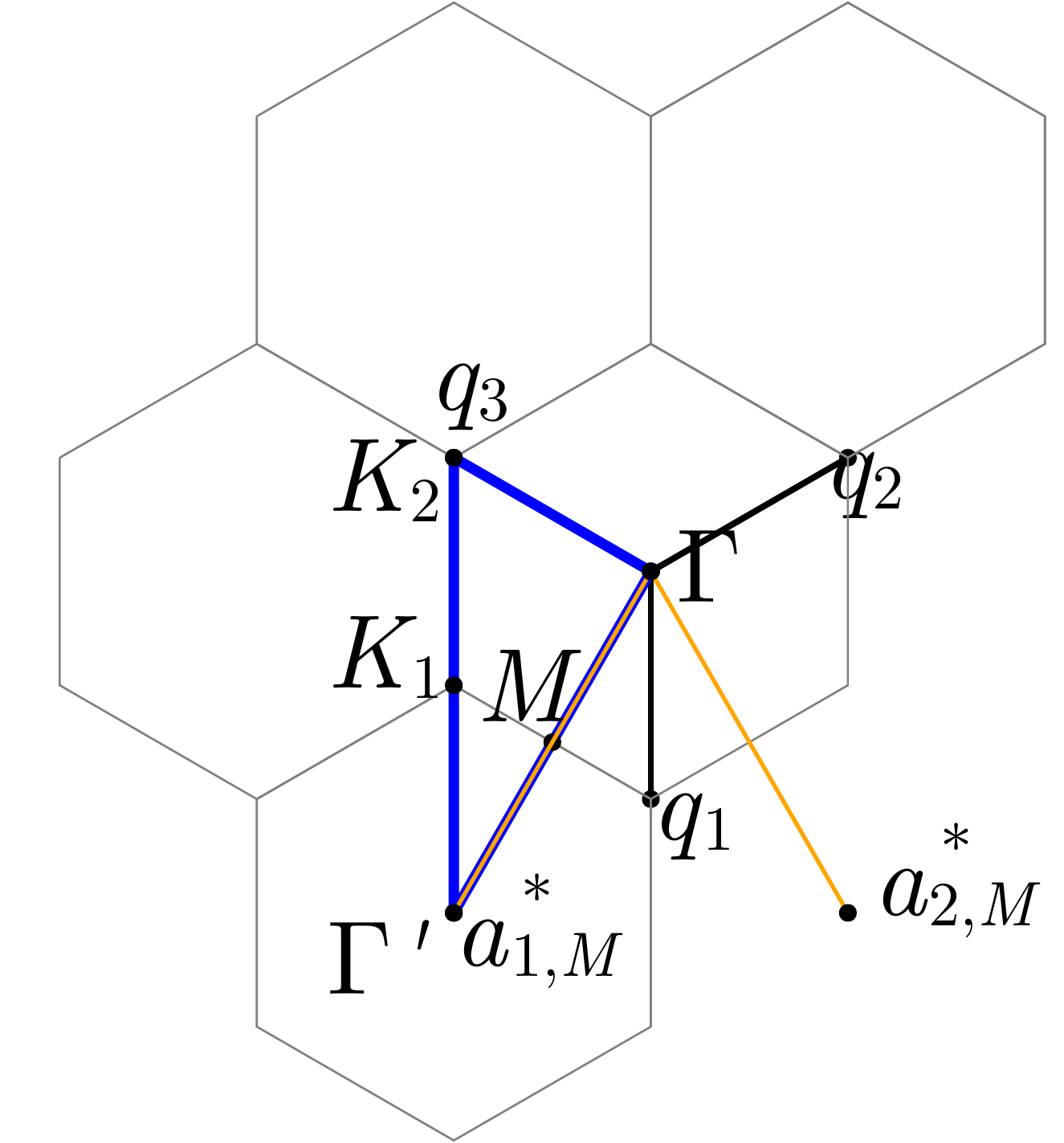}
\caption{Path $\bK_2 \to \bK_1 \to \Gamma' \to M \to \Gamma \to \bK_2$ in momentum-space, the hexagon centered at $\Gamma$ corresponds to the mini Brillouin zone} \label{fig:path}
\end{figure}

Quantities are computed with $d= 6.45$ bohr for our effective model. There are at least two ways of characterizing special angles associated to almost flat bands: the standard one is to consider the local minimizers of the Fermi velocity (called magic angles); an alternative consists in considering the local minimizers of the almost flat bands bandwidth. We choose here the second way. The first minimizing angle is $\theta \simeq 1.175^\circ$ for the BM model (black lines) with $w\ind{AA} = w\ind{AB} = 110$ meV, and $\theta \simeq 1.164^\circ$ for our effective model (red lines). A noticeable difference between our effective model and the BM model is that for both definitions of the special angles (minimal Fermi velocity vs minimal bandwidth), the BM Hamiltonian is not gapped at the flat bands while ours is so. Finally, we remark that the BM model with $w\ind{AA} = w\ind{AB} = 110$ meV (the empirical values used in \cite{BM11}) is a better approximation of our model than the BM model with $w\ind{AA} = w\ind{AB} = 126$ meV (the values derived from DFT, see next section). 

A thorough comparison of the band diagrams of various continuous models (including atomic relaxation) will be the matter of a forthcoming paper.

\begin{figure}[H]
\centering

  \includegraphics[width=0.22\textwidth,trim={0.4cm 0cm 0cm 0cm},clip]{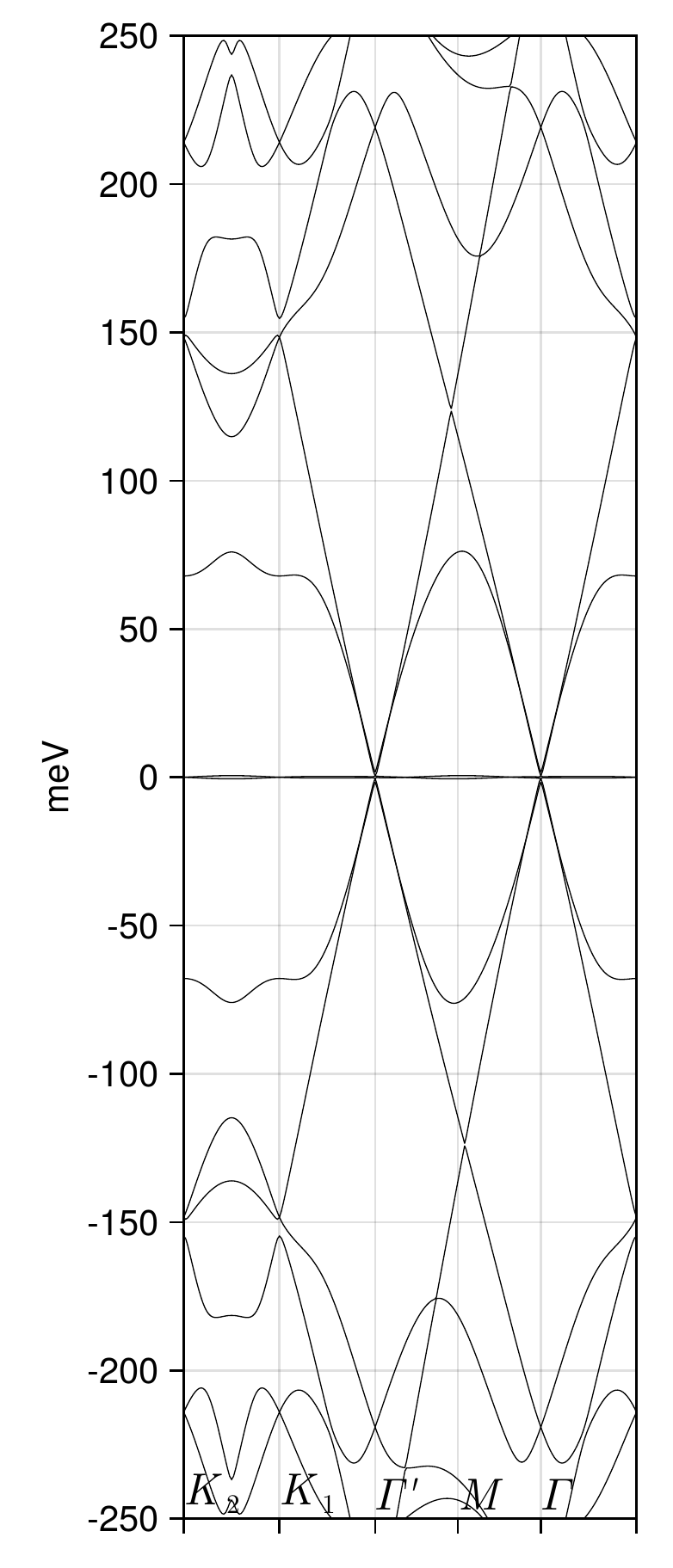}
  \includegraphics[width=0.22\textwidth,trim={0.4cm 0cm 0cm 0cm},clip]{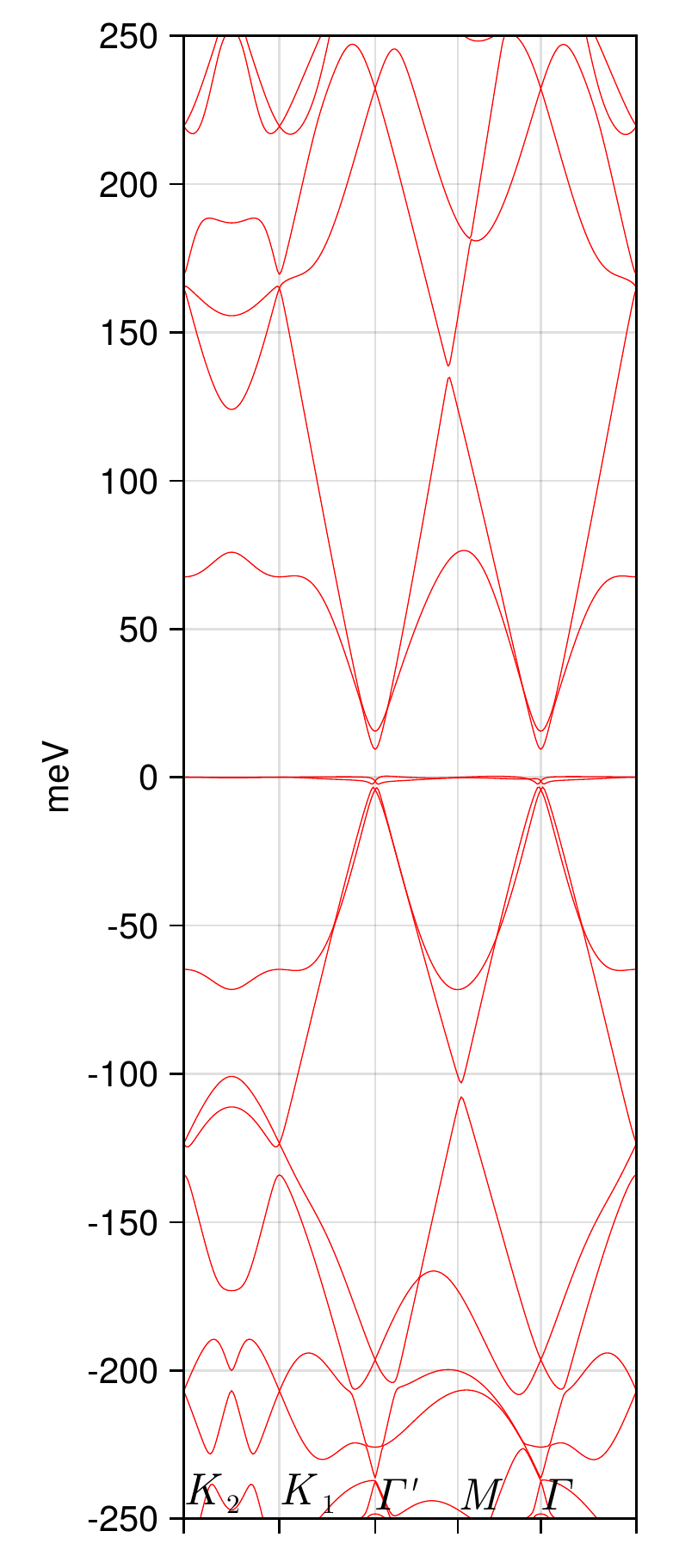}

\caption{Left: band diagram of $\cH^{\rm BM}_\theta$ (black) with $w\ind{AA} = w\ind{AB} = 110$ meV for the first bandwidth minimizing angle $\theta \simeq 1.175^\circ$ (bandwidth $1.1$ meV). Right: band diagram of $(\widetilde\cH_{d,\theta},\widetilde\cS_d)$ (red) for the first minimizing bandwidth angle $\theta \simeq 1.164^\circ$ (bandwidth $2.3$ meV). In both cases, the zero of the energy scale is the center of the almost flat band which is obtained by shifting the middle band diagram (red) of 957 meV.}
\label{fig:band_diag}
\end{figure}

\begin{figure}[H]
\centering
  \includegraphics[width=0.4\textwidth]{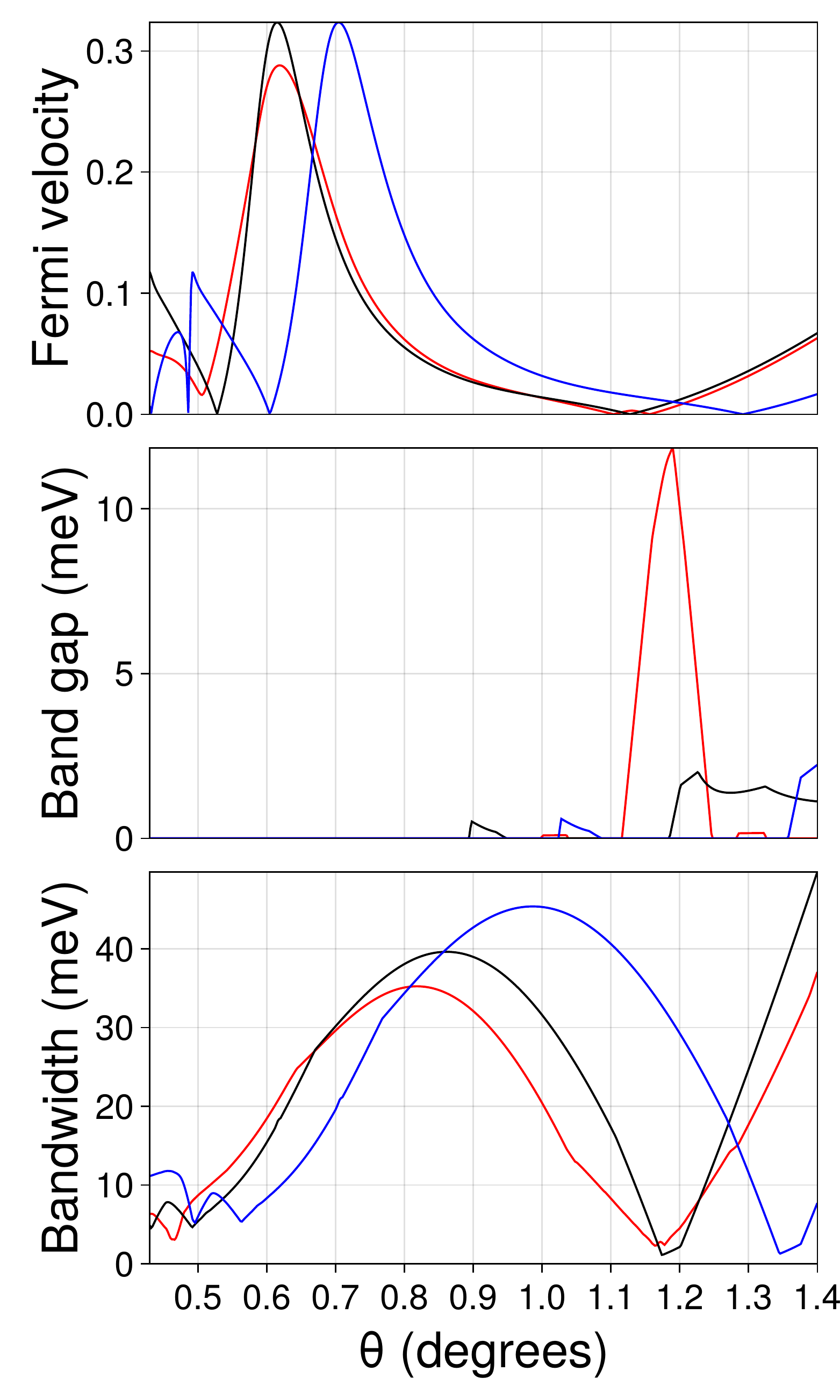}

\caption{Normalized Fermi velocities $v_{\text{F},\theta}/v\ind{F}$, band gaps, and bandwidths of the almost flat bands, where $v\ind{F}$ is the Fermi velocity of the monolayer, of $\cH^{\rm BM}_\theta$ (black) with $w\ind{AA} = w\ind{AB} = 110$ meV, of $\cH^{\rm BM}_\theta$ (blue) with $w\ind{AA} = w\ind{AB} = 126$ meV, and of $(\widetilde\cH_{d,\theta},\widetilde\cS_d)$ (red) as functions of $\theta$. }
\label{fig:bands_comparisions}
\end{figure}

\subsection{Numerical details}

In order to numerically compute the Fourier coefficients of $[\widetilde{\Sigma_d}]_{jj'}$, $[\widetilde{\bbV_d}]_{jj'}$, and $[ \widetilde{\bbW_d^\pm}]_{jj'}$, we have developed a code in Julia \cite{Julia}, interfaced with the DFTK planewave DFT package~\cite{DFTK}.

The single-layer graphene Kohn-Sham model is solved with DFTK using the PBE exchange-correlation functional with Goedecker-Teter-Hutter (GTH) pseudopotential, a unit cell of height $L_z= 110$ bohr, an energy cut-off of $E\ind{cut} = 900$ eV, and a $(5 \times 5 \times 1)$ $k$-point grid. We extract from the DFTK computation $\L$-periodic functions $u_1$ and $u_2$ such that $\Phi_1(\br) = e^{i \bK \cdot \bx} u_1(\br)$ and $\Phi_2(\br) = e^{i \bK \cdot \bx} u_2(\br)$ form an orthogonal basis of $\Ker \pa{H^{(1)}_\bK - \mu\ind{F}}$, where $H^{(1)}_\bK$ is the Bloch fiber of the single-layer graphene Kohn-Sham Hamiltonian at the Dirac point $\bK$. We also extract the local component of the single-layer graphene Kohn-Sham potential $V$ (as well as the required information on the nonlocal component of the carbon atom GTH pseudopotential, see Supplementary Materials~\ref{sec:NLP} for details). More precisely, DFTK returns the Fourier coefficients of the $\L$-periodic functions $u_j$ (assumed to be well-oriented, see~\eqref{eq:pre_Dirac}) and $V$, of the form
\begin{equation} \label{eq:FourierDecomposition_f}
    f(\br) = \sum_{m_1,m_2,m_z \in \Z} [f]_{m_1,m_2,m_z} \frac{\re^{\ri \left((m_1\ba_1^*+m_2\ba_2^*) \cdot \bx+ m_z \frac{2\pi z}{L_z}\right)}}{|\Omega|^{1/2}L_z^{1/2}}.
\end{equation}
At the discrete level, these sums are finite and run over the triplets of integers $(m_1,m_2,m_z) \in \Z^3$ such that 
$$
\frac{|m_1\ba_1^*+m_2\ba_2^*+\bK|^2+m_z^2 \frac{4\pi^2}{L_z^2}}2 \le E_{\rm cut}.
$$

The Kohn-Sham potential $V^{(2)}_{{\rm int},d}$ is computed by averaging the disregistries $\by$ in a $(5 \times 5)$ uniform grid of the graphene unit cell. In the results reported below, we used the experimental interlayer mean distance $d = 6.45$ bohr. In accordance with the results in~\cite{TSKCLPC16}, the potential $V_{\text{int},d,\by}(z)$ only slightly depends on $\by$:
\begin{align*}
    \delta_{V_{\text{int},d}} &:= \f{\int_{\Omega\times \R} \ab{V_{\rm{int},d,\by}(z) - V_{ \rm int,d}(z)}^2 \d \by \d z}{\ab{\Omega} \int_{\R} V\ind{int}(z)^2 \; \d z} \; \simeq 1 \times 10^{-4}.
\end{align*}
The effective potential $V_{\text{int},d}$, and in-plane averages of the vertically shifted single-layer Kohn-Sham potential $V$ and Bloch wave densities $|\Phi_1|^2$ are plotted in Fig.~\ref{fig:Vint}.

\begin{figure}[h]
    \begin{center}
        \includegraphics[width=0.5 \textwidth]{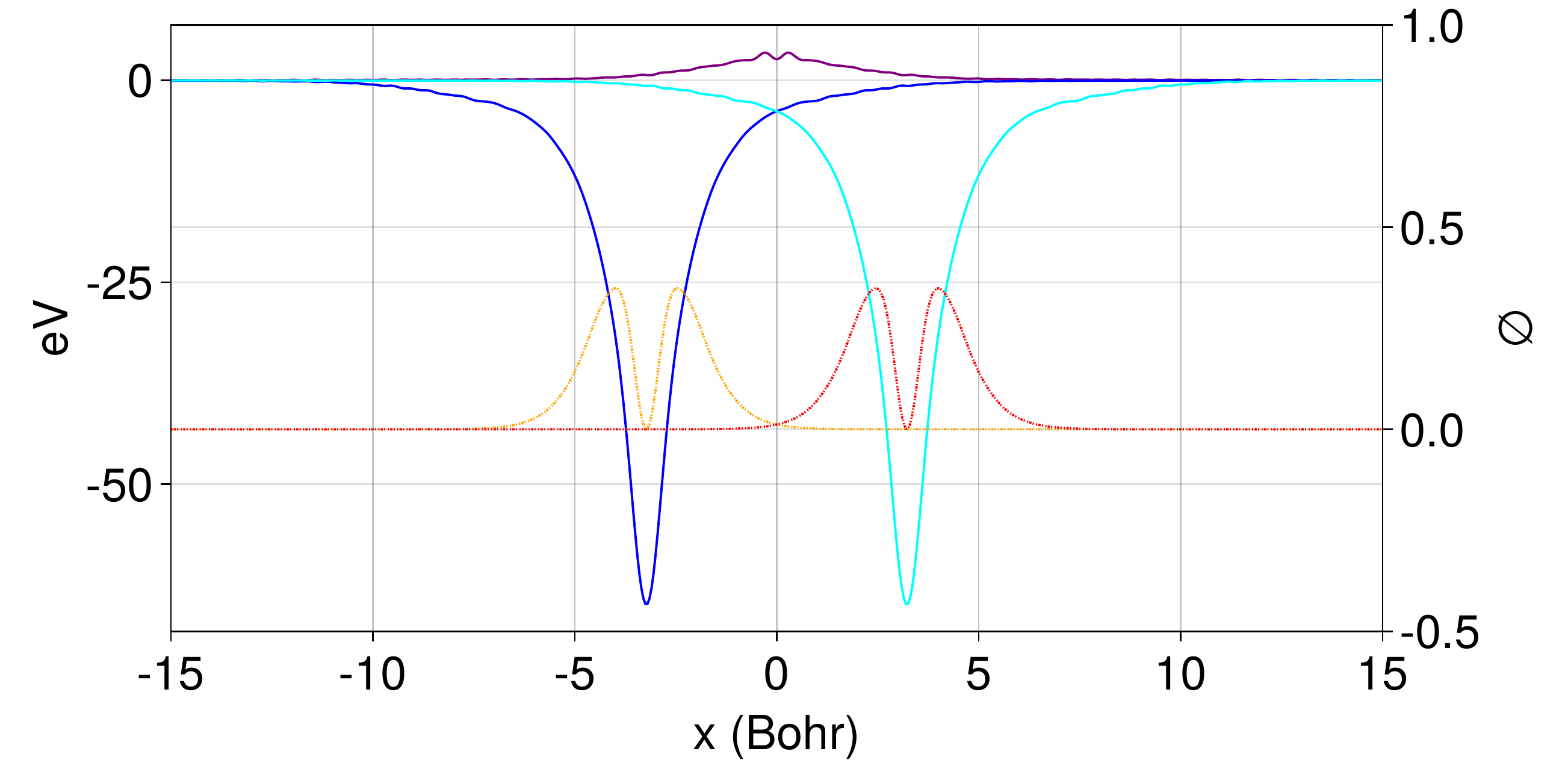}
        \includegraphics[width=0.22 \textwidth]{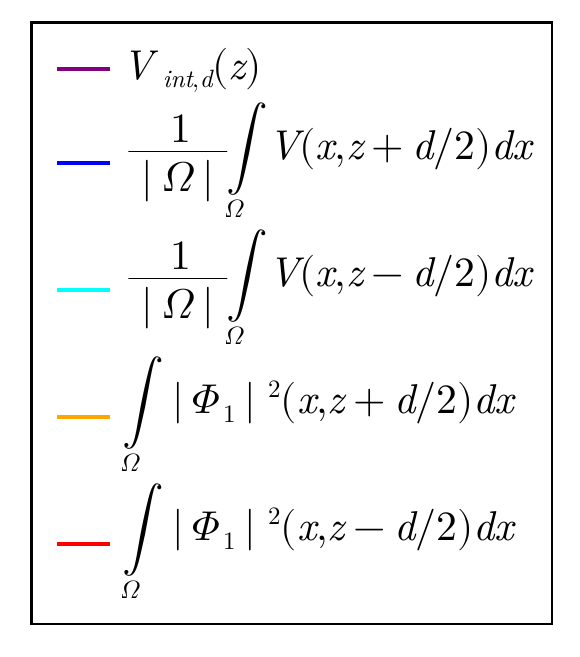}
        \caption{Effective potential $V_{\text{int},d}$, and in-plane averages of the vertically shifted single-layer Kohn-Sham potential $V$ and Bloch wave densities $|\Phi_1|^2$ for interlayer distance $d=6.45$ bohr.\label{fig:Vint}} 
    \end{center}
\end{figure}

For $f$ and $g$ of the form~\eqref{eq:FourierDecomposition_f}, we find that
\[
\ppa{f,g}^{+-}(\bX) = \sum_{m_1, m_2 \in \Z}  [f|g]_{d, m_1, m_2} \dfrac{\re^{ \ri (m_1 \ba_{1,M}^* + m_2 \ba_{2, M}^*) \cdot \bX} }{| \Omega_M |^{1/2}}
\]
where the coefficient $[f|g]_{d, m_1, m_2}$ is given by
\begin{equation*}
   |\Omega_{\rm M}|^{1/2} \sum_{m_z \in \Z} [\overline{f}]_{-m_1,-m_2,m_z} [g]_{-m_1,-m_2,m_z} \re^{\ri m_z \frac{ 2 \pi}{L_z} d}.
\end{equation*}


%
For $d = 6.45$ bohr, we obtain the BM parameters 
$$
w_{\rm AA}^{d=6.45 \, a.u.}= w_{\rm AB}^{d=6.45 \, a.u.} \simeq 126 \; {\rm meV},
$$
in good agreement with the value $w_{\rm AA} = w_{\rm AB} = 110$ meV chosen in \cite{BM11} to fit experimental data.

\subsection{Numerical justification of the Bistritzer-MacDonald model}

As discussed in Section~\ref{sec:comparison_BM}, the BM model can be deduced from our reduced model by assuming that $\Sigma_d(\bX)$ and its gradient can be neglected, that $\bbW_d^{\pm}(\bX)$ is proportional to the identity matrix, and that $\bbV_d(\bX)$ is of the form \eqref{eq:BM_pot_V} for some well-chosen parameters $w_{\rm AA}$ and $w_{\rm AB}$. To test these assumptions, we first plot in Figs.~\ref{fig:Sigma}-\ref{fig:W_plus_without_mean} the real-space structures and magnitudes of the functions $[\Sigma_d]_{jj'}(\bX)$, $\ab{\nabla \Sigma_d(\bX)}$, $[\bbW_d^\pm]_{jj'}(\bX) - \fint_\Omega [\bbW_d^\pm]_{jj'}$ and $[\bbV_d]_{jj'}(\bX) -
\bm V_{jj'}(\bX)$ for $d = 6.45$ bohr and $w_{\rm AA}= w_{\rm BB}=126 \; {\rm meV}$. We can see these fields are indeed small in the relevant units: $\Sigma_d(\bX)$ is small ($\sim 0.03$ compared to $1$), $\ab{\nabla \Sigma_d(\bX)}$ is small ($\sim 0.03$ compared to the Fermi velocity $v_{\rm F} \simeq 0.380$), and 
$\bbV_d(\bX) - \tbm(\bX)$ and $\bbW_d^{\pm}(\bX) - \fint_\Omega \bbW_d^{\pm}$ are small (resp. $\sim 1$ meV and $\sim 40$ meV compared to the interlayer characteristic interaction energy $126 \; {\rm meV}$). This provides a new argument supporting the validity of the BM model.
  
\begin{figure}[h]
\begin{center}
\includegraphics[width=0.2\textwidth]{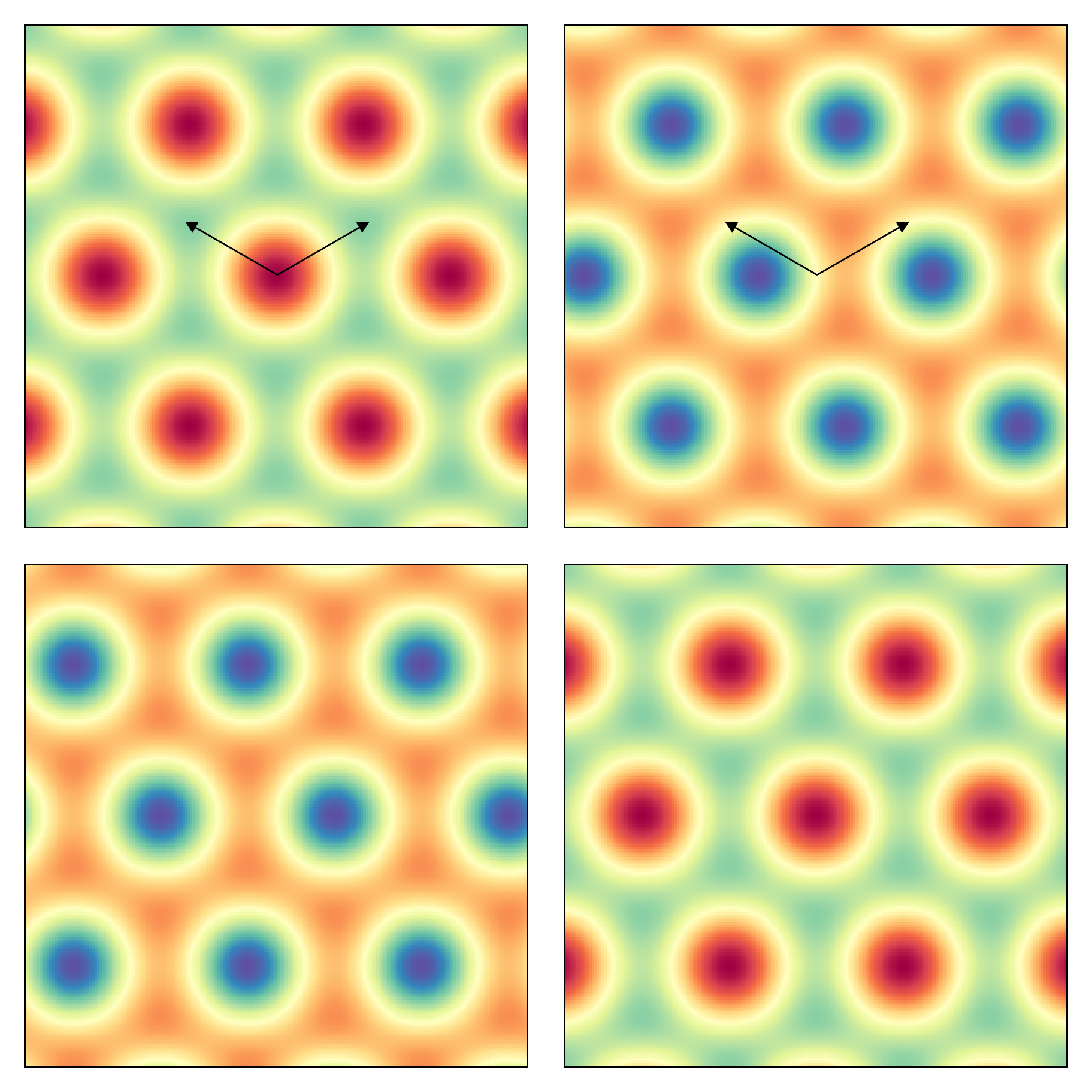}
\includegraphics[width=0.2\textwidth]{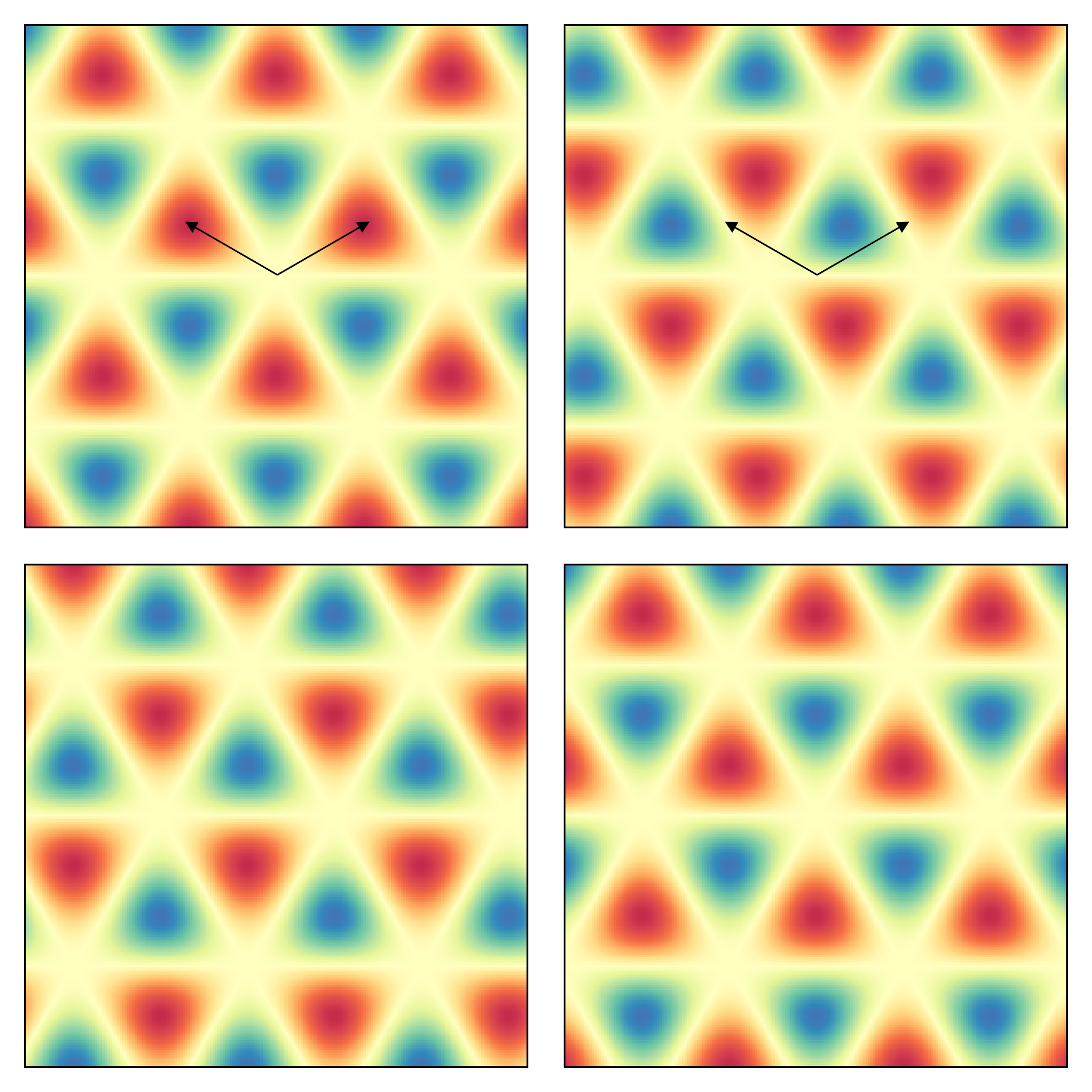}
\includegraphics[width=0.2\textwidth]{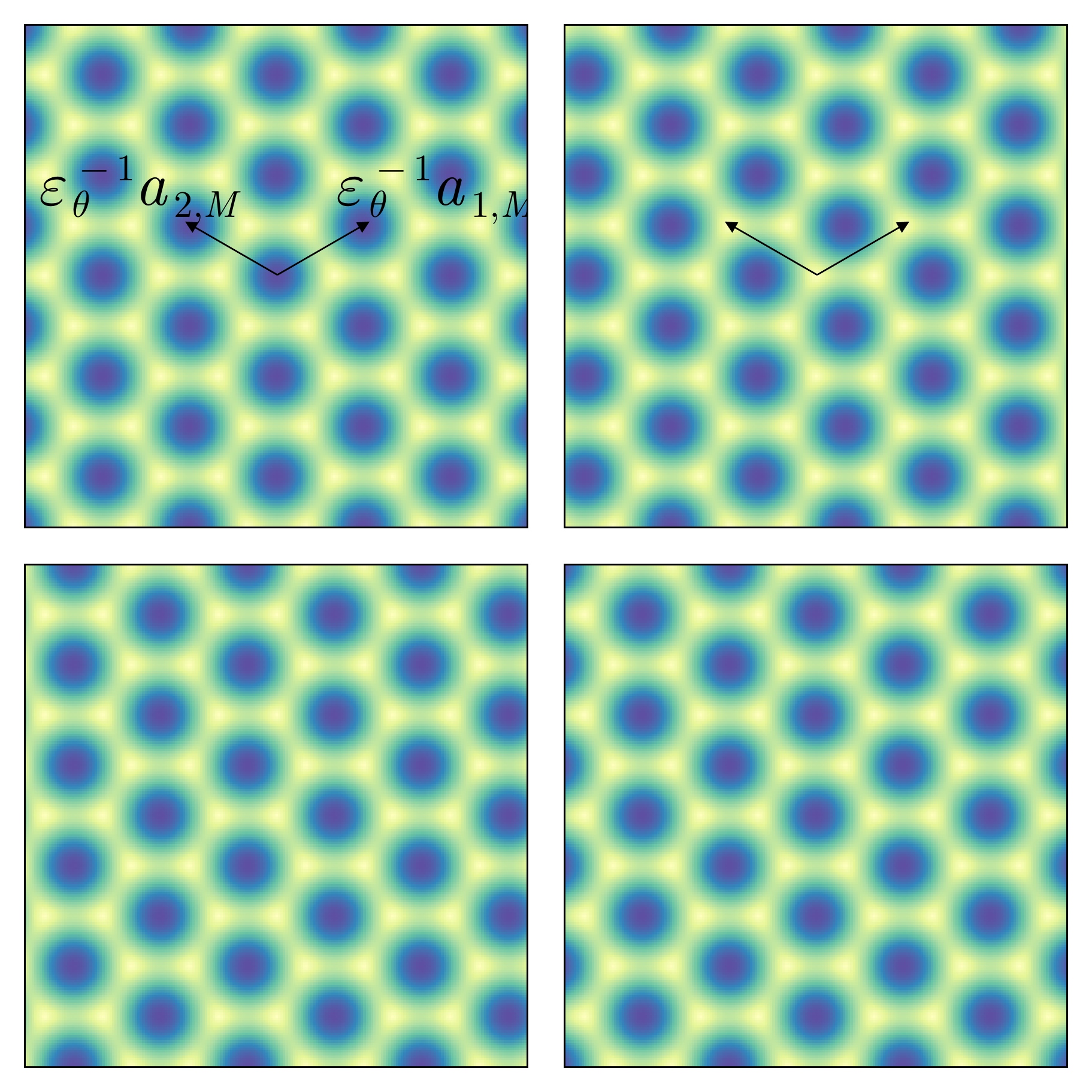}
\includegraphics[width=0.45\textwidth]{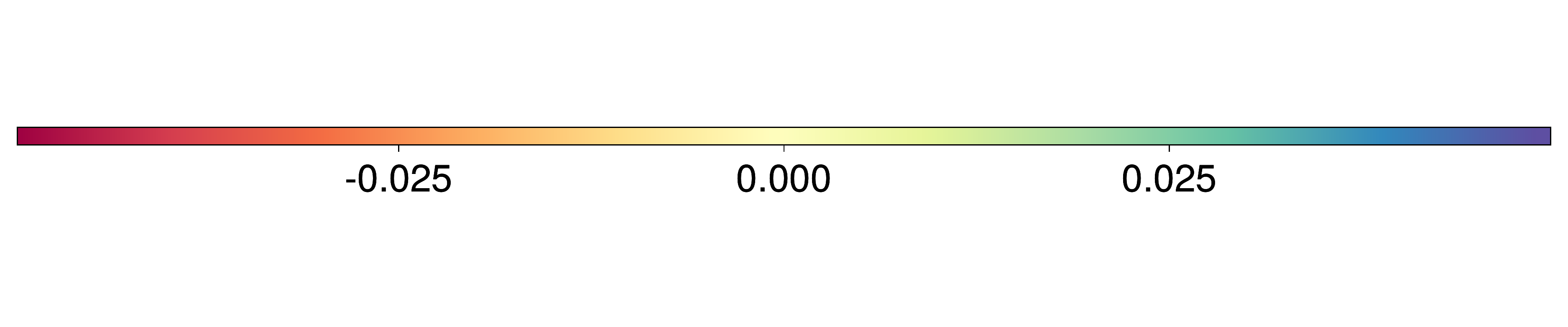}
\caption{The $4$ entries of resp. the real part, the imaginary part, and the modulus of the matrix-valued function $\Sigma_d(\bX)$ for $d = 6.45$ bohr.}\label{fig:Sigma}
\end{center}
\end{figure}

\begin{figure}[H]
\begin{center}
\includegraphics[height=4.9cm]{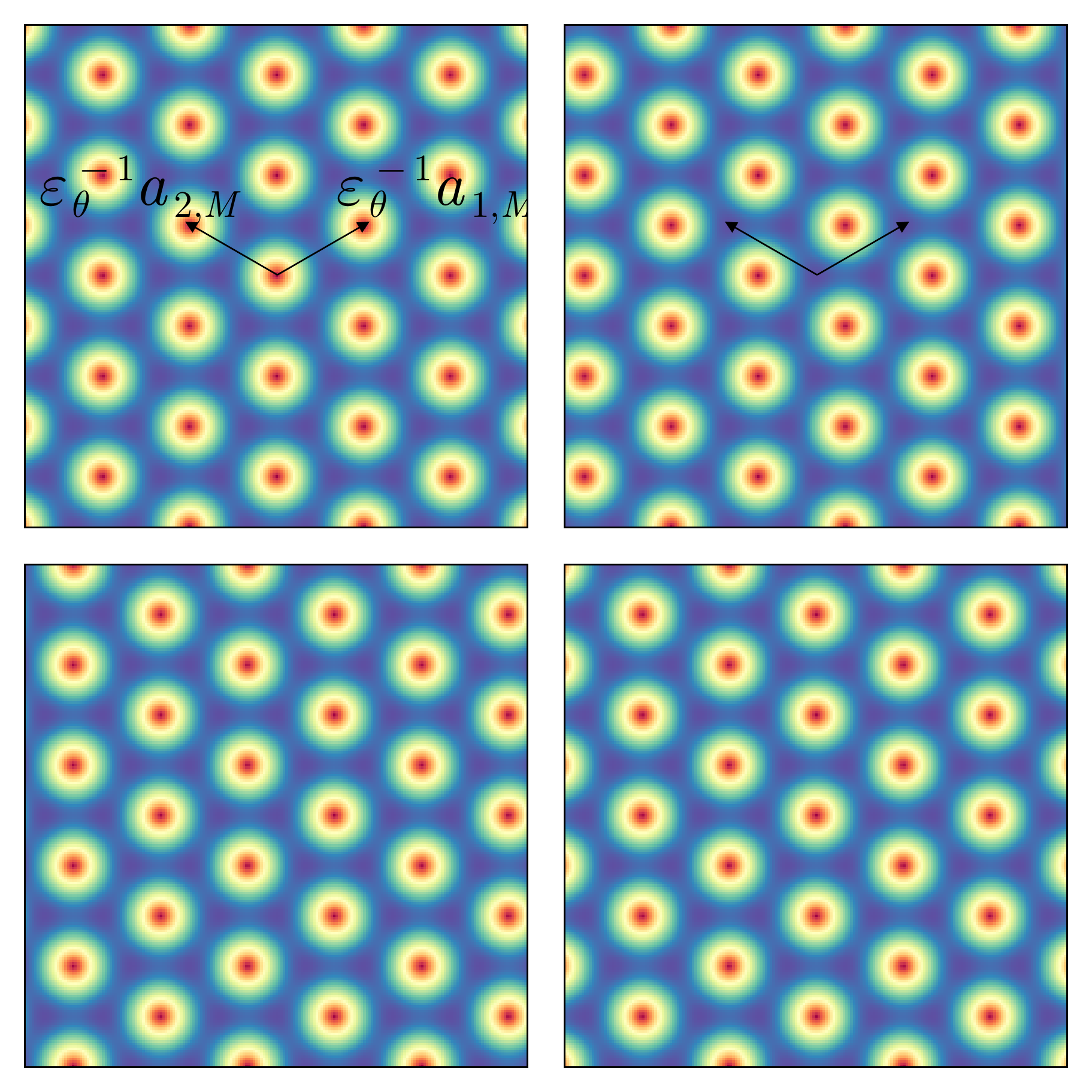}
\hspace{1cm}\includegraphics[height=4.5cm,trim={0.3cm -2cm 0.3cm 0cm},clip]{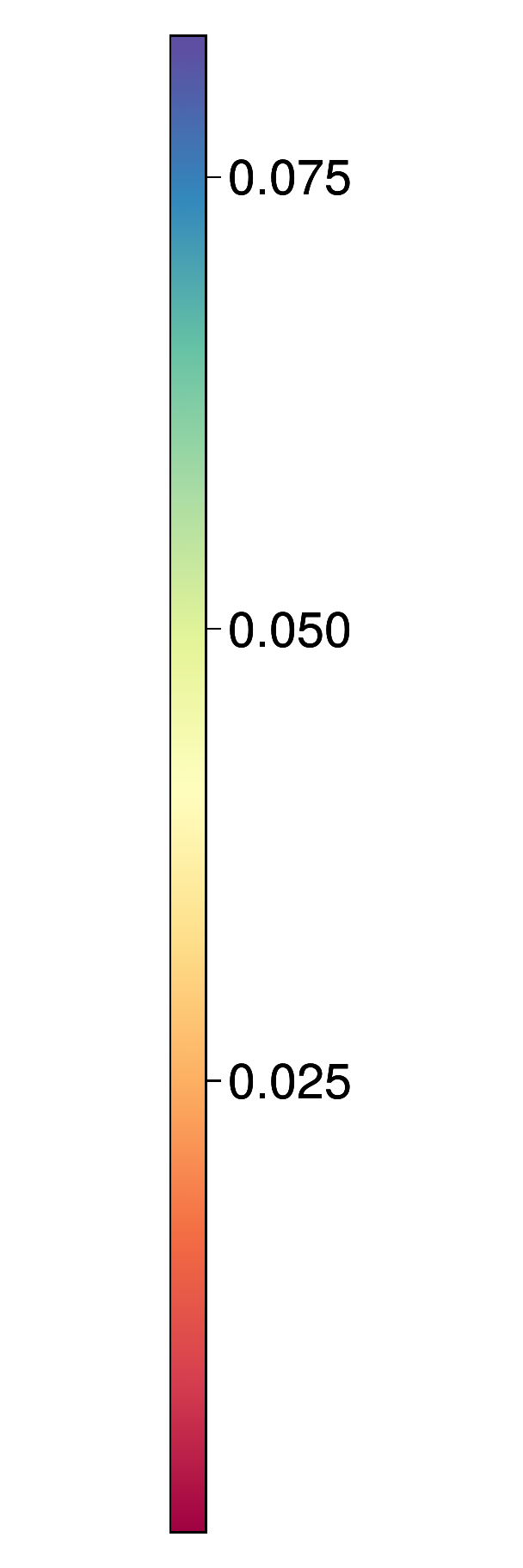}
\caption{The $4$ entries of the matrix-valued function $\ab{\nabla \Sigma_d(\bX)}/v\ind{F}$ for $d = 6.45$ bohr.}\label{fig:A}
\end{center}
\end{figure}

\begin{figure}[H]
\begin{center}
\includegraphics[width=0.2\textwidth]{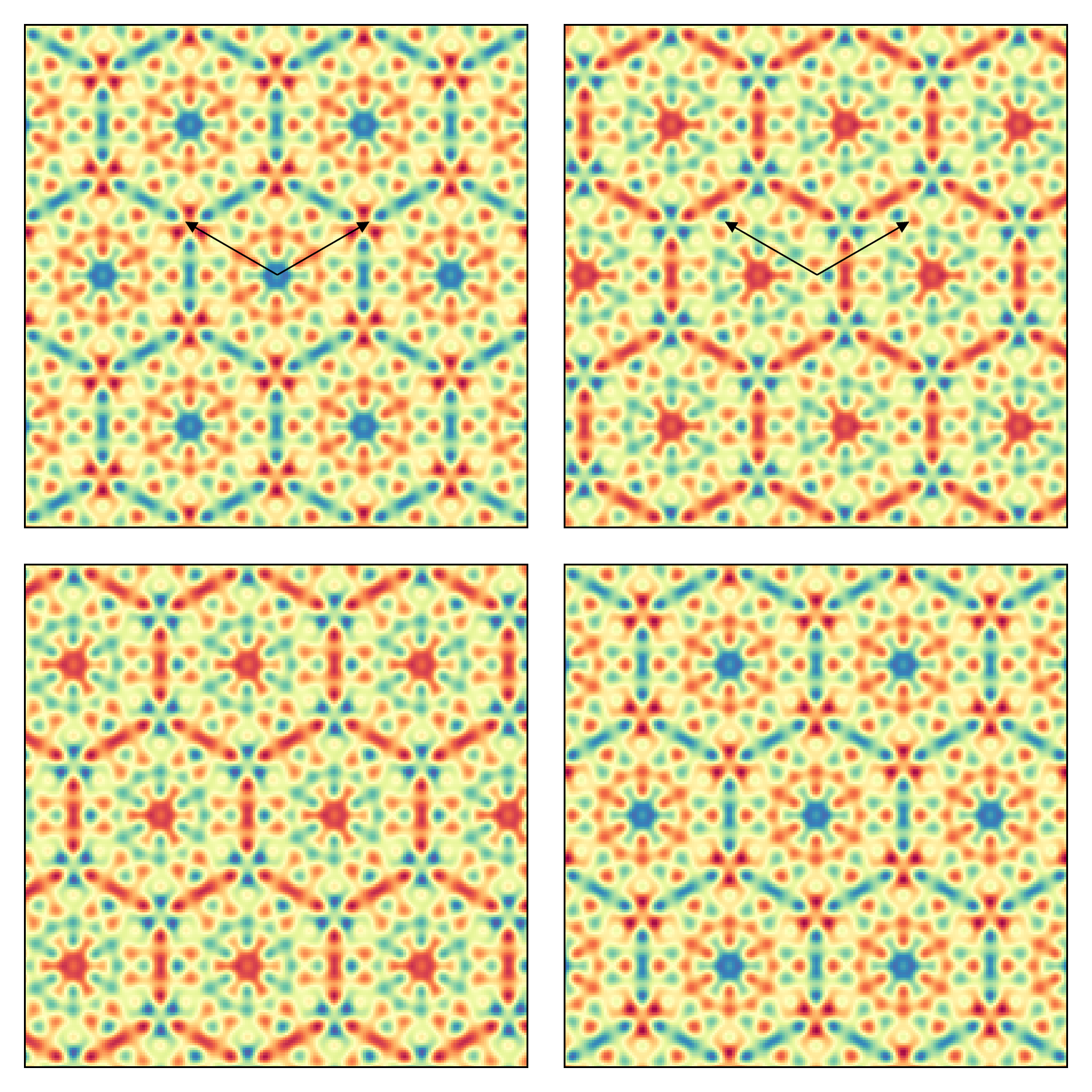}
\includegraphics[width=0.2\textwidth]{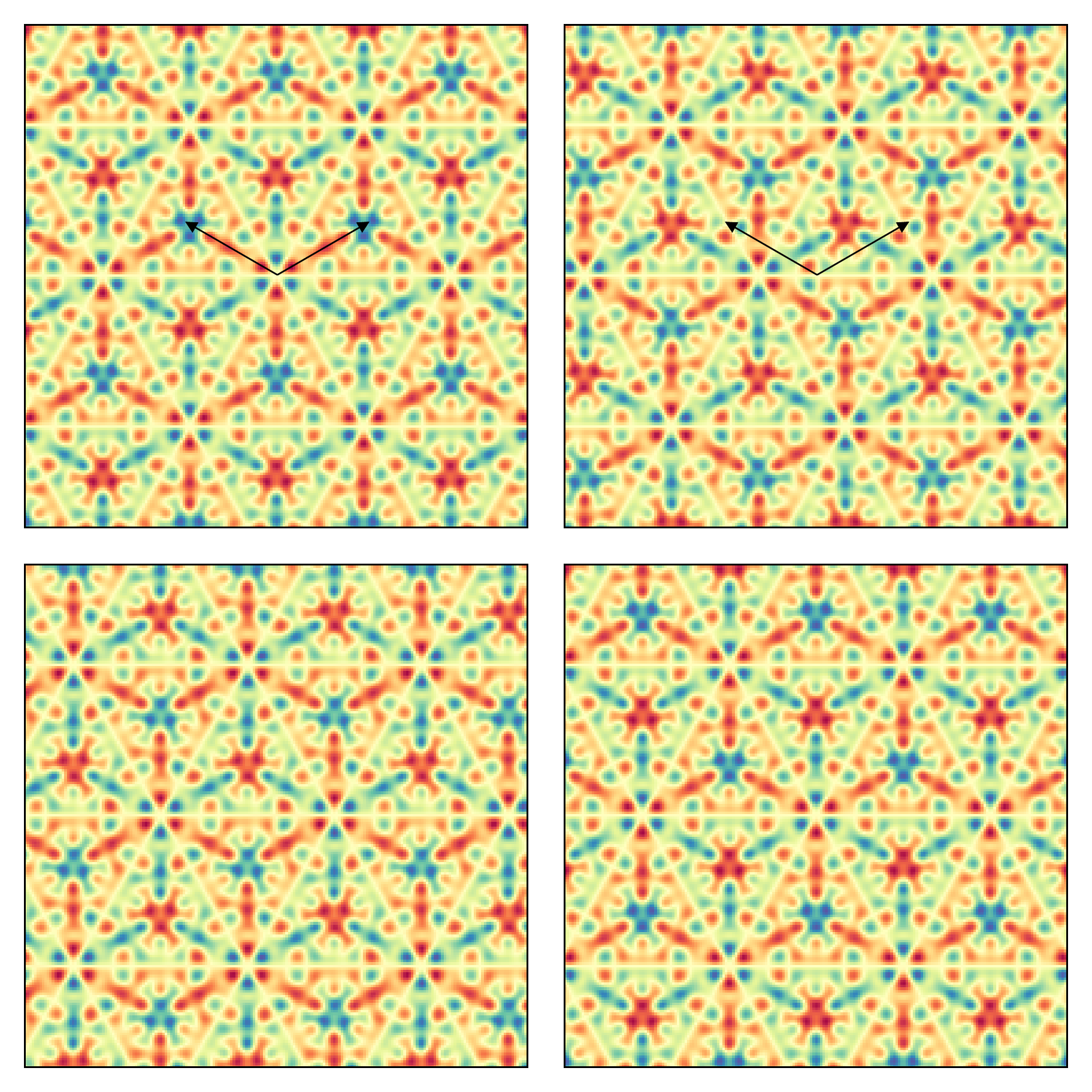}
\includegraphics[width=0.2\textwidth]{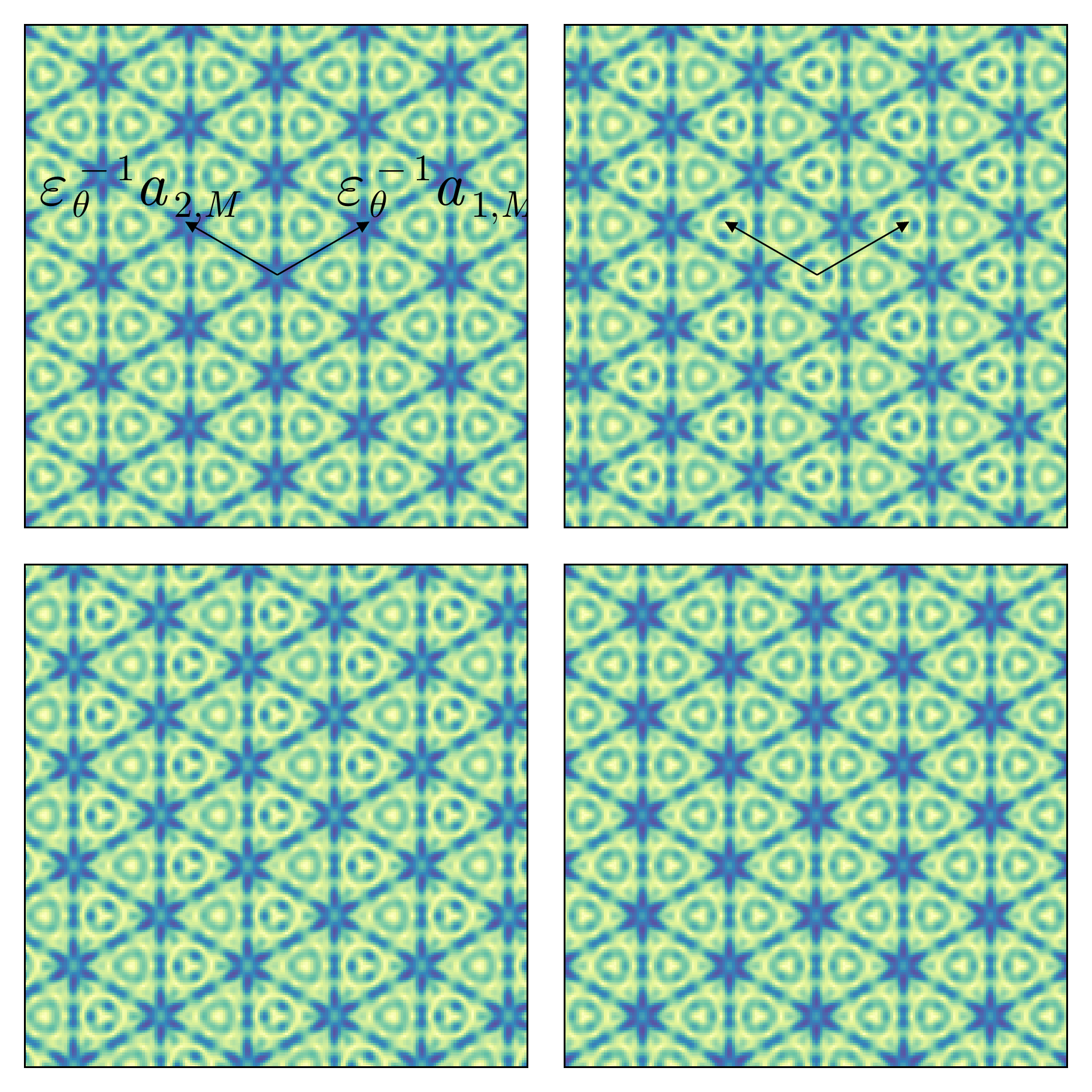}
\includegraphics[width=0.45 \textwidth]{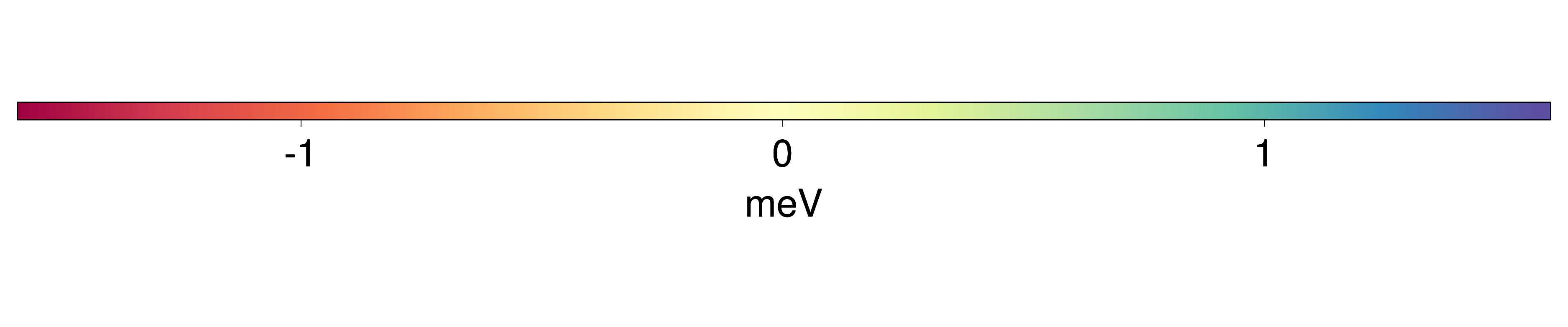}
\caption{The $4$ entries of resp. the real part, the imaginary part, and the modulus of the matrix-valued function $\bbV_d(\bX)-\tbm (\bX)$ for $d = 6.45$ bohr and $w_{\rm AA}= w_{\rm BB}=126 \; {\rm meV}$.}\label{fig:V_V}
\end{center}
\end{figure}

 \begin{figure}[H]
\begin{center}
\includegraphics[width=0.2\textwidth]{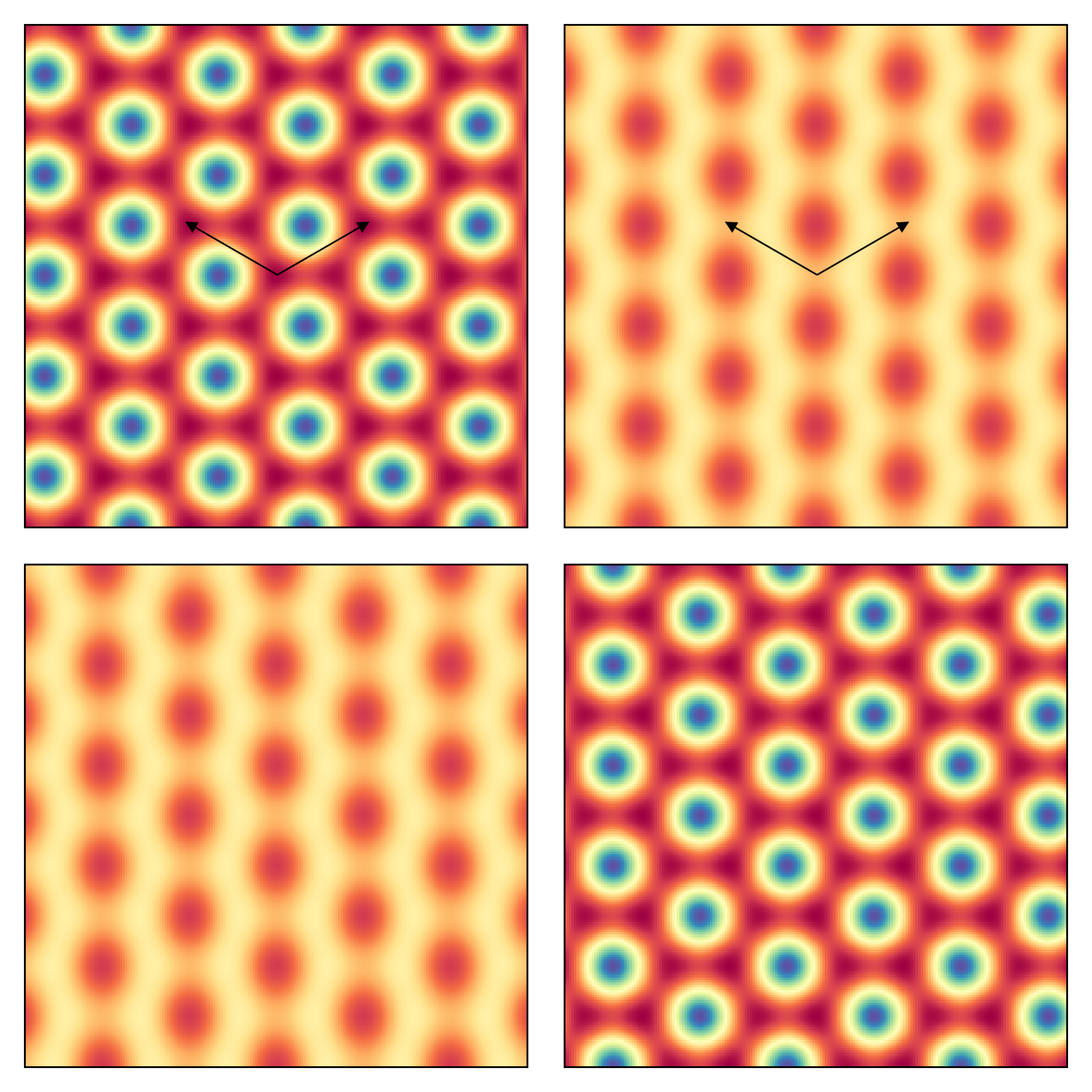}
\includegraphics[width=0.2\textwidth]{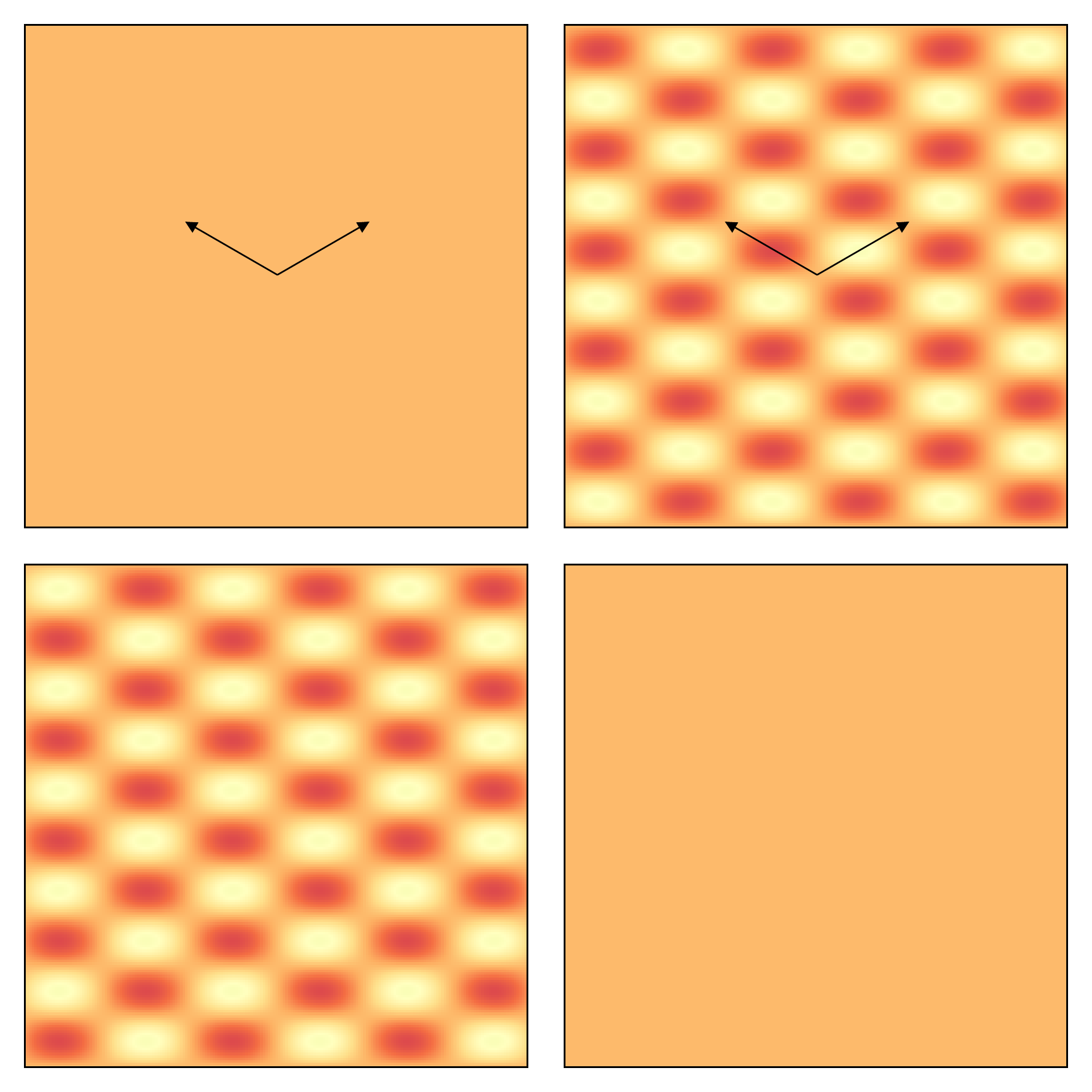}
\includegraphics[width=0.2\textwidth]{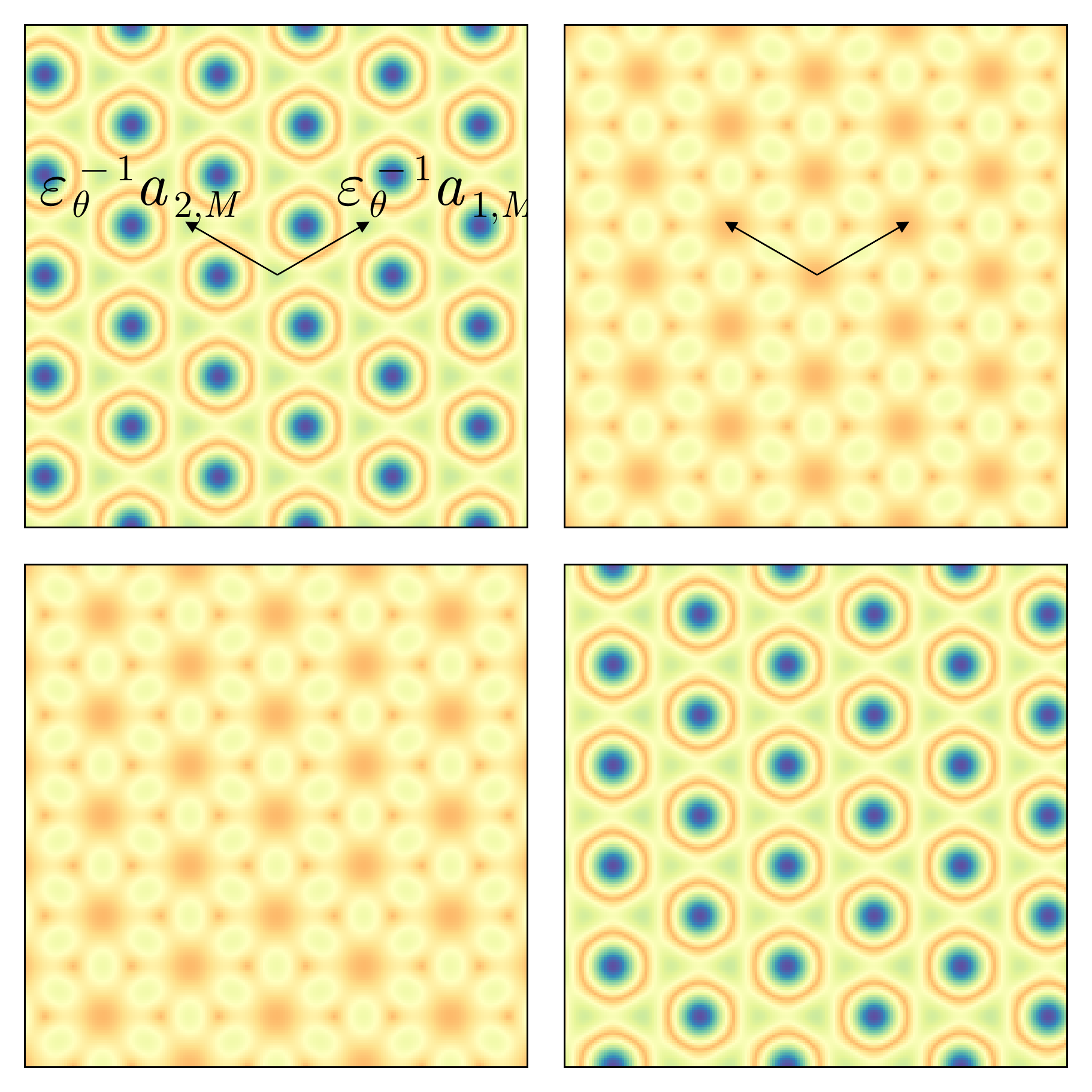}
\includegraphics[width=0.45\textwidth]{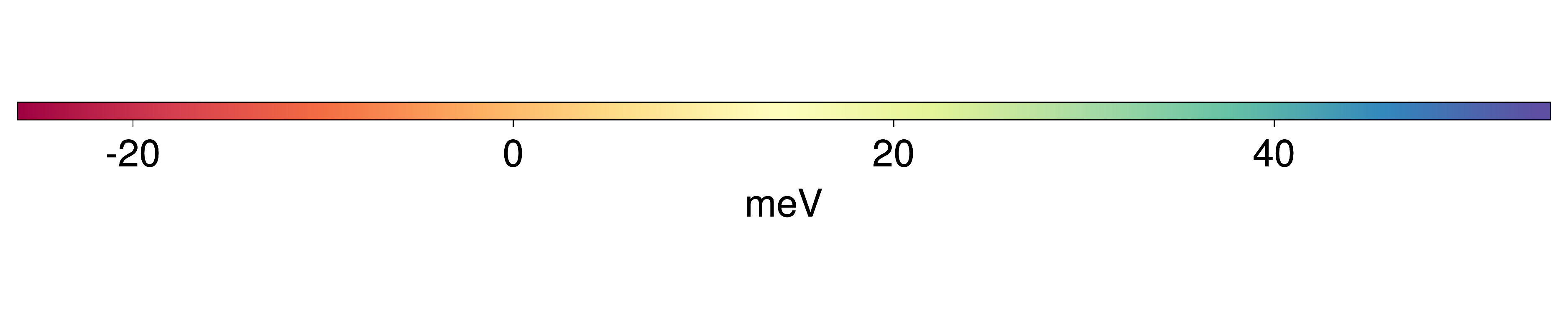}
\caption{The $4$ entries of resp. the real part, the imaginary part, and the modulus of the matrix-valued function $\bbW_d^+(\bX) - \fint_\Omega \bbW_d^+$ for $d = 6.45$ bohr.}\label{fig:W_plus_without_mean}
 \end{center}
  \end{figure}

\section{Conclusion and perspectives}

We have proposed a simple direct space construction of an effective model for TBG at the moiré scale from DFT based on 
\begin{itemize}
    \item an approximation of the TBG Kohn-Sham Hamiltonian following the method introduced in \cite{TSKCLPC16};
    \item a variational approximation of the so-obtained Hamiltonian at the Dirac point $\bK$ around the Fermi level;
    \item an asymptotic expansion valid for small twist angles $\theta$.
\end{itemize}
This effective model has a structure similar to the one of the Bistritzer-MacDonald model, but contains additional terms: a non-trivial overlap operator $\cS_d$, intralayer effective potentials $\bbW_d^\pm$, more complicated interlayer effective potentials $\bbV_d$, and higher-order corrections. We show numerically that both models give rise to similar band diagrams, the main difference being that at the first magic angle, the almost flat bands are separated from the rest of the spectrum in our model, which is not the case in the BM model.

It is well established that atomic relaxation plays a key role in the electronic properties of TBG (see e.g.~\cite{Nam17} and references therein), especially at very small twist angles $\theta < 1^\circ$. We derived our reduced model \eqref{eq:def:reduced_eqt} for an unrelaxed configuration, but it can be extended to take both in-plane and out-of-plane relaxation \cite{CarrMassatt18,Carr19,Cazeaux20} into account. The derivation and numerical simulation of such a model is work in progress.

The same methodology can be applied to study the propagation of low-energy wave-packets localized in momentum space in both the $\bK$ and $\bK'$ valleys. The approximation space then contains functions of the form
\begin{equation*}
      \left( \alpha : \Phi \right)_{d, \theta}(\bx, z) := \sum_{p \in \{ \bK, \bK'\}} \sum_{\eta \in \{ \pm 1\} \atop j \in \{ 1, 2\}} \alpha_{\eta, j, p} \left(\ept \bx \right) (\U^\eta \Phi_{j}^p)(\bx, z).
\end{equation*}
For the unrelaxed configuration, it can be checked that the two valleys are uncoupled. Intervalley coupling may appear however when taking atomic relaxations into account.


\section*{Acknowledgments}
This project has received funding from the European Research Council (ERC) under the European Union's Horizon 2020 research and innovation programme (grant agreement EMC2 No 810367) and from the Simons foundation (Targeted Grant Moir\'e Materials Magic). The authors thank Simon Becker, Antoine Levitt, Mitchell Luskin, Allan MacDonald, Etienne Polack and Alexander Watson for useful discussions and comments. Part of this work was done during the IPAM program {\it Advancing quantum mechanics with mathematics and statistics}.


\section*{Appendices}

\begin{appendix}

\section{Derivation of the effective model}
\label{sec:LOT}

Our goal is to identify the leading orders terms in \eqref{eq:proj_Schrodinger} in the limit of small twist angle. For that, we use the following elementary lemmas. Here $\cS$ is the Schwartz class of smooth function decaying faster than any polynomial.

\begin{lemma}\label{lem:Lemma1}
    Let $\beta \in {\cal S}(\R^2, \C)$ and $u \in L^1_{\rm per}(\Omega \times \R;\C)$. Then, as $\eps \to 0$, we have
    {\small 
    \[
        \int_{\R^3} \beta(\eps \bx) u(\bx, z) \, \rd \bx \, \rd z = \eps^{-2} \left( \dfrac{1}{ | \WS |} \int_{\WS \times \R} u  \right) \left( \int_{\R^2} \beta \right) + O(\eps^\infty).
    \]
}
\end{lemma}

\begin{proof}
    Expanding $u$ as $u(\bx, z) = \sum_{\bG \in \RLat} u_\bG(z) \re^{ \ri \bG \cdot \bx}$ and making the change of variable $\bX = \eps \bx$, we obtain that the left-hand side equals
    \begin{align*}
        & \eps^{-2} \sum_{\bG \in \RLat} \left( \int_{\R} u_\bG(z) \rd z \right) \int_{\R} \beta(\bX) \re^{ \ri \eps^{-1} \bG \cdot \bX} \, \rd \bX \\
        & \ =  \eps^{-2} \sum_{\bG \in \RLat} \left( \int_{\R} u_\bG(z) \, \rd z \right) \widehat{\beta} \left( \eps^{-1}\bG \right).
    \end{align*}
    As $\beta \in {\cal S}(\R^2;\C)$, $\widehat{\beta}$ decays faster than any polynomial. Isolating the term $\bG = \b0$ gives the result.
\end{proof}

We now denote by $L^2_\bK$ the set of locally square integrable functions which are $\bK$--quasiperiodic.
\begin{lemma} \label{lem:key_lemma_with_Angle}
    Let $\beta \in {\cal S}(\R^2, \C)$, $\Phi,\Phi' \in L^2_\bK(\Omega \times \R;\C)$, and $\eta, \eta' \in \{ \pm 1\}$. Then, as $\eps \to 0$, we have
    \begin{align*}
    & \int_{\R^3} \beta(\eps_\theta \bx) \left( \overline{\U^{\eta} \Phi} \cdot \U^{\eta'} \Phi' \right) (\bx, z) \, \rd \bx \, \rd z  \\
    & \quad = \frac{\ept^{-2}}{|\Omega|}  \int_{\R^2} \beta(\bX) \lAngle \Phi, \Phi' \rAngle_d^{\eta \eta'}(\bX) \rd \bX + O(\eps_\theta^\infty),
    \end{align*}
    where $\lAngle \Phi, \Phi' \rAngle_d^{\eta\eta'}(\bX)$ is defined by
    {\small
    \[
     \int_{\Omega \times \R} \overline{\Phi\left(\bx-\eta \tfrac12 J\bX,z-\eta \dd\right)} \, \Phi' \left(\bx-\eta' \tfrac12 J\bX,z-\eta' \dd \right)\, \d\bx \, \d z.
    \]
}
\end{lemma}

In the case $\eta = \eta'$, we have $\lAngle \Phi, \Phi' \rAngle_d^{\eta \eta'}(\bX) =  \langle \Phi, \Phi' \rangle_{L^2(\Omega \times \R)}$, independent of $\bX$. When $\eta \neq \eta'$, the function $\lAngle \Phi, \Phi' \rAngle_d^{\eta \eta'}(\bX)$ does depend on $\bX$ in general and is $J\L$-periodic.

Comparing $\lAngle \cdot, \cdot \rAngle_d$ and $\ppa{\cdot, \cdot}$ in~\eqref{eq:def:ppa}, we see that, with $\Phi = \re^{ \ri \bK \cdot \bx} u$,
\[
    \lAngle \Phi, \Phi' \rAngle_d^{+-}(\bX) = \re^{- \ri J \bK \cdot \bX} \ppa{u, u'}^{+-}(\bX),
\]
and
\[
     \lAngle \Phi, \Phi' \rAngle_d^{++} = \ppa{u, u'}^{++} = \langle \Phi, \Phi' \rangle = \langle u, u' \rangle.
\]
In what follows, we express our quantities with $\lAngle \Phi, \Phi' \rAngle_d$, but we translated our results with $\ppa{u, u'}$ to present our reduced model.

\begin{proof}[Proof of Lemma~\ref{lem:key_lemma_with_Angle}]
    In the case $\eta = \eta'$, the left-hand side equals
    \begin{align*}
     &
    \int_{\R^3} \beta(\eps_\theta \bx) \left( \overline{\Phi}\Phi' \right)\left(R_{\eta\frac\theta 2}\bx,z-\eta\frac d2\right) \, \rd \bx \, \rd z \\
    & \ = \ept^{-2} \int_{\R^3} \beta\left(\bX\right) \left( \overline{\Phi}\Phi' \right)\left(\ept ^{-1}R_{\eta\frac\theta 2}\bX,z\right) \, \rd \bX \, \rd z .
    \end{align*}
    As the function $\overline{\Phi}\Phi'$ is $\L$-periodic, the result can be obtained by applying the same arguments as in the proof of Lemma~\ref{lem:Lemma1}.
    
    Let us now focus on the case when $\eta \neq \eta'$, and prove the result for $\eta' = +1$ and $\eta = -1$, the other case being similar. Let $u(\br) := \re^{ - \ri \bK \cdot \bx} \Phi(\br)$ and $u'(\br) := \re^{ - \ri \bK \cdot \bx} \Phi'(\bx)$ be the periodic components of the Bloch waves $\Phi$ and $\Phi'$ respectively. We have
    \begin{align*}
        & \left( \overline{\U \Phi} \cdot \U^{-1} \Phi'  \right) (\bx, z) \\
        & \ = 
        \re^{ \ri \bK  \ept \cdot J \bx} \overline{u(c_\theta \bx - \tfrac12 \varepsilon_\theta J \bx, z - \dd)}
        u'(c_\theta \bx + \tfrac12 \varepsilon_\theta J \bx, z + \dd).
    \end{align*}
    Introducing the Fourier expansions of $u$ and $u'$, this is also
    \begin{align*}
        & \dfrac{1}{| \Omega |} \sum_{\bG, \bG' \in \RLat}  \left[ \overline{u_{\bG}}(z - \dd)  u'_{\bG'}(z + \dd)  e^{\ri (\bG' - \bG) \cdot c_\theta \bx} \right] \\
        & \quad \times e^{\ri \tfrac12 (\bG+ \bG' + 2 \bK) \ept \cdot J \bx}.
    \end{align*}
    Note that the last phase factor varies at the moiré scale. The term in brackets is a $c_\theta^{-1} \Lat$-periodic function with zero mean unless $\bG = \bG'$. Reasoning as above, we obtain
    \begin{widetext}
    \begin{align*}
        \int_{\R^3} \beta(\eps_\theta \bx) \left( \overline{\U^{\eta} \Phi} \cdot \U^{\eta'} \Phi' \right) (\bx, z) \, \rd \bx \, \rd z
        = 
         \frac{\ept^{-2}}{|\Omega|} \int_{\R^2} \beta(\bX)  \left[ \re^{ \ri \bK \cdot J \bX}  \sum_{\bG \in \RLat} \left( \int_{\R} \overline{u_\bG}(z - \dd) u'_{\bG}(z + \dd) \rd z \right) \re^{ \ri \bG \cdot J \bX} \right] \, \rd \bX + O(\ept^\infty).
    \end{align*}
    \end{widetext}
    Finally, by Parseval theorem, the term in brackets is equal to $\lAngle \Phi, \Phi' \rAngle_d^{\eta \eta'}(\bX)$, which proves the result.
\end{proof}

\subsection{Effective overlap operator}

Let us first focus on the left-hand side of \eqref{eq:proj_Schrodinger}. Setting $\Phi_{\eta j}:=\U^\eta\Phi_j$, we have  (all quantities are summed over $\eta, \eta' \in \{ \pm 1\}$ and $j, j' \in \{ 1, 2\}$)
\begin{align*}
    & \left\langle \left( \widetilde{\alpha} : \Phi \right)_{d, \theta}, \left( \alpha(\ept t) : \Phi \right)_{d, \theta} \right\rangle 
      \\
     &  \ = \int_{\R^3} \left( \overline{\widetilde{\alpha}_{\eta j}} \alpha_{\eta' j'}(t) \right)(\ept \bx)  \times \left( \overline{ \Phi_{\eta' j'}}  \Phi_{\eta j} \right)(\bx, z) \, \rd \bx \, \rd z.
\end{align*}
Using directly Lemma~\ref{lem:key_lemma_with_Angle} (with $\beta = \overline{\widetilde{\alpha}_{\eta j}} \alpha_{\eta' j'}$), we obtain that this term equals
\[
    \ept^{-2}
       \int_{\R^2} \left( \overline{\widetilde{\alpha}_{\eta j}} \alpha_{\eta' j'}(\ept t) \right)(\bX) \lAngle \Phi_{j}, \Phi_{j'} \rAngle^{\eta \eta'}_d(\bX) \rd \bX
\]
up to errors of order $O(\eps^\infty)$. Ranking the components of $\alpha$ as $\alpha = (\alpha_{+, 1}, \alpha_{+, 2}, \alpha_{-, 1}, \alpha_{-, 2})^T$ (the first two entries correspond to the top layer, the last two to the bottom one), we obtain
\[
 \partial_t \left\langle \left( \widetilde{\alpha} : \Phi \right)_{d, \theta}, \left( \alpha(\ept t) : \Phi \right)_{d, \theta} \right\rangle = \frac{\ept^{-1}}{|\Omega|} \left\langle \widetilde{\alpha}, \cS_d \partial_\tau\alpha(\ept t) \right\rangle,
\]
up to errors of order $O(\ept^\infty)$, with
$$
    \cS_d =  \begin{pmatrix}
        \bbI_2 & \Sigma_d(\bX) \\ \Sigma_d^*(\bX) & \bbI_2
    \end{pmatrix}, \  \left[ \Sigma_d\right]_{jj'}(\bX) := \lAngle \Phi_{j}, \Phi_{j'} \rAngle^{+-}_d(\bX).
$$

\subsection{Effective Hamiltonian operator}
We now focus on the terms on the right-hand side. First, we record that 
\begin{align}
    &  (- \tfrac12 \Delta_\br) \left[  \beta (\ept \bx) \varphi(\br) \right] = \nonumber \\
    & \ = \ept^2 (- \tfrac12 \Delta \beta)(\ept \bx) \varphi(\br) + \ept ( - \ri \nabla \beta) (\ept \bx) \cdot (- \ri \nabla_\bx \varphi) (\br) \nonumber \\
    & \quad + \beta (\ept \bx) (- \tfrac12 \Delta \varphi)(\br). \label{eq:split_laplacian}
\end{align}
This gives
\[
  \left\langle \left( \widetilde{\alpha} : \Phi \right)_{d, \theta}, \left( H_{d, \theta}^{(2)} - \mu_F \right) \left( {\alpha} : \Phi \right)_{d, \theta} \right\rangle = g_1 + g_2 + g_3,
\]
with
\[
    g_1  :=  \ept^2 \int_{\R^3} \left[ \overline{\widetilde{\alpha}_{\eta, j}} (-\tfrac12 \Delta \alpha_{\eta', j'}) \right](\ept \bx) ( \overline{\Phi}_{\eta, j} \Phi_{\eta', j'}) (\br) \rd \br,
\]
\[
    g_2 := \ept \int_{\R^3}  \left[ \overline{\widetilde{\alpha}_{\eta, j}} (- \ri \nabla \alpha_{\eta', j'})\right] (\ept \bx) \cdot \left[ \overline{\Phi}_{\eta, j} ( - \ri \nabla_\bx \Phi_{\eta', j'})\right] (\br)  \rd \br,
\]
\[
    g_3 := \int_{\R^3} \left[ \overline{\widetilde{\alpha}_{\eta, j}} \alpha_{\eta', j'}(\ept \bx) \right] \cdot \left[ \overline{\Phi}_{\eta, j} \times (H_{d, \theta}^{(2)} - \mu_F) \Phi_{\eta', j'}) \right] (\br)  \rd \br.
\]
For the first term $g_1$, we apply directly Lemma~\ref{lem:key_lemma_with_Angle}, which gives
\begin{align*}
     g_1 =  \frac{1}{|\Omega|} \int_{\R^2} \left( \overline{\widetilde{\alpha}_{\eta j}} (- \tfrac12 \Delta  \alpha_{\eta' j'}) \right)(\bX) \lAngle \Phi_{j}, \Phi_{j'} \rAngle^{\eta \eta'}_d(\bX) \rd \bX
\end{align*}
up to errors of order $O(\ept^\infty)$. This term can be written as $g_1 = |\Omega|^{-1}  \left\langle \widetilde{\alpha}, G_1 \alpha \right\rangle + O(\ept^\infty)$ with
\[
    G_1 := \cS_d(\bX) ( - \tfrac12 \Delta).
\]
For the second term $g_2$, we first notice that
\begin{align*}
        (- \ri \nabla_\bx \Phi_{\eta', j'})(\br)  & = ( - \ri \nabla_\bx) \left[ \Phi_{j'}(R_{-\theta/2}^{-\eta'} \bx, z - \eta' \dd)   \right] \\
         & = R_{-\theta/2}^{\eta'} \left[\U^{\eta'} \left(  - \ri \nabla_\bx \Phi_{j'}  \right)\right](\br).
\end{align*}
This gives, using arguments similar to the ones of Lemma~\ref{lem:key_lemma_with_Angle}, that $g_2$ equals
{\small
\[
     \frac{ \ept^{-1}}{|\Omega|}  \int_{\R^2} \left( \overline{\widetilde{\alpha}_{\eta j}} (- \ri \nabla \alpha_{\eta' j'}) \right)(\bX) \cdot  R_{-\theta/2}^{\eta'} \lAngle \Phi_{j}, (- \ri \nabla_\bx \Phi_{j'}) \rAngle^{\eta \eta'}_d(\bX) \rd \bX
\]}
up to errors of order $O(\eps^\infty)$.
To deal with the diagonal terms $\eta = \eta'$, we use our orientation \eqref{eq:pre_Dirac}. Regarding the off-diagonal terms (here for $\eta = +1$ and $\eta' = -1$), we have, using that $J^T = -J$,
\begin{widetext}
    \begin{align*}
        \nabla \lAngle \Phi, \Phi' \rAngle^{+-}_d(\bX) 
        & = \tfrac12 J \int_{\Omega \times \R}  \overline{\left( \nabla_\bx \Phi \right) \left(\bx- \tfrac12 J\bX,z- \dd\right)} \, \Phi' \left(\bx+ \tfrac12 J\bX,z + \dd \right)\, \d\bx \, \d z \\
        & \quad - \tfrac12 J \int_{\Omega \times \R}  \overline{\Phi \left(\bx- \tfrac12 J\bX,z- \dd\right)} \,  \left( \nabla_\bx \Phi' \right) \left(\bx+ \tfrac12 J\bX,z + \dd \right)\, \d\bx \, \d z.
    \end{align*}
\end{widetext}

Integrating by part the first term in the RHS, and multiplying by $(-i)$ shows that
\[
    \lAngle \Phi, ( - \ri \nabla_\bx) \Phi' \rAngle^{+-}_d =  J (- \ri \nabla) \lAngle \Phi, \Phi' \rAngle^{+-}_d.
\]
Similarly, in the case $\eta = -1$ and $\eta' = 1$, we have
\[
   \lAngle \Phi, ( - \ri \nabla_\bx) \Phi' \rAngle^{-+}_d = - J (- \ri \nabla) \lAngle \Phi, \Phi' \rAngle^{-+}_d.
\]
This gives $g_2 = |\Omega|^{-1} \langle \widetilde{\alpha}, G_2 \alpha \rangle + O(\ept^\infty)$ with $G_2$ of the form (we use that $R_{-\theta/2} = c_\theta \bbI_2 + \frac12 \ept J$)
\begin{align*}
   & \dfrac{c_\theta}{\ept}\begin{pmatrix}
    v_F \bsigma \cdot (- \ri \nabla) &   J (- \ri \nabla \Sigma_d)(\bX) \cdot (- \ri \nabla)  \\
      J (- \ri \nabla \Sigma_d^*)(\bX) \cdot (- \ri \nabla)  & v_F \bsigma \cdot  (- \ri \nabla)
\end{pmatrix} \\
& \qquad + \frac12 \begin{pmatrix}
    v_F \bsigma \cdot \left[ J (- \ri \nabla) \right] &  (- \ri \nabla \Sigma_d)(\bX)\cdot  (- \ri \nabla) \\
     (- \ri \nabla \Sigma_d^*)(\bX)\cdot (- \ri \nabla)  & -v_F \bsigma \cdot \left[J (- \ri \nabla) \right]
\end{pmatrix}.
\end{align*}
Finally, for the term $g_3$, we recall that $\Phi_j$ is an eigenvector of the single-layer graphene Hamiltonian $H^{(1)}$ associated with the eigenvalue $\mu_{\rm F}$ and get
\[
    \left[ (H_{d, \theta}^{(2)} - \mu_F) \Phi_{\eta, j} \right] = \left( \U^{-\eta} V\right)(\br) \Phi_{\eta, j}(\br) + V_{{\rm int}, d}(z)\Phi_{\eta, j}(\br),
\]
and
\begin{align*} 
    g_3  = & \int_{\R^3} \left[ (\overline{\widetilde{\alpha}_{\eta', j'}} \alpha_{\eta, j})(\ept \bx) \right] \times \\
    & \qquad \left[ \left( \U^{-\eta} V\right)  \overline{\Phi}_{\eta', j'} \Phi_{\eta, j} 
     + V_{{\rm int}, d} \overline{\Phi}_{\eta', j'} \Phi_{\eta, j} \right] (\br)  \rd \br.
\end{align*}
Using reasoning similar to the proof of Lemma~\ref{lem:key_lemma_with_Angle}, we obtain that $g_3 = \langle \widetilde{\alpha}, G_3 \alpha \rangle + O(\ept^\infty)$, with
\[
    G_3 =  \ept^{-2} \begin{pmatrix}
        \bbW_d^+(\bX) & \bbV_d(\bX) \\
        \bbV_d(\bX)^* & \bbW_d^-(\bX)
    \end{pmatrix},
\]
where (recall that the notation~$\ppa{f,g}$ was defined in~\eqref{eq:def:ppa}, and is used when $f$ and $g$ are periodic.)
\begin{widetext}
    \begin{align*}
        \left[ \bbV_d(\bX) \right]_{jj'} 
        &= \lAngle V \Phi_{j}, \Phi_{j'} \rAngle_d^{+-}(\bX)
        + \int_{\Omega \times \R} \overline{\Phi_{j} \left(\bx- \tfrac12 J\bX,z- \dd\right)} \, \Phi_{j'} \left(\bx + \tfrac12 J\bX,z + \dd \right) V_{{\rm int}, d}(z) \, \d \br, \\
        \left[ \bbW_d^\pm(\bX) \right]_{j j'} 
        &= (\!(\Phi_{j} \overline{\Phi_{j'}}, V )\!)_d^{\pm \mp}(\bX) +  \int_{\Omega \times \R} (\overline{\Phi_{j}}\Phi_{j'})(\bx, z \mp \dd) V_{{\rm int}, d}(z) \rd \br. 
    \end{align*}
\end{widetext}
The second term of $\bbW^\pm_d$ is a constant matrix (independent of $\bX$). We prove in the Supplementary Material that this term is proportional to $\bbI_2$. Finally, since $V_{{\rm int}, d}(-z)= V_{{\rm int}, d}(z)$, this matrix is the same for the $\bbW^+_d$ and the $\bbW^-_d$ terms.

\medskip

To conclude and obtain the expression in \eqref{eq:ourH}, we have used the equality
\[
    \Sigma_d ( - \tfrac12 \Delta) + \tfrac12 ( - \ri \nabla \Sigma_d) \cdot ( - \ri \nabla) =  - \frac12 {\rm div} \left( \Sigma_d(\bX) \nabla \bullet \right).
\]

\section{Proof that $w_{\rm AA}^d = w_{\rm AB}^d$}
\label{app:wAA}

Recall that $w_{\rm AA}^d$ and $w_{\rm AB}^d$ are defined resp. in~\eqref{eq:value_wAA}-\eqref{eq:value_wAB}. We prove in the Supplementary Material that $[\bbV_d]_{11}$ satisfies $[\bbV_d]_{11}(R_{\frac{2 \pi}{3}} \bX) = [\bbV_d]_{11}(\bX)$, where $R_{\frac{2 \pi}{3}}$ is the $2 \pi/3$ rotation, see~\eqref{eq:rotationSV1}. Recalling the definition of $G(\bX)$ in~\eqref{eq:defFG} and using that $\bq_3 = R_{\frac{2 \pi}{3}} \bq_1$ while $\bq_2 = (R_{\frac{2 \pi}{3}})^2 \bq_1$, we obtain
\begin{align*}
    w_{\rm AA}^d &= \frac{1}{3|\Omega_{\rm M}|} \left( \sum_{n=1}^3 \int_{\Omega_M}  [\bbV_d]_{11}(\bX) \re^{\ri \bq_n \cdot \bX} \, \rd \bX \right) \\
    &= \frac{1}{|\Omega_{\rm M}|} \left( \int_{\Omega_M}  [\bbV_d]_{11}\left(\bX\right) \re^{\ri \bq_1 \cdot \bX} \, \rd \bX \right).
\end{align*}
From the definition of $\bbV$ and the fact that $\bq_1 = J \bK$, while $| \Omega_M | = |\Omega |$, we get 
\begin{align*}
    w_{\rm AA}^d = & \int_\R \left( \int_\Omega \overline{\left[\pa{V+V_{\text{int},d}(\cdot+ \dd)} u_1\right](\bx,z-\dd)} \, \rd \bx \right) \times \\
    & \qquad  \left( \fint_\Omega u_1(\bx,z+\dd) \, \rd \bx \right) \, \rd z.
\end{align*}
A similar calculation leads to
\begin{align*}
    w_{\rm AB}^d = & \int_\R \left( \int_\Omega \overline{\left[\pa{V+V_{\text{int},d}(\cdot+ \dd)} u_2\right](\bx,z-\dd)} \, \rd \bx \right) \times \\
    & \quad  \left( \fint_\Omega u_1(\bx,z+\dd) \, \rd \bx \right) \, \rd z.
\end{align*}
In the Supplementary Materials, we prove that $\overline{u_j(x_1,x_2,z)}=-u_j(-x_1,x_2,z)$ (see~\eqref{eq:symmetries_uj}). Since $V$ and $V_{\rm int}$ are real-valued, with $V(x_1,x_2,z)=V(-x_1,x_2,z)$, the parameters $w_{\rm AA}^d$ and $w_{\rm AB}^d$ are real valued. In addition, we also proved that $u_2(x_1, - x_2, z) = u_1(x_1, x_2, z)$. Together with the fact that $V(x_1,-x_2,z)=V(x_1,x_2,z)$, we deduce $w_{\rm AA}^d=w_{\rm AB}^d$.

\end{appendix}

\bibliography{moire}

\begin{thebibliography}{30}%
\makeatletter
\providecommand \@ifxundefined [1]{%
 \@ifx{#1\undefined}
}%
\providecommand \@ifnum [1]{%
 \ifnum #1\expandafter \@firstoftwo
 \else \expandafter \@secondoftwo
 \fi
}%
\providecommand \@ifx [1]{%
 \ifx #1\expandafter \@firstoftwo
 \else \expandafter \@secondoftwo
 \fi
}%
\providecommand \natexlab [1]{#1}%
\providecommand \enquote  [1]{``#1''}%
\providecommand \bibnamefont  [1]{#1}%
\providecommand \bibfnamefont [1]{#1}%
\providecommand \citenamefont [1]{#1}%
\providecommand \href@noop [0]{\@secondoftwo}%
\providecommand \href [0]{\begingroup \@sanitize@url \@href}%
\providecommand \@href[1]{\@@startlink{#1}\@@href}%
\providecommand \@@href[1]{\endgroup#1\@@endlink}%
\providecommand \@sanitize@url [0]{\catcode `\\12\catcode `\$12\catcode
  `\&12\catcode `\#12\catcode `\^12\catcode `\_12\catcode `\%12\relax}%
\providecommand \@@startlink[1]{}%
\providecommand \@@endlink[0]{}%
\providecommand \url  [0]{\begingroup\@sanitize@url \@url }%
\providecommand \@url [1]{\endgroup\@href {#1}{\urlprefix }}%
\providecommand \urlprefix  [0]{URL }%
\providecommand \Eprint [0]{\href }%
\providecommand \doibase [0]{https://doi.org/}%
\providecommand \selectlanguage [0]{\@gobble}%
\providecommand \bibinfo  [0]{\@secondoftwo}%
\providecommand \bibfield  [0]{\@secondoftwo}%
\providecommand \translation [1]{[#1]}%
\providecommand \BibitemOpen [0]{}%
\providecommand \bibitemStop [0]{}%
\providecommand \bibitemNoStop [0]{.\EOS\space}%
\providecommand \EOS [0]{\spacefactor3000\relax}%
\providecommand \BibitemShut  [1]{\csname bibitem#1\endcsname}%
\let\auto@bib@innerbib\@empty
\bibitem [{\citenamefont {{Andrei}}\ \emph {et~al.}(2021)\citenamefont
  {{Andrei}}, \citenamefont {{Efetov}}, \citenamefont {{Jarillo-Herrero}},
  \citenamefont {{MacDonald}}, \citenamefont {{Mak}}, \citenamefont
  {{Senthil}}, \citenamefont {{Tutuc}}, \citenamefont {{Yazdani}},\ and\
  \citenamefont {{Young}}}]{Andrei21}%
  \BibitemOpen
  \bibfield  {author} {\bibinfo {author} {\bibfnamefont {E.~Y.}\ \bibnamefont
  {{Andrei}}}, \bibinfo {author} {\bibfnamefont {D.~K.}\ \bibnamefont
  {{Efetov}}}, \bibinfo {author} {\bibfnamefont {P.}~\bibnamefont
  {{Jarillo-Herrero}}}, \bibinfo {author} {\bibfnamefont {A.~H.}\ \bibnamefont
  {{MacDonald}}}, \bibinfo {author} {\bibfnamefont {K.~F.}\ \bibnamefont
  {{Mak}}}, \bibinfo {author} {\bibfnamefont {T.}~\bibnamefont {{Senthil}}},
  \bibinfo {author} {\bibfnamefont {E.}~\bibnamefont {{Tutuc}}}, \bibinfo
  {author} {\bibfnamefont {A.}~\bibnamefont {{Yazdani}}},\ and\ \bibinfo
  {author} {\bibfnamefont {A.~F.}\ \bibnamefont {{Young}}},\ }\bibfield
  {title} {\bibinfo {title} {{The marvels of moir{\'e} materials}},\ }\href
  {https://doi.org/10.1038/s41578-021-00284-1} {\bibfield  {journal} {\bibinfo
  {journal} {Nature Reviews Materials}\ }\textbf {\bibinfo {volume} {6}},\
  \bibinfo {pages} {201} (\bibinfo {year} {2021})}\BibitemShut {NoStop}%
\bibitem [{\citenamefont {Carr}\ \emph {et~al.}(2020)\citenamefont {Carr},
  \citenamefont {Fang},\ and\ \citenamefont {Kaxiras}}]{Carr20}%
  \BibitemOpen
  \bibfield  {author} {\bibinfo {author} {\bibfnamefont {S.}~\bibnamefont
  {Carr}}, \bibinfo {author} {\bibfnamefont {S.}~\bibnamefont {Fang}},\ and\
  \bibinfo {author} {\bibfnamefont {E.}~\bibnamefont {Kaxiras}},\ }\bibfield
  {title} {\bibinfo {title} {Electronic-structure methods for twisted moir{\'e}
  layers},\ }\href@noop {} {\bibfield  {journal} {\bibinfo  {journal} {Nature
  Reviews Materials}\ }\textbf {\bibinfo {volume} {5}},\ \bibinfo {pages} {748}
  (\bibinfo {year} {2020})}\BibitemShut {NoStop}%
\bibitem [{\citenamefont {Cao}\ \emph {et~al.}(2018)\citenamefont {Cao},
  \citenamefont {Fatemi}, \citenamefont {Fang}, \citenamefont {Watanabe},
  \citenamefont {Taniguchi}, \citenamefont {Kaxiras},\ and\ \citenamefont
  {Jarillo-Herrero}}]{Cao18}%
  \BibitemOpen
  \bibfield  {author} {\bibinfo {author} {\bibfnamefont {Y.}~\bibnamefont
  {Cao}}, \bibinfo {author} {\bibfnamefont {V.}~\bibnamefont {Fatemi}},
  \bibinfo {author} {\bibfnamefont {S.}~\bibnamefont {Fang}}, \bibinfo {author}
  {\bibfnamefont {K.}~\bibnamefont {Watanabe}}, \bibinfo {author}
  {\bibfnamefont {T.}~\bibnamefont {Taniguchi}}, \bibinfo {author}
  {\bibfnamefont {E.}~\bibnamefont {Kaxiras}},\ and\ \bibinfo {author}
  {\bibfnamefont {P.}~\bibnamefont {Jarillo-Herrero}},\ }\bibfield  {title}
  {\bibinfo {title} {Unconventional superconductivity in magic-angle graphene
  superlattices},\ }\href@noop {} {\bibfield  {journal} {\bibinfo  {journal}
  {Nature}\ }\textbf {\bibinfo {volume} {556}},\ \bibinfo {pages} {43}
  (\bibinfo {year} {2018})}\BibitemShut {NoStop}%
\bibitem [{\citenamefont {Mele}(2010)}]{Mele10}%
  \BibitemOpen
  \bibfield  {author} {\bibinfo {author} {\bibfnamefont {E.~J.}\ \bibnamefont
  {Mele}},\ }\bibfield  {title} {\bibinfo {title} {Commensuration and
  interlayer coherence in twisted bilayer graphene},\ }\href
  {https://doi.org/10.1103/PhysRevB.81.161405} {\bibfield  {journal} {\bibinfo
  {journal} {Phys. Rev. B}\ }\textbf {\bibinfo {volume} {81}},\ \bibinfo
  {pages} {161405} (\bibinfo {year} {2010})}\BibitemShut {NoStop}%
\bibitem [{\citenamefont {Bistritzer}\ and\ \citenamefont
  {MacDonald}(2011{\natexlab{a}})}]{BM11_2}%
  \BibitemOpen
  \bibfield  {author} {\bibinfo {author} {\bibfnamefont {R.}~\bibnamefont
  {Bistritzer}}\ and\ \bibinfo {author} {\bibfnamefont {A.~H.}\ \bibnamefont
  {MacDonald}},\ }\bibfield  {title} {\bibinfo {title} {Moir\'e butterflies in
  twisted bilayer graphene},\ }\href
  {https://doi.org/10.1103/PhysRevB.84.035440} {\bibfield  {journal} {\bibinfo
  {journal} {Phys. Rev. B}\ }\textbf {\bibinfo {volume} {84}},\ \bibinfo
  {pages} {035440} (\bibinfo {year} {2011}{\natexlab{a}})}\BibitemShut
  {NoStop}%
\bibitem [{\citenamefont {Bistritzer}\ and\ \citenamefont
  {MacDonald}(2011{\natexlab{b}})}]{BM11}%
  \BibitemOpen
  \bibfield  {author} {\bibinfo {author} {\bibfnamefont {R.}~\bibnamefont
  {Bistritzer}}\ and\ \bibinfo {author} {\bibfnamefont {A.~H.}\ \bibnamefont
  {MacDonald}},\ }\bibfield  {title} {\bibinfo {title} {Moir{\'{e}} bands in
  twisted double-layer graphene},\ }\href
  {https://doi.org/10.1073/pnas.1108174108} {\bibfield  {journal} {\bibinfo
  {journal} {Proceedings of the National Academy of Sciences}\ }\textbf
  {\bibinfo {volume} {108}},\ \bibinfo {pages} {12233} (\bibinfo {year}
  {2011}{\natexlab{b}})}\BibitemShut {NoStop}%
\bibitem [{\citenamefont {Lopes~dos Santos}\ \emph {et~al.}(2012)\citenamefont
  {Lopes~dos Santos}, \citenamefont {Peres},\ and\ \citenamefont
  {Castro~Neto}}]{Lopes12}%
  \BibitemOpen
  \bibfield  {author} {\bibinfo {author} {\bibfnamefont {J.~M.~B.}\
  \bibnamefont {Lopes~dos Santos}}, \bibinfo {author} {\bibfnamefont
  {N.~M.~R.}\ \bibnamefont {Peres}},\ and\ \bibinfo {author} {\bibfnamefont
  {A.~H.}\ \bibnamefont {Castro~Neto}},\ }\bibfield  {title} {\bibinfo {title}
  {Continuum model of the twisted graphene bilayer},\ }\href
  {https://doi.org/10.1103/PhysRevB.86.155449} {\bibfield  {journal} {\bibinfo
  {journal} {Phys. Rev. B}\ }\textbf {\bibinfo {volume} {86}},\ \bibinfo
  {pages} {155449} (\bibinfo {year} {2012})}\BibitemShut {NoStop}%
\bibitem [{\citenamefont {Watson}\ \emph {et~al.}(2022)\citenamefont {Watson},
  \citenamefont {Kong}, \citenamefont {MacDonald},\ and\ \citenamefont
  {Luskin}}]{Watson22}%
  \BibitemOpen
  \bibfield  {author} {\bibinfo {author} {\bibfnamefont {A.~B.}\ \bibnamefont
  {Watson}}, \bibinfo {author} {\bibfnamefont {T.}~\bibnamefont {Kong}},
  \bibinfo {author} {\bibfnamefont {A.~H.}\ \bibnamefont {MacDonald}},\ and\
  \bibinfo {author} {\bibfnamefont {M.}~\bibnamefont {Luskin}},\ }\bibfield
  {title} {\bibinfo {title} {Bistritzer-{M}ac{D}onald dynamics in twisted
  bilayer graphene},\ }\href@noop {} {\bibfield  {journal} {\bibinfo  {journal}
  {arXiv preprint arXiv:2207.13767}\ } (\bibinfo {year} {2022})}\BibitemShut
  {NoStop}%
\bibitem [{\citenamefont {Tarnopolsky}\ \emph {et~al.}(2019)\citenamefont
  {Tarnopolsky}, \citenamefont {Kruchkov},\ and\ \citenamefont
  {Vishwanath}}]{TKV19}%
  \BibitemOpen
  \bibfield  {author} {\bibinfo {author} {\bibfnamefont {G.}~\bibnamefont
  {Tarnopolsky}}, \bibinfo {author} {\bibfnamefont {A.~J.}\ \bibnamefont
  {Kruchkov}},\ and\ \bibinfo {author} {\bibfnamefont {A.}~\bibnamefont
  {Vishwanath}},\ }\bibfield  {title} {\bibinfo {title} {Origin of magic angles
  in twisted bilayer graphene},\ }\bibfield  {journal} {\bibinfo  {journal}
  {Physical Review Letters}\ }\textbf {\bibinfo {volume} {122}},\ \href
  {https://doi.org/10.1103/physrevlett.122.106405}
  {10.1103/physrevlett.122.106405} (\bibinfo {year} {2019})\BibitemShut
  {NoStop}%
\bibitem [{\citenamefont {Becker}\ \emph {et~al.}(2021)\citenamefont {Becker},
  \citenamefont {Embree}, \citenamefont {Wittsten},\ and\ \citenamefont
  {Zworski}}]{Becker21}%
  \BibitemOpen
  \bibfield  {author} {\bibinfo {author} {\bibfnamefont {S.}~\bibnamefont
  {Becker}}, \bibinfo {author} {\bibfnamefont {M.}~\bibnamefont {Embree}},
  \bibinfo {author} {\bibfnamefont {J.}~\bibnamefont {Wittsten}},\ and\
  \bibinfo {author} {\bibfnamefont {M.}~\bibnamefont {Zworski}},\ }\bibfield
  {title} {\bibinfo {title} {Spectral characterization of magic angles in
  twisted bilayer graphene},\ }\href
  {https://doi.org/10.1103/PhysRevB.103.165113} {\bibfield  {journal} {\bibinfo
   {journal} {Phys. Rev. B}\ }\textbf {\bibinfo {volume} {103}},\ \bibinfo
  {pages} {165113} (\bibinfo {year} {2021})}\BibitemShut {NoStop}%
\bibitem [{\citenamefont {Watson}\ and\ \citenamefont {Luskin}(2021)}]{WL21}%
  \BibitemOpen
  \bibfield  {author} {\bibinfo {author} {\bibfnamefont {A.~B.}\ \bibnamefont
  {Watson}}\ and\ \bibinfo {author} {\bibfnamefont {M.}~\bibnamefont
  {Luskin}},\ }\bibfield  {title} {\bibinfo {title} {Existence of the first
  magic angle for the chiral model of bilayer graphene},\ }\href
  {https://doi.org/10.1063/5.0054122} {\bibfield  {journal} {\bibinfo
  {journal} {Journal of Mathematical Physics}\ }\textbf {\bibinfo {volume}
  {62}},\ \bibinfo {pages} {091502} (\bibinfo {year} {2021})}\BibitemShut
  {NoStop}%
\bibitem [{\citenamefont {Po}\ \emph {et~al.}(2018)\citenamefont {Po},
  \citenamefont {Zou}, \citenamefont {Vishwanath},\ and\ \citenamefont
  {Senthil}}]{Po18}%
  \BibitemOpen
  \bibfield  {author} {\bibinfo {author} {\bibfnamefont {H.~C.}\ \bibnamefont
  {Po}}, \bibinfo {author} {\bibfnamefont {L.}~\bibnamefont {Zou}}, \bibinfo
  {author} {\bibfnamefont {A.}~\bibnamefont {Vishwanath}},\ and\ \bibinfo
  {author} {\bibfnamefont {T.}~\bibnamefont {Senthil}},\ }\bibfield  {title}
  {\bibinfo {title} {Origin of mott insulating behavior and superconductivity
  in twisted bilayer graphene},\ }\href
  {https://doi.org/10.1103/PhysRevX.8.031089} {\bibfield  {journal} {\bibinfo
  {journal} {Phys. Rev. X}\ }\textbf {\bibinfo {volume} {8}},\ \bibinfo {pages}
  {031089} (\bibinfo {year} {2018})}\BibitemShut {NoStop}%
\bibitem [{\citenamefont {Trambly~de Laissardi\`ere}\ \emph
  {et~al.}(2012)\citenamefont {Trambly~de Laissardi\`ere}, \citenamefont
  {Mayou},\ and\ \citenamefont {Magaud}}]{TdLMM12}%
  \BibitemOpen
  \bibfield  {author} {\bibinfo {author} {\bibfnamefont {G.}~\bibnamefont
  {Trambly~de Laissardi\`ere}}, \bibinfo {author} {\bibfnamefont
  {D.}~\bibnamefont {Mayou}},\ and\ \bibinfo {author} {\bibfnamefont
  {L.}~\bibnamefont {Magaud}},\ }\bibfield  {title} {\bibinfo {title}
  {Numerical studies of confined states in rotated bilayers of graphene},\
  }\href {https://doi.org/10.1103/PhysRevB.86.125413} {\bibfield  {journal}
  {\bibinfo  {journal} {Phys. Rev. B}\ }\textbf {\bibinfo {volume} {86}},\
  \bibinfo {pages} {125413} (\bibinfo {year} {2012})}\BibitemShut {NoStop}%
\bibitem [{\citenamefont {Jung}\ and\ \citenamefont {MacDonald}(2014)}]{JM14}%
  \BibitemOpen
  \bibfield  {author} {\bibinfo {author} {\bibfnamefont {J.}~\bibnamefont
  {Jung}}\ and\ \bibinfo {author} {\bibfnamefont {A.~H.}\ \bibnamefont
  {MacDonald}},\ }\bibfield  {title} {\bibinfo {title} {Accurate tight-binding
  models for the $\ensuremath{\pi}$ bands of bilayer graphene},\ }\href
  {https://doi.org/10.1103/PhysRevB.89.035405} {\bibfield  {journal} {\bibinfo
  {journal} {Phys. Rev. B}\ }\textbf {\bibinfo {volume} {89}},\ \bibinfo
  {pages} {035405} (\bibinfo {year} {2014})}\BibitemShut {NoStop}%
\bibitem [{\citenamefont {Carr}\ \emph {et~al.}(2017)\citenamefont {Carr},
  \citenamefont {Massatt}, \citenamefont {Fang}, \citenamefont {Cazeaux},
  \citenamefont {Luskin},\ and\ \citenamefont {Kaxiras}}]{Carr17}%
  \BibitemOpen
  \bibfield  {author} {\bibinfo {author} {\bibfnamefont {S.}~\bibnamefont
  {Carr}}, \bibinfo {author} {\bibfnamefont {D.}~\bibnamefont {Massatt}},
  \bibinfo {author} {\bibfnamefont {S.}~\bibnamefont {Fang}}, \bibinfo {author}
  {\bibfnamefont {P.}~\bibnamefont {Cazeaux}}, \bibinfo {author} {\bibfnamefont
  {M.}~\bibnamefont {Luskin}},\ and\ \bibinfo {author} {\bibfnamefont
  {E.}~\bibnamefont {Kaxiras}},\ }\bibfield  {title} {\bibinfo {title}
  {Twistronics: Manipulating the electronic properties of two-dimensional
  layered structures through their twist angle},\ }\href
  {https://doi.org/10.1103/PhysRevB.95.075420} {\bibfield  {journal} {\bibinfo
  {journal} {Phys. Rev. B}\ }\textbf {\bibinfo {volume} {95}},\ \bibinfo
  {pages} {075420} (\bibinfo {year} {2017})}\BibitemShut {NoStop}%
\bibitem [{\citenamefont {Carr}\ \emph
  {et~al.}(2018{\natexlab{a}})\citenamefont {Carr}, \citenamefont {Fang},
  \citenamefont {Jarillo-Herrero},\ and\ \citenamefont {Kaxiras}}]{Carr18}%
  \BibitemOpen
  \bibfield  {author} {\bibinfo {author} {\bibfnamefont {S.}~\bibnamefont
  {Carr}}, \bibinfo {author} {\bibfnamefont {S.}~\bibnamefont {Fang}}, \bibinfo
  {author} {\bibfnamefont {P.}~\bibnamefont {Jarillo-Herrero}},\ and\ \bibinfo
  {author} {\bibfnamefont {E.}~\bibnamefont {Kaxiras}},\ }\bibfield  {title}
  {\bibinfo {title} {Pressure dependence of the magic twist angle in graphene
  superlattices},\ }\href {https://doi.org/10.1103/PhysRevB.98.085144}
  {\bibfield  {journal} {\bibinfo  {journal} {Phys. Rev. B}\ }\textbf {\bibinfo
  {volume} {98}},\ \bibinfo {pages} {085144} (\bibinfo {year}
  {2018}{\natexlab{a}})}\BibitemShut {NoStop}%
\bibitem [{\citenamefont {Le}\ and\ \citenamefont {Do}(2018)}]{Le18}%
  \BibitemOpen
  \bibfield  {author} {\bibinfo {author} {\bibfnamefont {H.~A.}\ \bibnamefont
  {Le}}\ and\ \bibinfo {author} {\bibfnamefont {V.~N.}\ \bibnamefont {Do}},\
  }\bibfield  {title} {\bibinfo {title} {Electronic structure and optical
  properties of twisted bilayer graphene calculated via time evolution of
  states in real space},\ }\href {https://doi.org/10.1103/PhysRevB.97.125136}
  {\bibfield  {journal} {\bibinfo  {journal} {Phys. Rev. B}\ }\textbf {\bibinfo
  {volume} {97}},\ \bibinfo {pages} {125136} (\bibinfo {year}
  {2018})}\BibitemShut {NoStop}%
\bibitem [{\citenamefont {Moon}\ and\ \citenamefont {Koshino}(2013)}]{Moon13}%
  \BibitemOpen
  \bibfield  {author} {\bibinfo {author} {\bibfnamefont {P.}~\bibnamefont
  {Moon}}\ and\ \bibinfo {author} {\bibfnamefont {M.}~\bibnamefont {Koshino}},\
  }\bibfield  {title} {\bibinfo {title} {Optical absorption in twisted bilayer
  graphene},\ }\href {https://doi.org/10.1103/PhysRevB.87.205404} {\bibfield
  {journal} {\bibinfo  {journal} {Phys. Rev. B}\ }\textbf {\bibinfo {volume}
  {87}},\ \bibinfo {pages} {205404} (\bibinfo {year} {2013})}\BibitemShut
  {NoStop}%
\bibitem [{\citenamefont {Carr}\ \emph {et~al.}(2019)\citenamefont {Carr},
  \citenamefont {Fang}, \citenamefont {Zhu},\ and\ \citenamefont
  {Kaxiras}}]{Carr19}%
  \BibitemOpen
  \bibfield  {author} {\bibinfo {author} {\bibfnamefont {S.}~\bibnamefont
  {Carr}}, \bibinfo {author} {\bibfnamefont {S.}~\bibnamefont {Fang}}, \bibinfo
  {author} {\bibfnamefont {Z.}~\bibnamefont {Zhu}},\ and\ \bibinfo {author}
  {\bibfnamefont {E.}~\bibnamefont {Kaxiras}},\ }\bibfield  {title} {\bibinfo
  {title} {Exact continuum model for low-energy electronic states of twisted
  bilayer graphene},\ }\href {https://doi.org/10.1103/PhysRevResearch.1.013001}
  {\bibfield  {journal} {\bibinfo  {journal} {Phys. Rev. Research}\ }\textbf
  {\bibinfo {volume} {1}},\ \bibinfo {pages} {013001} (\bibinfo {year}
  {2019})}\BibitemShut {NoStop}%
\bibitem [{\citenamefont {Fang}\ \emph {et~al.}(2019)\citenamefont {Fang},
  \citenamefont {Carr}, \citenamefont {Zhu}, \citenamefont {Massatt},\ and\
  \citenamefont {Kaxiras}}]{Fang19}%
  \BibitemOpen
  \bibfield  {author} {\bibinfo {author} {\bibfnamefont {S.}~\bibnamefont
  {Fang}}, \bibinfo {author} {\bibfnamefont {S.}~\bibnamefont {Carr}}, \bibinfo
  {author} {\bibfnamefont {Z.}~\bibnamefont {Zhu}}, \bibinfo {author}
  {\bibfnamefont {D.}~\bibnamefont {Massatt}},\ and\ \bibinfo {author}
  {\bibfnamefont {E.}~\bibnamefont {Kaxiras}},\ }\bibfield  {title} {\bibinfo
  {title} {Angle-dependent ab initio low-energy hamiltonians for a relaxed
  twisted bilayer graphene heterostructure}} (\bibinfo {year}
  {2019})\BibitemShut {NoStop}%
\bibitem [{\citenamefont {Koshino}\ and\ \citenamefont
  {Nam}(2020)}]{Koshino20}%
  \BibitemOpen
  \bibfield  {author} {\bibinfo {author} {\bibfnamefont {M.}~\bibnamefont
  {Koshino}}\ and\ \bibinfo {author} {\bibfnamefont {N.~N.~T.}\ \bibnamefont
  {Nam}},\ }\bibfield  {title} {\bibinfo {title} {Effective continuum model for
  relaxed twisted bilayer graphene and moir\'e electron-phonon interaction},\
  }\href {https://doi.org/10.1103/PhysRevB.101.195425} {\bibfield  {journal}
  {\bibinfo  {journal} {Phys. Rev. B}\ }\textbf {\bibinfo {volume} {101}},\
  \bibinfo {pages} {195425} (\bibinfo {year} {2020})}\BibitemShut {NoStop}%
\bibitem [{\citenamefont {Bernevig}\ \emph {et~al.}(2021)\citenamefont
  {Bernevig}, \citenamefont {Song}, \citenamefont {Regnault},\ and\
  \citenamefont {Lian}}]{Bernevig21}%
  \BibitemOpen
  \bibfield  {author} {\bibinfo {author} {\bibfnamefont {B.~A.}\ \bibnamefont
  {Bernevig}}, \bibinfo {author} {\bibfnamefont {Z.-D.}\ \bibnamefont {Song}},
  \bibinfo {author} {\bibfnamefont {N.}~\bibnamefont {Regnault}},\ and\
  \bibinfo {author} {\bibfnamefont {B.}~\bibnamefont {Lian}},\ }\bibfield
  {title} {\bibinfo {title} {Twisted bilayer graphene. i. matrix elements,
  approximations, perturbation theory, and a
  $k\ifmmode\cdot\else\textperiodcentered\fi{}p$ two-band model},\ }\href
  {https://doi.org/10.1103/PhysRevB.103.205411} {\bibfield  {journal} {\bibinfo
   {journal} {Phys. Rev. B}\ }\textbf {\bibinfo {volume} {103}},\ \bibinfo
  {pages} {205411} (\bibinfo {year} {2021})}\BibitemShut {NoStop}%
\bibitem [{\citenamefont {Canc{\`e}s}\ \emph {et~al.}(2013)\citenamefont
  {Canc{\`e}s}, \citenamefont {Lahbabi},\ and\ \citenamefont {Lewin}}]{CLL13}%
  \BibitemOpen
  \bibfield  {author} {\bibinfo {author} {\bibfnamefont {{\'E}.}~\bibnamefont
  {Canc{\`e}s}}, \bibinfo {author} {\bibfnamefont {S.}~\bibnamefont
  {Lahbabi}},\ and\ \bibinfo {author} {\bibfnamefont {M.}~\bibnamefont
  {Lewin}},\ }\bibfield  {title} {\bibinfo {title} {Mean-field models for
  disordered crystals},\ }\href {https://doi.org/10.1016/j.matpur.2012.12.003}
  {\bibfield  {journal} {\bibinfo  {journal} {Journal de Math{\'{e}}matiques
  Pures et Appliqu{\'{e}}es}\ }\textbf {\bibinfo {volume} {100}},\ \bibinfo
  {pages} {241} (\bibinfo {year} {2013})}\BibitemShut {NoStop}%
\bibitem [{\citenamefont {Tritsaris}\ \emph {et~al.}(2016)\citenamefont
  {Tritsaris}, \citenamefont {Shirodkar}, \citenamefont {Kaxiras},
  \citenamefont {Cazeaux}, \citenamefont {Luskin}, \citenamefont {Plech{\'a}{\v
  c}},\ and\ \citenamefont {Canc{\`e}s}}]{TSKCLPC16}%
  \BibitemOpen
  \bibfield  {author} {\bibinfo {author} {\bibfnamefont {G.~A.}\ \bibnamefont
  {Tritsaris}}, \bibinfo {author} {\bibfnamefont {S.~N.}\ \bibnamefont
  {Shirodkar}}, \bibinfo {author} {\bibfnamefont {E.}~\bibnamefont {Kaxiras}},
  \bibinfo {author} {\bibfnamefont {P.}~\bibnamefont {Cazeaux}}, \bibinfo
  {author} {\bibfnamefont {M.}~\bibnamefont {Luskin}}, \bibinfo {author}
  {\bibfnamefont {P.}~\bibnamefont {Plech{\'a}{\v c}}},\ and\ \bibinfo {author}
  {\bibfnamefont {E.}~\bibnamefont {Canc{\`e}s}},\ }\bibfield  {title}
  {\bibinfo {title} {Perturbation theory for weakly coupled two-dimensional
  layers},\ }\href {https://doi.org/10.1557/jmr.2016.99} {\bibfield  {journal}
  {\bibinfo  {journal} {Journal of Materials Research}\ }\textbf {\bibinfo
  {volume} {31}},\ \bibinfo {pages} {959} (\bibinfo {year} {2016})}\BibitemShut
  {NoStop}%
\bibitem [{\citenamefont {Fefferman}\ and\ \citenamefont
  {Weinstein}(2013)}]{FefWei13}%
  \BibitemOpen
  \bibfield  {author} {\bibinfo {author} {\bibfnamefont {C.~L.}\ \bibnamefont
  {Fefferman}}\ and\ \bibinfo {author} {\bibfnamefont {M.~I.}\ \bibnamefont
  {Weinstein}},\ }\bibfield  {title} {\bibinfo {title} {Wave packets in
  honeycomb structures and two-dimensional dirac equations},\ }\href
  {https://doi.org/10.1007/s00220-013-1847-2} {\bibfield  {journal} {\bibinfo
  {journal} {Communications in Mathematical Physics}\ }\textbf {\bibinfo
  {volume} {326}},\ \bibinfo {pages} {251} (\bibinfo {year}
  {2013})}\BibitemShut {NoStop}%
\bibitem [{\citenamefont {Bezanson}\ \emph {et~al.}(2017)\citenamefont
  {Bezanson}, \citenamefont {Edelman}, \citenamefont {Karpinski},\ and\
  \citenamefont {Shah}}]{Julia}%
  \BibitemOpen
  \bibfield  {author} {\bibinfo {author} {\bibfnamefont {J.}~\bibnamefont
  {Bezanson}}, \bibinfo {author} {\bibfnamefont {A.}~\bibnamefont {Edelman}},
  \bibinfo {author} {\bibfnamefont {S.}~\bibnamefont {Karpinski}},\ and\
  \bibinfo {author} {\bibfnamefont {V.}~\bibnamefont {Shah}},\ }\bibfield
  {title} {\bibinfo {title} {Julia: A fresh approach to numerical computing},\
  }\href {https://doi.org/10.1137/141000671} {\bibfield  {journal} {\bibinfo
  {journal} {SIAM Rev.}\ }\textbf {\bibinfo {volume} {59}},\ \bibinfo {pages}
  {65} (\bibinfo {year} {2017})}\BibitemShut {NoStop}%
\bibitem [{\citenamefont {Herbst}\ \emph {et~al.}(2021)\citenamefont {Herbst},
  \citenamefont {Levitt},\ and\ \citenamefont {Canc{\`e}s}}]{DFTK}%
  \BibitemOpen
  \bibfield  {author} {\bibinfo {author} {\bibfnamefont {M.~F.}\ \bibnamefont
  {Herbst}}, \bibinfo {author} {\bibfnamefont {A.}~\bibnamefont {Levitt}},\
  and\ \bibinfo {author} {\bibfnamefont {E.}~\bibnamefont {Canc{\`e}s}},\
  }\bibfield  {title} {\bibinfo {title} {{DFTK}: A julian approach for
  simulating electrons in solids},\ }in\ \href@noop {} {\emph {\bibinfo
  {booktitle} {Proceedings of the JuliaCon Conferences}}},\ Vol.~\bibinfo
  {volume} {3}\ (\bibinfo {year} {2021})\ p.~\bibinfo {pages} {69}\BibitemShut
  {NoStop}%
\bibitem [{\citenamefont {Nam}\ and\ \citenamefont {Koshino}(2017)}]{Nam17}%
  \BibitemOpen
  \bibfield  {author} {\bibinfo {author} {\bibfnamefont {N.~N.~T.}\
  \bibnamefont {Nam}}\ and\ \bibinfo {author} {\bibfnamefont {M.}~\bibnamefont
  {Koshino}},\ }\bibfield  {title} {\bibinfo {title} {Lattice relaxation and
  energy band modulation in twisted bilayer graphene},\ }\href
  {https://doi.org/10.1103/PhysRevB.96.075311} {\bibfield  {journal} {\bibinfo
  {journal} {Phys. Rev. B}\ }\textbf {\bibinfo {volume} {96}},\ \bibinfo
  {pages} {075311} (\bibinfo {year} {2017})}\BibitemShut {NoStop}%
\bibitem [{\citenamefont {Carr}\ \emph
  {et~al.}(2018{\natexlab{b}})\citenamefont {Carr}, \citenamefont {Massatt},
  \citenamefont {Torrisi}, \citenamefont {Cazeaux}, \citenamefont {Luskin},\
  and\ \citenamefont {Kaxiras}}]{CarrMassatt18}%
  \BibitemOpen
  \bibfield  {author} {\bibinfo {author} {\bibfnamefont {S.}~\bibnamefont
  {Carr}}, \bibinfo {author} {\bibfnamefont {D.}~\bibnamefont {Massatt}},
  \bibinfo {author} {\bibfnamefont {S.~B.}\ \bibnamefont {Torrisi}}, \bibinfo
  {author} {\bibfnamefont {P.}~\bibnamefont {Cazeaux}}, \bibinfo {author}
  {\bibfnamefont {M.}~\bibnamefont {Luskin}},\ and\ \bibinfo {author}
  {\bibfnamefont {E.}~\bibnamefont {Kaxiras}},\ }\bibfield  {title} {\bibinfo
  {title} {Relaxation and domain formation in incommensurate two-dimensional
  heterostructures},\ }\href {https://doi.org/10.1103/PhysRevB.98.224102}
  {\bibfield  {journal} {\bibinfo  {journal} {Phys. Rev. B}\ }\textbf {\bibinfo
  {volume} {98}},\ \bibinfo {pages} {224102} (\bibinfo {year}
  {2018}{\natexlab{b}})}\BibitemShut {NoStop}%
\bibitem [{\citenamefont {Cazeaux}\ \emph {et~al.}(2020)\citenamefont
  {Cazeaux}, \citenamefont {Luskin},\ and\ \citenamefont
  {Massatt}}]{Cazeaux20}%
  \BibitemOpen
  \bibfield  {author} {\bibinfo {author} {\bibfnamefont {P.}~\bibnamefont
  {Cazeaux}}, \bibinfo {author} {\bibfnamefont {M.}~\bibnamefont {Luskin}},\
  and\ \bibinfo {author} {\bibfnamefont {D.}~\bibnamefont {Massatt}},\
  }\bibfield  {title} {\bibinfo {title} {Energy minimization of two dimensional
  incommensurate heterostructures},\ }\href
  {https://doi.org/10.1007/s00205-019-01444-y} {\bibfield  {journal} {\bibinfo
  {journal} {Archive for Rational Mechanics and Analysis}\ }\textbf {\bibinfo
  {volume} {235}},\ \bibinfo {pages} {1289} (\bibinfo {year}
  {2020})}\BibitemShut {NoStop}%
\end{thebibliography}%


\pagebreak

\widetext
\begin{center}
    \textbf{\large Supplemental Materials of the article ``A simple derivation of moiré-scale continuous models for twisted bilayer graphene'', \\ written by \'Eric Cancès, Louis Garrigue and David Gontier}
\end{center}

\renewcommand{\theequation}{S\arabic{equation}}
\renewcommand{\thefigure}{S\arabic{figure}}
\renewcommand{\bibnumfmt}[1]{[S#1]}
\renewcommand{\citenumfont}[1]{S#1}

\setcounter{equation}{0}
\setcounter{figure}{0}
\setcounter{table}{0}
\setcounter{section}{0}
\setcounter{page}{1}

\section{Symmetries of single layer graphene}
\label{supp:graphene}

In this section, we recall the symmetries of a single graphene-sheet, and explain how to choose the orientation of $\Phi_1$ and $\Phi_2$ in order to satisfy~\eqref{eq:pre_Dirac}. These functions will actually satisfy other symmetries that we will use later when studying the symmetries of TBG. For $\by \in \R^2$, $\theta \in \R$, and $f : \R^3 \to \C$, we set
\begin{align*}
    &(\tau_\by f)(\bx,z) :=  f(\bx-\by,z), & \mbox{(horizontal translation of vector $\by$),} \\
    & (\cR_\theta f)(\bx,z) := f(R_{-\theta}\bx,z), & \mbox{(rotation of angle $\theta$ around the $z$-axis),} \\
    & ({\fR}f)(x_1,x_2,z):=f(x_1,-x_2,-z) & \mbox{(rotation of angle $\pi$ around the $x_1$-axis),} \\
    & (\cP f)(\bx,z)=(\cR_{\pi}f)(\bx,z):=  f(-\bx,z), & \mbox{(in-plane parity operator),}\\
    &(\cC f)(\br):=\overline{f(\br)},& \mbox{(complex conjugation)}, \\
    & (\cS f)(\bx,z):=f(\bx,-z) & \mbox{(mirror symmetry w.r.t. the plane $z=0$).}
\end{align*}

\medskip

We denote by $H_\bk^{(1)}$ the Bloch fibers of $H^{(1)}$, acting on the set of $\bk$-quasiperiodic functions $\phi(\bx - \bR, z) = \re^{- \ri\bk \cdot \bR} \phi(\bx, z)$ for $\bR \in \Lat$. We denote the set of such quasiperiodic wave functions by $L^2_\bk$, with usual inner product (recall that $\Omega \subset \R^2$ is the Wigner-Seitz cell of the lattice $\L$, an hexagon with a carbon atom at each of its six vertices)
$$
\langle \phi,\psi\rangle_{L^2_\bk} := \int_{\Omega \times \R} \overline{\phi(\bx,z)} \, \psi(\bx,z) \, \d\bx \, \d z.
$$

The single layer potential $V$ is $\Lat$-periodic, and satisfies  $\cR_{\frac{2 \pi}{3}} V = \fR V = \cP V = \cC V = \cS V  = V$. This implies
\[
\cC {H_\bk} \cC^* = {H_{-\bk}},  \quad
\cP {H_\bk} \cP^* = {H_{-\bk}}, \quad 
\fR H_\bk \fR = H_{M \bk}, \quad \text{and} \quad
\cR_{\frac{2\pi}3} {H_\bk} \cR_{\frac{2\pi}3}^* = {H_{R_{\frac{2\pi}3}\bk}},
\]
where the operators $\cP, \fR$ and $\cR_{\frac{2\pi}3}$ (resp. $\cC$) are here considered as unitary (resp. anti-unitary) operators between two different fibers of the Bloch bundle. The matrix $M := \begin{pmatrix}
    1 & 0 \\ 0 & -1
\end{pmatrix}$ represents the in-plane mirror symmetry with respect to $x_1$-axis. At the Dirac point $\bK := \frac13 (\ba_1^* + \ba_2^*) $, we have $R_{\frac{2\pi}3} \bK = \bK \ {\rm mod} \ \RLat$ and $M \bK = \bK$, so the operator $H_\bK$ commutes with $(\cC \cP)$, $\fR$, $\cS$ and $\cR_{\frac{2\pi}3}$. 

\medskip

Since $\cR_{\frac{2\pi}3}$ satisfies $\cR_{\frac{2\pi}3}^3 = \bbI_3$, the eigenvalues of this operator are $1$, $\omega := \re^{ \ri 2 \pi/3}$ and $\omega^2 = \overline{\omega}$. Let us decompose $L^2_\bK$ accordingly as
\[
L^2_\bK = L^2_{\bK,1} \oplus L^2_{\bK,\omega} \oplus L^2_{\bK,\overline{\omega}}.
\]
Let $\mu$ be an eigenvalue of multiplicity $1$ of $H^{(1)}_{\bK, \omega}$, the restriction of $H^{(1)}_{\bK}$ to $L^2_{\bK, \omega}$, and let $\widetilde{\Phi_1} \in L^2_{\bK, \omega}$ be a corresponding normalized eigenvector. Since
\[
    \fR \cR_\theta = \cR_{-\theta} \fR, \quad \text{and} \quad (\cC \cP) \cR_\theta = \cR_\theta (\cC \cP),
\]
the function $\fR \cC \cP \widetilde{\Phi_1}$ also belongs to $L^2_{\bK, \omega}$, and is an eigenvector of $H^{(1)}_{\bK, \omega}$ for the same eigenvalue $\mu$. As a consequence, it is collinear to $\widetilde{\Phi_1}$, so there is $\alpha \in \R$ for which
\[
    \fR  \cC \cP  \widetilde{\Phi_1} = \re^{ \ri \alpha} \widetilde{\Phi_1}.
\]
We set
\begin{equation}\label{eq:phases_Phi12}
    \Phi_1 = \ri^\sigma \re^{ \ri \alpha/2} \widetilde{\Phi_1} \in L^2_{\bK, \omega}, \quad  \text{and} \quad 
    \Phi_2= (\cC \cP) \Phi_1  \in L^2_{\bK, \overline{\omega}},
\end{equation}
where $\sigma \in \{ 0, 1\}$ will be chosen later. 
Since $\Phi_1 \in L^2_{\bK, \omega}$ while $\Phi_2 \in  L^2_{\bK, \overline{\omega}}$, the functions $\Phi_1$ and $\Phi_2$ are orthogonal, and $\mu$ is an eigenvalue of $H^{(1)}_\bK$ of multiplicity at least $2$. This happens in particular at the Fermi-level $\mu = \mu_F$ of uncharged graphene, which is a two-fold degenerate eigenvalue of $H_\bK$. In what follows, $\Phi_1$ and $\Phi_2$ are normalized eigenfunctions of $H_\bK$ (for the eigenvalue $\mu_{\rm F}$), $\cR_{\frac{2\pi}3}$ (for the eigenvalues $\omega$ and $\overline{\omega}$ respectively), and $\fR  \cC \cP$ (for the eigenvalue $(-1)^\sigma$) constructed as above. 
The valence orbitals of graphene belong to a $\pi$--shell, hence we also have
\begin{align} \label{eq:sym_z_Phi12}
    \Phi_1(\bx,-z) = -\Phi_1(\bx,z), \quad \text{and} \quad \Phi_2(\bx,-z) = -\Phi_2(\bx,z).
\end{align}
We now define the Fermi velocity. For $j,j' \in \{ 1, 2\}$, we introduce the vector
\[
    \bmm_{jj'} := \langle \Phi_j, (- \ri\nabla_\bx) \Phi_{j'} \rangle.
\]
Since  $( - \ri \nabla_\bx\Phi_1)(R_{\frac{2 \pi}{3}}^*\bx,z) = R_{\frac{2 \pi}{3}}^* [( - \ri \nabla_\bx) (\cR_{\frac{2\pi}{3}}\Phi_1)](\bx,z)$ and $\cR_{\frac{2\pi}{3}} \Phi_1=\omega\Phi_1$, we have
\begin{align} 
    \bmm_{11} & = \int_{\Omega \times \R} \overline{\Phi_1(\bx,z)} (- \ri \nabla_\bx \Phi_1)(\bx,z) \, \d\bx \, \d z 
        = \int_{\Omega \times \R} \overline{\Phi_1(R_{\frac{2 \pi}{3}}^*\bx,z)} (- \ri \nabla_\bx \Phi_1)(R_{\frac{2 \pi}{3}}^*\bx,z) \, \d\bx \, \d z  \nonumber  \\
            &= 
     \langle \cR_{\frac{2\pi}{3}} \Phi_1,  R_{\frac{2 \pi}{3}}^* ( - \ri \nabla_\bx) (\cR_{\frac{2\pi}{3}} \Phi_1) \rangle     = R_{\frac{2 \pi}{3}}^* \langle \omega \Phi_1,  - \ri \nabla_\bx (\omega   \Phi_1) \rangle 
    = R_{\frac{2 \pi}{3}}^*  \bmm_{11}. \label{eq:eqt_bmm11}
\end{align}
This implies that $\bmm_{11} = (0, 0)^T$. Likewise, we obtain $\bmm_{22} = (0, 0)^T$. A similar calculation using $\cR_{\frac{2\pi}{3}} \Phi_2=\overline{\omega}\Phi_2$ leads to $\bmm_{12} = \omega R_{\frac{2 \pi}{3}}^* \bmm_{12}$, which shows that $\bmm_{12}$ belongs to the eigenspace of $R_{\frac{2 \pi}{3}}$ associated with the eigenvalue $\omega$. We deduce that it is collinear to $(1,-\ri)^T$: there exists $v_F \in \C$ so that
\begin{align*}
    \langle \Phi_1, (- \ri\nabla_\bx) \Phi_1 \rangle = \langle \Phi_2, (- \ri\nabla_\bx) \Phi_2 \rangle = \begin{pmatrix}
        0 \\ 0
    \end{pmatrix}, \quad \text{and} \quad 
    \langle \Phi_1,  (-\ri \nabla_\bx) \Phi_2 \rangle 
    = v_{\rm F}  \begin{pmatrix} 1 \\ -\ri \end{pmatrix}.
 \end{align*}

Using that $\fR ( - \ri \partial_{x_1}) = ( - \ri \partial_{x_1}) \fR$ and $\fR \Phi_1 = (-1)^\sigma \Phi_2$, we have
\[
   v_F = \langle \Phi_1,  (-\ri \partial_{x_1}) \Phi_2 \rangle = \langle \fR \Phi_1,  \fR (-\ri \partial_{x_1}) \Phi_2 \rangle 
    = (-1)^{2 \sigma} \langle \Phi_2,  (-\ri \partial_{x_1}) \Phi_1 \rangle = \overline{v_F},
\]
hence $v_F \in \R$. Changing $\sigma$ into $(1 - \sigma)$ changes $v_F$ into $-v_F$, and we therefore choose~$\sigma$ so that $v_F \ge 0$. Numerical simulations give $\sigma=1$. The quantity $v_{\rm F}$, which turns out to be strictly positive for graphene, is the single-layer graphene Fermi velocity. It follows from perturbation theory that it is equal to the slope of the Dirac cone at $\bK$. Note that with $\sigma = 1$, we have $\fR \Phi_1 = - \Phi_2$, hence also $\Phi_1(x_1,-x_2, z) = \Phi_2(x_1, x_2, z)$.

\medskip

We finally denote by $u_1$ and $u_2$ the unique functions in $L^2_\per$ such that 
\begin{equation} \label{eq:Phijuj}
    \Phi_1(\bx, z) = \re^{\ri \bK \cdot \bx} u_1(\bx, z), \quad \text{and} \quad
    \Phi_2(\bx, z)= \re^{\ri \bK \cdot \bx} u_2(\bx, z).
\end{equation}
It results from the above arguments that the functions $u_1$ and $u_2$ satisfy the normalization conditions 
$$
\langle u_j,u_{j'} \rangle = \delta_{jj'}, 
$$
and the symmetry properties
\begin{align} \label{eq:symmetries_uj}
& (\cR_{\f{2\pi}{3}} u_j)(\bx,z) = \omega^j \re^{\ri (1-R_{\f{2\pi}{3}}) \bK \cdot \bx} u_j(\bx,z), \\
&  u_j(\bx,-z) = -u_j(\bx,z), \quad u_2(\bx,z)=\overline{u_1(-\bx,z)}, \quad u_2(x_1,x_2,z)= u_1(x_1,-x_2,z) . \label{eq:symmetries_uj2}
\end{align}

\begin{remark}
    Consider a tight-binding model in which each $2p_z$ orbital is represented by a function of the form $\chi(\bx, z) := \zeta(\bx) h(z)$ with $\zeta$ radial decreasing, $h(-z) = - h(z)$, and $\| \zeta \|_{L^2(\mathbb R^2)} = \| h \|_{L^2(\mathbb R)} = 1$. If only nearest-neighbor orbitals overlap, one can take $\Phi_1$ and $\Phi_2$ to be, for $\bx$ in the Wigner-Seitz cell $\Omega$, (see also Fig.~\ref{fig:Phi12})
    \begin{equation} \label{eq:Phi12_schema}
       \forall \br \in \Omega \times \R, \quad  \begin{cases}
            \Phi_1(\br) & = \sum_{\bR \in \Lat} \chi(\br - (\bR + \bR_B)) \re^{\ri \bK \cdot (\bR + \bR_B)} 
            = \chi(\br + \bdelta_1) + \omega^2 \chi(\br + \bdelta_2) + \omega \chi(\br + \bdelta_3) \\
            \Phi_2(\br) & = \sum_{\bR \in \Lat} \chi(\br - (\bR + \bR_A)) \re^{\ri \bK \cdot (\bR + \bR_A)} 
            = \chi(\br - \bdelta_1) + \omega \chi(\br - \bdelta_2) + \omega^2 \chi(\br - \bdelta_3),
        \end{cases}
    \end{equation}
    with $\bR_A = \frac13 \ba_1 + \frac23 \ba_2$ and $\bR_B = \frac23 \ba_1 + \frac13 \ba_2$. One can check that this choice satisfies the above symmetries. In addition, one finds $\langle \Phi_1, (- \ri \nabla_\bx) \Phi_2 \rangle = v_F (1, - \ri)^T$, with the Fermi velocity
    \[
        v_F := - \frac32 g'(a), \quad \text{with} \quad
        G(\bdelta) := (\zeta * \zeta)(\bdelta)=g(|\bdelta|), \qquad \nabla G(\bdelta) = g'(|\bdelta|) \dfrac{\bdelta}{| \bdelta |}.
    \]
    Since $\zeta$ is radial decreasing, so is $G$, hence $g'(a) < 0$ (recall that $a$ is the C-C distance), and the Fermi velocity is positive.
\end{remark}

\begin{figure}[ht]
    \centering
    
    \includegraphics[width=0.25\textwidth]{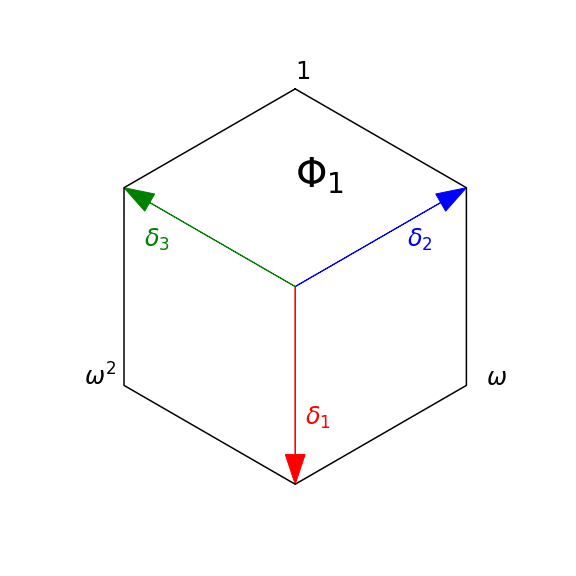}
    \includegraphics[width=0.25\textwidth]{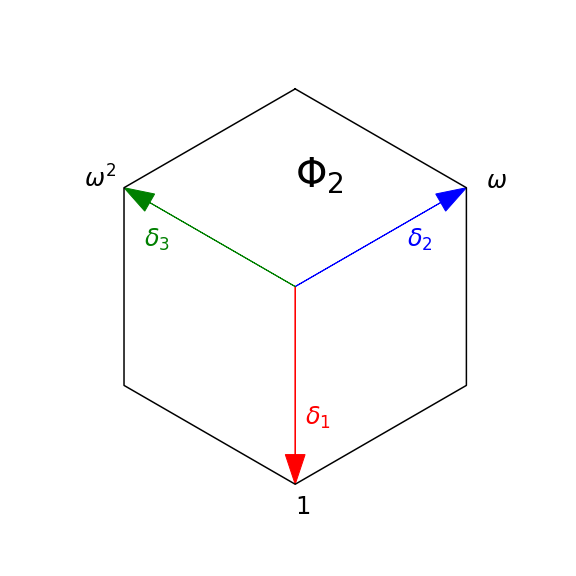} 
    \caption{Representation of the functions $\Phi_1$ and $\Phi_2$ in~\eqref{eq:Phi12_schema}.}
    \label{fig:Phi12}
\end{figure}

\section{Symmetries of our reduced model for TBG}
\label{app:symmetries}

\subsubsection{Symmetries induced by the atomic configuration} 
We now record the symmetries of the operator induced by the specific atomic configuration of the TBG. These symmetries are valid for all $d > 0$ and $\theta > 0$. First, we note that
\[
\left[ \cR_{\tfrac{2\pi}{3}},  \U \right] = \left[ \cP,  \U \right] = \left[ \cC,  \U \right] = 0
\]
while $ \fR \U = \U^* \fR$. We deduce that
\[
\left[ H_{d, \theta}^{(2)}, \cR_{\tfrac{2\pi}{3}} \right] =
\left[ H_{d, \theta}^{(2)}, \cP \right] =
\left[ H_{d, \theta}^{(2)}, \cC \right] = 
\left[ H_{d, \theta}^{(2)}, \fR \right] = 0.
\]
We now study how these commutation relations translate in our reduced model. Let us first focus on the $3$-fold rotation symmetry. Using~\eqref{eq:phases_Phi12}, we obtain that, for $\alpha : \R^2 \to \C^4$, we have
\[
\cR_{\tfrac{2\pi}{3}} \left( \alpha : \Phi \right)_{d, \theta} = \left(( \cR_{\tfrac{2\pi}{3}} \alpha) : ( \cR_{\tfrac{2\pi}{3}} \Phi) \right)_{d, \theta} = \left( (\widetilde{\cR_{\tfrac{2\pi}{3}}} \alpha ) : \Phi \right)_{d, \theta},
\]
where we introduced the unitary operator $\widetilde{\cR_{\tfrac{2\pi}{3}}} : L^2(\R^2;\C^4) \to L^2(\R^2; \C^4)$ defined by
\[
\widetilde{\cR_{\tfrac{2\pi}{3}}} \begin{pmatrix}
    \alpha_{+,1} \\ \alpha_{+,2} \\ \alpha_{-,1} \\ \alpha_{-,2} 
\end{pmatrix}
= \cR_{\tfrac{2\pi}{3}} \begin{pmatrix}
    \omega  \alpha_{+,1} \\ \omega^2 \alpha_{+,2} \\ \omega \alpha_{-,1} \\ \omega^2 \alpha_{-,2} 
\end{pmatrix}
= \begin{pmatrix}
    \omega & 0 & 0 & 0 \\
    0 & \omega^2 & 0 & 0 \\
    0 & 0 & \omega & 0 \\
    0 & 0 & 0 & \omega^2
\end{pmatrix} \cR_{\tfrac{2\pi}{3}} \begin{pmatrix}
    \alpha_{+,1} \\ \alpha_{+,2} \\  \alpha_{-,1} \\  \alpha_{-,2} 
\end{pmatrix}.
\]
We deduce that
\[
\left\langle \left( \widetilde{\alpha} : \Phi \right)_{d, \theta}, \left(H_{d, \theta}^{(2)} - \mu_F \right) \left(\alpha : \Phi\right)_{d, \theta} \right\rangle 
=
\left\langle \left( ( \widetilde{\cR_{\tfrac{2\pi}{3}}} \widetilde{\alpha}) : \Phi \right)_{d, \theta}, \left(H_{d, \theta}^{(2)} - \mu_F \right) \left( (\widetilde{\cR_{\tfrac{2\pi}{3}}} \alpha) : \Phi\right)_{d, \theta} \right\rangle ,
\]
and a similar formula for $\left\langle \left( \widetilde{\alpha} : \Phi \right)_{d, \theta}, \left(\alpha : \Phi\right)_{d, \theta} \right\rangle $. Hence
\[
\left[ \cS_d,  \widetilde{\cR_{\tfrac{2\pi}{3}}}  \right]= 0, \quad \text{and} \quad  \left[ \cH_{d, \theta},  \widetilde{\cR_{\tfrac{2\pi}{3}}}  \right] = 0.
\]
These commutation relations give rises to symmetry properties on the entries of $\cS_d$ and~$\cH_{d,\theta}$. For $\bbV_d$, we have (the results are similar for $\Sigma_d$ and $\bbW_d$)

\begin{align}
    &  [\bbV_d]_{11}(R_{\frac{2 \pi}{3}}^* \bX)  =  [\bbV_d]_{11}(\bX),  &  & [\bbV_d]_{12}(R_{\frac{2 \pi}{3}}^* \bX)  = \omega  [\bbV_d]_{12}(\bX),  \label{eq:rotationSV1} \\
    & [\bbV_d]_{21}(R_{\frac{2 \pi}{3}}^* \bX)  = \overline{\omega}  [\bbV_d]_{21}(\bX),  & & [\bbV_d]_{22}(R_{\frac{2 \pi}{3}}^* \bX)  = [\bbV_d]_{22}(\bX)  \label{eq:rotationSV2}.
\end{align}
\medskip

We now consider the $\cC\cP$ symmetry. Using the definition of $\Phi_2$ in~\eqref{eq:phases_Phi12} we deduce that for $\alpha \in L^2(\R^2; \C^4)$, 
\[
(\cC\cP) \left( \alpha : \Phi \right)_{d, \theta} = \left(( \cC\cP \alpha) : ( \cC\cP \Phi) \right)_{d, \theta} = \left( (\widetilde{\cC\cP} \alpha ) : \Phi \right)_{d, \theta},
\]
where $\widetilde{\cC\cP}$ is the anti-unitary of $L^2(\R^2;\C^4)$ defined by
\[
\widetilde{\cC\cP} \begin{pmatrix}
    \alpha_{+,1} \\ \alpha_{+,2} \\ \alpha_{-,1} \\ \alpha_{-,2} 
\end{pmatrix}
= (\cC\cP) \begin{pmatrix}
    \alpha_{+,2} \\ \alpha_{+,1} \\ \alpha_{-,2} \\ \alpha_{-,1} 
\end{pmatrix} = 
\begin{pmatrix}
    0  & (\cC \cP) & 0 & 0 \\
    (\cC \cP) & 0 & 0 & 0 \\
    0 & 0 & 0 & (\cC \cP) \\
    0 & 0 & (\cC \cP) & 0
\end{pmatrix}
\begin{pmatrix}
    \alpha_{+,1} \\ \alpha_{+,2} \\ \alpha_{-,1} \\ \alpha_{-,2} 
\end{pmatrix}.
\]
We deduce as previously that $\left[ \cS_d,  \widetilde{\cC\cP}  \right]= 0$ and $\left[ \cH_{d, \theta},  \widetilde{ \cC\cP}  \right] = 0$, from which we infer, for $\bbV_d$ (the results are similar for $\Sigma_d$ and $\bbW_d^\pm$)
\begin{align}
    & \overline{[\bbV_d]_{11}(-\bX)}  =  [\bbV_d]_{22}(\bX) \quad \mbox{and} \quad \overline{[\bbV_d]_{12}(-\bX)}  = [\bbV_d]_{21}(\bX).  \label{eq:CPSV} 
\end{align}

Finally, we consider the mirror symmetry, coming from the $\fR$ operator. Using~the relation $\fR\Phi_1=- \Phi_2$, we get
\[
(\fR) \left( \alpha : \Phi \right)_{d, \theta} = \left(( \fR \alpha) : ( \fR \Phi) \right)_{d, \theta} = \left( (\widetilde{\fR} \alpha ) : \Phi \right)_{d, \theta},
\]
with, setting $(\cM f)(\bx) = f(M \bx)$ with $M = \begin{pmatrix}
    1 & 0 \\ 0 & - 1
\end{pmatrix}$,
\[
\widetilde{\fR} \begin{pmatrix}
    \alpha_{+,1} \\ \alpha_{+,2} \\ \alpha_{-,1} \\ \alpha_{-,2} 
\end{pmatrix}
= - \cM \begin{pmatrix}
    \alpha_{-, 2} \\ \alpha_{-, 1} \\ \alpha_{+,2} \\ \alpha_{+, 1}
\end{pmatrix}
= - \begin{pmatrix}
    0  &0 & 0 & \cM \\
    0 & 0 & \cM & 0 \\
    0 & \cM & 0 & 0 \\
    \cM & 0 & 0 & 0
\end{pmatrix}
\begin{pmatrix}
    \alpha_{+,1} \\ \alpha_{+,2} \\ \alpha_{-,1} \\ \alpha_{-,2} 
\end{pmatrix}.
\]
Note that this time, since $\fR \U = \U^* \fR$, the mirror symmetry exchanges the top and bottom layers. From the relations $\left[ \cS_d,  \widetilde{\fR}  \right]= 0$ and $ \left[ \cH_{d, \theta},  \widetilde{ \fR}  \right] = 0$, we infer
\begin{align*}
    & [\bbV_d]_{11}(M\bX)  =  \overline{[\bbV_d]_{22}(\bX)},  \quad
    [\bbV_d]_{12}( M \bX)  = \overline{[\bbV_d]_{12}(\bX)},  \\
    & [\bbV_d]_{21}(M \bX)  = \overline{[\bbV_d]_{21}(\bX)}, \quad
    [\bbV_d]_{22}(M \bX)  = \overline{[\bbV_d]_{22}(\bX)},
\end{align*}
and similar relations for $\Sigma_d$,  as well as
\begin{align*}
    & [\bbW^+]_{11}(M \bX) = [\bbW^-]_{22}(\bX), \quad [\bbW^+]_{12}(M \bX) = [\bbW^-]_{21}(\bX), \\
    & [\bbW^+]_{21}(M \bX) = [\bbW^-]_{12}(\bX), \quad [\bbW^+]_{22}(M \bX) = [\bbW^-]_{11}(\bX).
\end{align*}

\subsubsection{Extra symmetries at the limit}
In addition to the previously identified symmetries, extra symmetries appear due to the averaging process over all possible stacking. Indeed, we claim that
\begin{equation}
     \cC \ppa{f,g}^{\eta \eta'} = \ppa{g, f}^{\eta' \eta},  \qquad
     \cC \cP \ppa{f,g}^{+-} = \ppa{\cS g, \cS f}^{+-}, \qquad 
      \ppa{f,g}^{+-} = \ppa{\cC \cP g, \cC \cP f}^{+-}.
      \label{eq:extra_symm}
\end{equation}
The first equality comes directly from the definition~\eqref{eq:def:ppa}. To prove the second equality, we write
\begin{align*}
    \overline{\ppa{ f, g }^{+-} (-\bX)} & 
    = \int_{\Omega \times \R} f(\bx + \tfrac12 J \bX, z - \dd) \overline{g(\bx - \tfrac12J \bX, z + \dd)} \, \rd \bx \, \rd z \\
    & = \int_{\Omega \times \R}  \overline{g(\bx - \tfrac12 J \bX, -z + \dd)} f(\bx + \tfrac12 J \bX, -z - \dd) \, \rd \bx \, \rd z  \\
    & = \int_{\Omega \times \R}  \overline{\cS g(\bx - \tfrac12 J \bX, z - \dd)} \cS f(\bx + \tfrac12 J \bX, z + \dd) \, \rd \bx \, \rd z = \ppa{\cS g, \cS f}^{+-}(\bX).
\end{align*}
The third equality follows from similar arguments.

\medskip

From the first equality of~\eqref{eq:extra_symm}, we deduce that
\[
\left[ \bbW_d^+ \right]_{11}(\bX) = \left[ \bbW_d^- \right]_{22}(\bX), \quad \text{and} \quad 
\left[ \bbW_d^+ \right]_{22}(\bX) = \left[ \bbW_d^- \right]_{11}(\bX),
\]
and all these quantities are real-valued. Next, from the second equality of~\eqref{eq:extra_symm}, we obtain
\[
\Sigma_d(-\bX)^* = \Sigma_d(\bX), \quad 
\bbV_d(-\bX)^* = \bbV_d(\bX), \quad \text{and} \quad
\bbW_d^+(-\bX)^* = \bbW_d^-(\bX).
\]
Finally, from the third equality, we get
\[
\left[ \Sigma_d \right]_{11}(\bX) = \left[ \Sigma_d \right]_{22}(\bX), \quad 
\left[ \bbV_d \right]_{11}(\bX) = \left[ \bbV_d \right]_{22}(\bX), \quad 
\left[ \bbW_d^+ \right]_{12}(\bX)= \left[ \bbW_d^- \right]_{21}(\bX).
\]


\section{TBG Hamiltonians with nonlocal pseudopotentials}%
\label{sec:NLP}

In the presence of norm-conserving pseudopotentials, the TBG approximate Kohn-Sham Hamiltonian contains an additional term originating from the nonlocal contribution to the pseudopotential, that is
$$
H_{\theta,d}^{(2)}:= - \frac 12 \Delta + V_{d,\theta}^{(2)} + \U V_{\rm nl} \U^{-1}+ \U^{-1} V_{\rm nl} \U.
$$
where $V\ind{nl}$ is the $\L$-periodic nonlocal component of the pseudopotential of single-layer graphene.
For the Goedeker-Teter-Hutter pseudopotential implemented in the DFTK software used in our numerical simulations, $V_{\rm nl}$ is of the form
\begin{align*}
    V_{\rm nl} = v_0 \sum_{\bR \in \fC} | \tau_{\bR} \varphi \rangle \langle \tau_{\bR} \varphi |, \qquad \tau_{\bR} \varphi := \varphi(\cdot - \bR),
\end{align*}
where $v_0 > 0$, and where $\fC := \left( \L + \bb_1 \right) \cup \left( \L + \bb_2 \right)$ denotes the locations of the Carbon atoms in graphene. Here, $\bb_1 := \tfrac{2}{3} \ba_1 + \tfrac13 \ba_2$, $\bb_2 = \tfrac{1}{3} \ba_1 + \tfrac23 \ba_2$ are the positions of the atom in the Wigner-Seitz cell, and $\varphi$ is a radial normalized Gaussian function with variance similar to the characteristic radius of the carbon 1s orbital.

The functions $\Phi_j$ satisfy
$$
\left( - \frac 12 \Delta + V + V_{\rm nl} - \mu_{\rm F} \right) \Phi_j = 0.
$$
Since $\varphi$ is even while $\Phi_j$ is odd with respect to the horizontal symmetry plane, we have $\langle \varphi, \Phi_j \rangle = 0$, hence $V_{\rm nl} \Phi_j = 0$, and $\Phi_j$ also satisfies
\[
\left( - \frac 12 \Delta + V - \mu_{\rm F} \right) \Phi_j = 0.
\]

\medskip

We need to compute $H^{(2)}_{\theta, d} (\alpha : \Phi)_{d, \theta}$. For the Laplacian part, we can still use~\eqref{eq:split_laplacian}, and the local part of the potential can still act directly on the Bloch modes. We now focus on the non-local part of the potential. We have
\begin{align}
    &\bigg\langle (\widetilde \alpha : \Phi)_{d,\theta} \bigg|\left( \U V_{\rm nl} \U^{-1}+ \U^{-1} V_{\rm nl} \U \right) \bigg|
    (\alpha : \Phi)_{d,\theta} \bigg\rangle \label{eq:nonlocal_term} \\
    &\quad = \sum_{\eta,j} \sum_{\eta',j'} \sum_{\eta''} \langle \widetilde \alpha_{\eta,j}(\ept \cdot) \U^\eta \Phi_j | \U^{\eta''} V_{\rm nl} \U^{-\eta''} | \alpha_{\eta',j'}(\ept \cdot) \U^{\eta'} \Phi_{j'} \rangle \nonumber \\
    &\quad = v_0 \sum_{\eta,j} \sum_{\eta',j'} \sum_{\eta''} \sum_{\bR} \langle \U^{-\eta''} ( \widetilde \alpha_{\eta,j}(\ept \cdot) \U^\eta \Phi_j) |  \tau_{\bR}\varphi \rangle \langle \tau_{\bR}\varphi | \U^{-\eta''} ( \alpha_{\eta',j'}(\ept \cdot) \U^{\eta'} \Phi_{j'}) \rangle. \nonumber
\end{align}
Again, since $\varphi$ is even and $\Phi_{j}$ is odd with respect to the horizontal symmetry plane, we have that for all $\bR,\eta,j$,
$$
\langle \tau_{\bR}\varphi | \U^{-\eta} ( \beta(\ept \cdot) \U^{\eta} \Phi_{j}) \rangle = 0.
$$
Thus only the terms $\eta'' = - \eta' = - \eta$ contribute to the sum, which reduces to
\begin{align*}
    v_0 \sum_{\eta} \sum_{j,j'} \sum_{\bR} \langle \U^{\eta} ( \widetilde \alpha_{\eta,j}(\ept \cdot) \U^\eta \Phi_j) |  \tau_{\bR}\varphi \rangle \langle \tau_{\bR}\varphi | \U^{\eta} ( \alpha_{\eta,j'}(\ept \cdot) \U^{\eta} \Phi_{j'}) \rangle.
\end{align*}


We now compute the above inner products. We have
\begin{align*}
    \langle \tau_{\bR}\varphi | \U^{\eta} ( \beta(\ept \cdot) \U^{\eta} \Phi_{j}) \rangle
     & = \int_{\R^3} \beta\left(\ept R_{-\eta \frac\theta 2}^* \bx\right) \, \overline{\varphi(\bx - \bR,z)} \; \Phi_{j}\left(R_{-\eta \theta}^*\bx,z-\eta d\right) \, \rd\bx \, \rd z \\
    &  = \int_{\R^3} \beta\left(\ept R_{-\eta\frac\theta 2}^* (\by+\bR) \right) \; \overline{\varphi(\by,z)} \; \Phi_{j}\left(R_{-\eta \theta}^*(\by+\bR),z-\eta d\right) \, \rd \by \, \rd z \\
    &  = \int_{\R^3} \beta\left(\ept R_{-\eta\frac\theta 2}^* \bR + \ept R_{\eta \frac{\theta}{2}} \by \right) \; \overline{\varphi(\by,z)} \; \Phi_{j}\left(R_{-\eta \theta}^*\bR + \by,z-\eta d\right) \, \rd \by \rd z,
\end{align*}
where we used that $\varphi$ is radially symmetric in the last line. Since $\varphi$ is localized near $\by = \bold{0}$, we perform a Taylor expansion of $\beta$ in $\by$. This one takes the form
\[
\beta\left(\ept R_{-\eta\frac\theta 2}^* \bR + \ept R_{-\eta \frac{\theta}{2}} \by \right) = 
\beta\left(\ept R_{-\eta\frac\theta 2}^* \bR \right) + \ept \nabla \beta\left(\ept R_{-\eta\frac\theta 2}^* \bR \right) \cdot R_{-\eta \frac{\theta}{2}} \by + \ldots
\]
We obtain a series in $\ept$, whose first terms are given by
\[
\langle \tau_{\bR}\varphi | \U^{\eta} ( \beta(\ept \cdot) \U^{\eta} \Phi_{j}) \rangle = 
\beta\left(\ept R_{-\eta\frac\theta 2}^* \bR \right) F_{0,j}(R_{\eta \theta}^*\bR) + \ept R_{-\eta \frac{\theta}{2}}^* \nabla \beta\left(\ept R_{-\eta\frac\theta 2}^* \bR \right) \cdot G_{1,j}(R_{-\eta \theta}^*\bR) + \ldots
\]
where
\begin{align*}
    F_{0, j}(\bX) & := \int_{\R^3} \overline{\varphi(\by,z)} \; \Phi_{j}\left( \bX + \by,z-\eta d\right) \, \rd \by \, \rd z, \\
    G_{1,j}(\bX) & := \int_{\R^3} \by \overline{\varphi(\by,z)} \; \Phi_{j}\left(\bX + \by,z-\eta d\right) \, \rd \by \, \rd z. 
\end{align*}
We can go on with the Taylor expansion, but we stop here, as the leading order is already enough for numerical purpose (we show below that even the leading order term can be neglected). The leading order term of~\eqref{eq:nonlocal_term} is therefore
\begin{align*}
    & v_0 \sum_{\eta, j, j'} \sum_{\bR \in \fC} (\overline{ \widetilde{\alpha_{\eta, j}} } \alpha_{\eta, j'})\left(\ept R_{-\eta\frac\theta 2}^* \bR \right) \left[ \overline{F_{0, j}} F_{0, j'} \right] (R_{\eta \theta}^*\bR).
\end{align*}
Performing a Taylor expansion in $\ept$ of the $F$-functions using the equality $R_{\eta \theta}^* = c_\theta - 2 \ept \eta J = 1 - 2 \ept \eta J + O(\ept^2)$, we obtain, for $s \in \{ 1, 2\}$ and $\bR \in \Lat$, 
\begin{align*}
    F_{0, j}(R_{-\eta \theta}^* (\bR + \bb_s)) &
    \approx F_{0, j} \left( (\bR + \bb_s) - 2 \ept \eta J (\bR + \bb_s) \right) 
    \\
    & = \int_{\R^3} \overline{\varphi(\by,z)} \; \Phi_{j}\left( \bR + \bb_s - 2 \ept \eta J (\bR + \bb_s) + \by ,z-\eta d\right) \rd \by \rd z \\
    & = \re^{ \ri \bK \cdot \bR} \widetilde{F}_{0, s, j, \eta}(\ept(\bR + \bb_s)),
\end{align*}
with
\[
\widetilde{F}_{0, s, j, \eta}(\bX) := \int_{\R^3} \overline{\varphi(\by,z)} \; \Phi_{j}\left(\bb_s -2 \eta J \bX + \by ,z-\eta d\right) \rd \by \rd z.
\]
A similar expansion for the $\alpha$-functions gives, to leading order,
\[
v_0 \sum_{\eta, j, j'} \sum_{s \in \{ 1, 2\}} \sum_{\bR \in \Lat} (\overline{ \widetilde{\alpha_{\eta, j}} } \alpha_{\eta, j'})\left(\ept (\bR + \bb_s) \right) \left[ \overline{\widetilde{F}_{0, s, j}} \widetilde{F}_{0,s, j'} \right] (\ept(\bR + \bb_s)).
\]
We recognize a Riemann sum. Therefore, at leading order, this term also equals (the superscript $-1$ refers to the fact that only the leading order term was taken into account)
\[
\frac{1}{\ept^2} \sum_{\eta, j, j'} \int_{\R^2} (\overline{ \widetilde{\alpha_{\eta, j}} } \alpha_{\eta, j'})(\bX) [ \bbW_{d}^{\rm nl,-1}]_{jj'}(\bX)\rd \bX, 
\quad \text{with} \quad [ \bbW_{d}^{\rm nl,-1}]_{jj'} := v_0 \sum_{s \in \{1, 2\} } \overline{\widetilde{F}_{0, s, j}} \widetilde{F}_{0,s, j'}.
\]
This term modifies the effective Hamiltonian as follows:
\begin{equation}\label{def:Hred_nl}
    \cH_{d,\theta}^{\rm nl,-1} = \ept^{-1} \left(\cV_d+\cV_d^{\rm nl,-1} \right) + c_\theta T_{d} + \ept T_{d}^{(1)},
\end{equation}
where
\begin{align}
    \cV_{d}^{\rm nl,-1} &:= \left( \begin{array}{cc} \bbW_{d}^{\rm nl, -1} & 0 \\ 0 & \bbW_{d}^{\rm nl,-1} \end{array} \right).
\end{align}
The remainder in the three-term expansion \eqref{eq:asymptotic_VF} valid for local (multiplicative) potentials is of order $\ept^\infty$. In contrast, the non-local component of the pseudopotential gives rise to an infinite series of terms of orders $\ept^k$ for all $k \ge -1$. All these terms, including the leading order one are very small ($\sim 6 \times 10^{-1}$ meV, see Fig.~\ref{fig:W_Vnl}).

\begin{figure}[H]
    \begin{center}

        \includegraphics[width=0.3\textwidth]{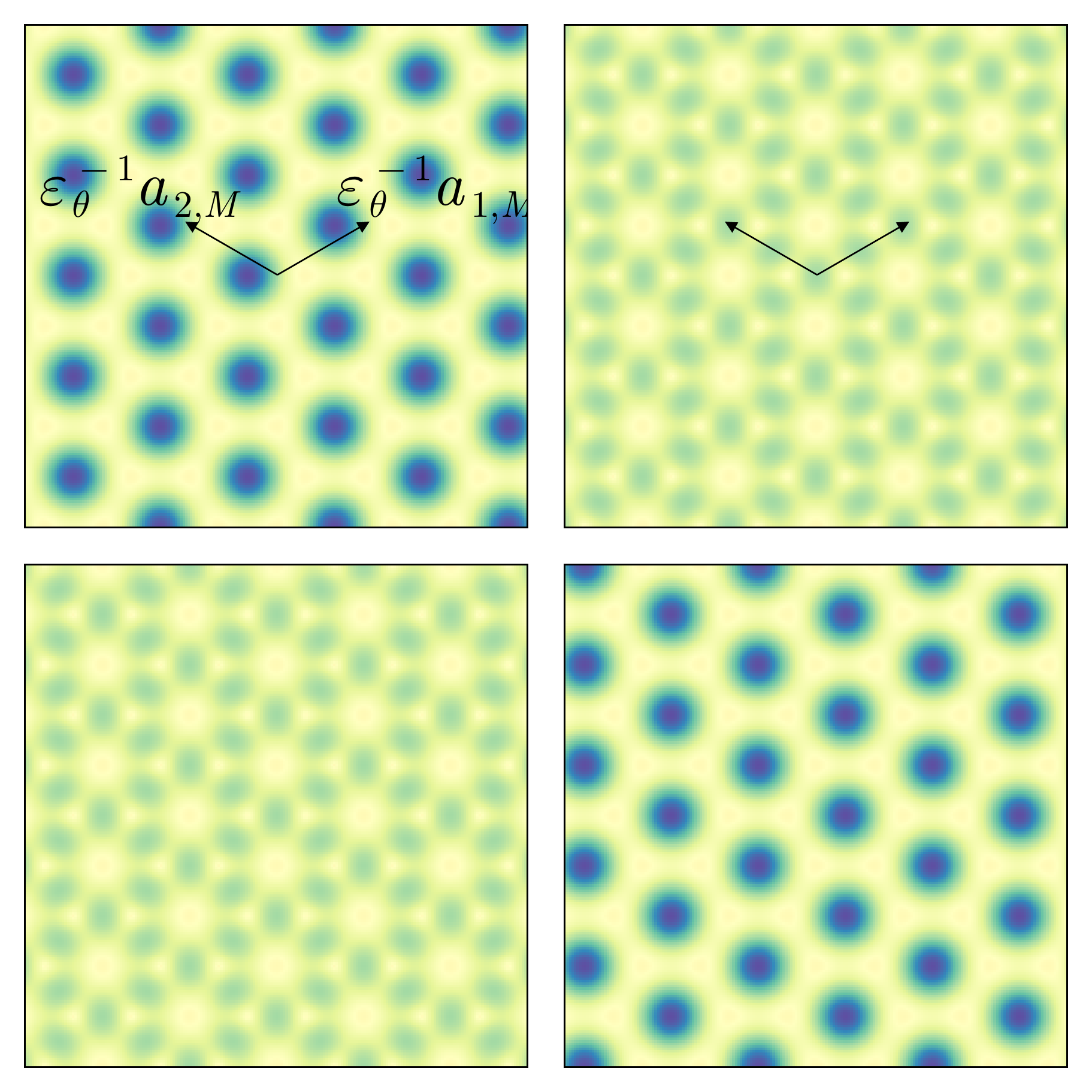}
        \hspace{1cm}\includegraphics[width=0.1\textwidth]{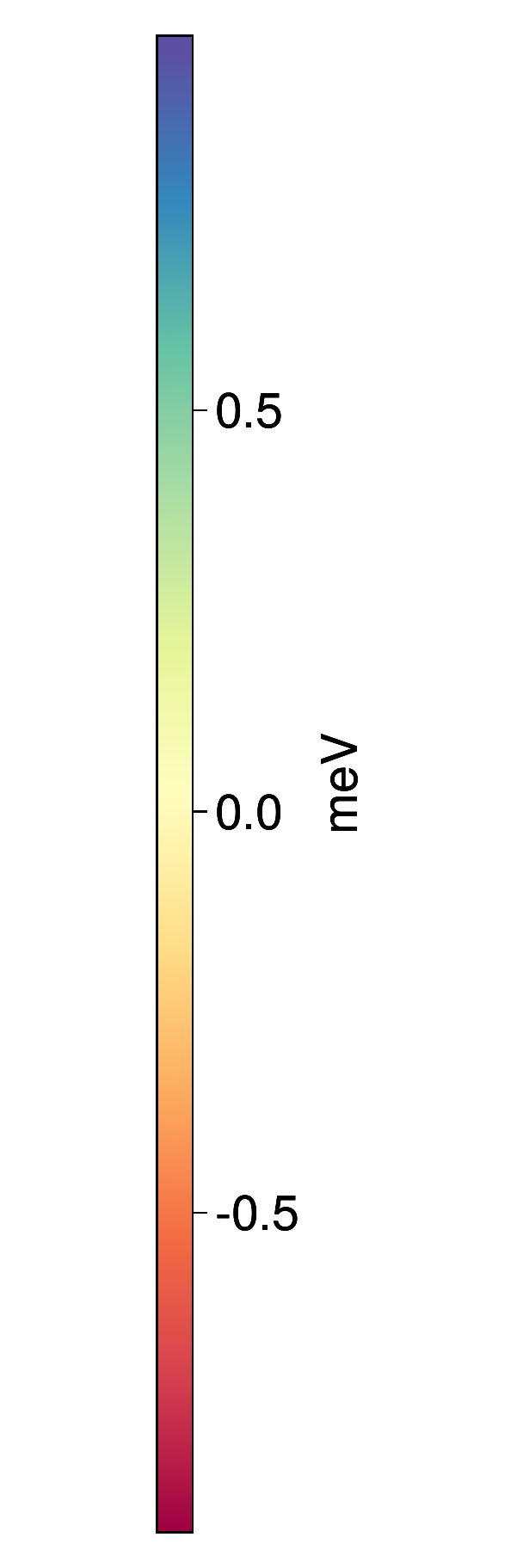}
        \caption{Moduli of the $4$ entries of the matrix-valued function $\bbW^{{\rm nl},-1}_d(\bX)$. }
        \label{fig:W_Vnl}
    \end{center}
\end{figure}


\section{Variation of the reduced quantities as a function of $d$}
Let us finally study the variations of the fields $\Sigma_d(\bX)$, $\bbV_d(\bX)$ and $\bbW_d^\pm(\bX)$ as functions of $d$. Let us expand these three fields of interest as
\begin{align*}
    \Sigma_d(\bX) &=  [u_1 | u_1]_{d,0,0} \bm S^{(0)}(\bX) + r_{\Sigma,d}(\bX), \\
    \bbV_d(\bX) &= w_{\rm AA} \bm S^{(0)}(\bX) + [V_d u_1 | u_1]_{d,-1,-1} \bm S^{(1)}(\bX) + r_{\bbV,d}(\bX), \\
    \bbW_d^+(\bX) &=  W_{\text{int},d}^+ \bbI_2 +   \bm{W}^+_d(\bX)  + r_{{\bbW},d}(\bX), 
\end{align*}
where $[u_1 | u_1]_{d,0,0}$, $w_{\rm AA} = [V_d u_1 | u_1]_{d,0,0}$, $[V_d u_1 | u_1]_{d,-1,-1}$, and $W_{\text{int},d}^+$ are real numbers, and where 
\begin{align*}
    \bm S^{(0)}(\bX) &:= \mat{G(\bX) & \overline{F(-\bX)} \\ F(\bX) & G(\bX)}, \qquad \bm S^{(1)}(\bX) := \mat{G^{(1)}(\bX) & \overline{F^{(1)}(-\bX)} \\ F^{(1)}(\bX) & G^{(1)}(\bX)}, \\
    \bm{W}^+_d(\bX) &:=  \mat{\ab{[\ab{u_1}^2 | V_d]_{d,0,1}} D(\bX) & \ab{[\overline{u_1}u_2 | V_d]_{d,0,1}} \overline{Z(-\bX)} \\ \ab{[\overline{u_1}u_2 | V_d]_{d,0,1}} Z(\bX) & \ab{[\ab{u_1}^2 | V_d]_{d,0,1}} D(\bX)},
\end{align*}
with 
\begin{align*}
    G^{(1)}(\bX) &:= e^{-i \bq_1 \cdot \bX} \pa{e^{i \pa{\ba_1^* - \ba_2^*} \cdot \bX} + e^{i \pa{\ba_2^* - \ba_1^*} \cdot \bX} + e^{-i \pa{\ba_1^* + \ba_2^*} \cdot \bX}} \\
    &= e^{i2 \bq_1 \cdot \bX} + e^{i2 \bq_2 \cdot \bX} + e^{i2 \bq_3 \cdot \bX}, \\
    F^{(1)}(\bX) &:= e^{-i \bq_1 \cdot \bX} \pa{ \omega e^{i \pa{\ba_1^* - \ba_2^*} \cdot \bX} + \omega^2 e^{i \pa{\ba_2^* - \ba_1^*} \cdot \bX} + e^{-i \pa{\ba_1^* + \ba_2^*} \cdot \bX}}  \\ 
    &= e^{i2 \bq_1 \cdot \bX} + \omega e^{i2 \bq_2 \cdot \bX} + \omega^2 e^{i2 \bq_3 \cdot \bX}, \\
    D(\bX) &:=  \omega \pa{e^{-i\ba_2^* \cdot \bX} + e^{i\ba_1^* \cdot \bX} + e^{i\pa{-\ba_1^* + \ba_2^*} \cdot \bX}} + \omega^2 \pa{e^{-i\ba_1^* \cdot \bX} + e^{i\ba_2^* \cdot \bX} + e^{i\pa{\ba_1^* - \ba_2^*} \cdot \bX} }  \\
    & =  \omega \pa{e^{i\pa{\bq_1-\bq_3}\cdot \bX}  + e^{i\pa{\bq_2-\bq_1}\cdot \bX}+ e^{i\pa{\bq_3-\bq_2}\cdot \bX}} + \text{c.c.}, \\
    Z(\bX) &:=  \omega \pa{e^{i\ba_2^* \cdot \bX} + e^{-i\ba_2^* \cdot \bX}} + \omega^2 \pa{e^{i\ba_1^* \cdot \bX} + e^{-i\ba_1^* \cdot \bX}}  + e^{i\pa{\ba_1^* - \ba_2^*} \cdot \bX} + e^{i\pa{-\ba_1^* + \ba_2^*} \cdot \bX} \\
    &=  2\pa{\cos\pa{\pa{\bq_3 - \bq_2}\cdot \bX} + \omega \cos\pa{\pa{\bq_3 - \bq_1}\cdot \bX} + \omega^2 \cos\pa{\pa{\bq_2 - \bq_1}\cdot \bX}}.
\end{align*}
The above four functions are expanded on the subdominant Fourier modes and fulfill the expected symmetries.

We infer from the data in Figure \ref{fig:dep_d} that for all values of $d$ around the experimental average interlayer distance $d \simeq 6.45$ bohr, we have
\begin{align*}
    \nor{r_{\Sigma,d}}{L^2} &\ll \nor{[u_1 | u_1]_{d,0,0} \bm S^{(0)}}{L^2} \simeq \nor{\Sigma_d}{L^2}, \\
    \nor{r_{\bbV,d}}{L^2} &\ll \nor{[V_d u_1 | u_1]_{d,-1,-1} \bm S^{(1)}}{L^2} \ll \nor{w_{AA} \bm S^{(0)}}{L^2}  \simeq \nor{\bbV_d}{L^2}, \\
    \nor{r_{\bbW,d}}{L^2} &\ll \nor{\bm{W}^+_d}{L^2}  \ll \nor{W_{\text{int},d}^+ \bbI_2}{L^2} \simeq  \nor{\bbW_d^+}{L^2},
\end{align*}
where $\nor{\bm S^{(0)}}{L^2} = \nor{\bm S^{(1)}}{L^2} \simeq 0.8$. We also see that the main corrections with respect to the Bistritzer-MacDonald model come from the term $\bf{W}^+_d$ and from the term involving $\nabla\Sigma_d$.
\begin{figure}[H]
    \begin{center}
        \includegraphics[width=0.35\textwidth]{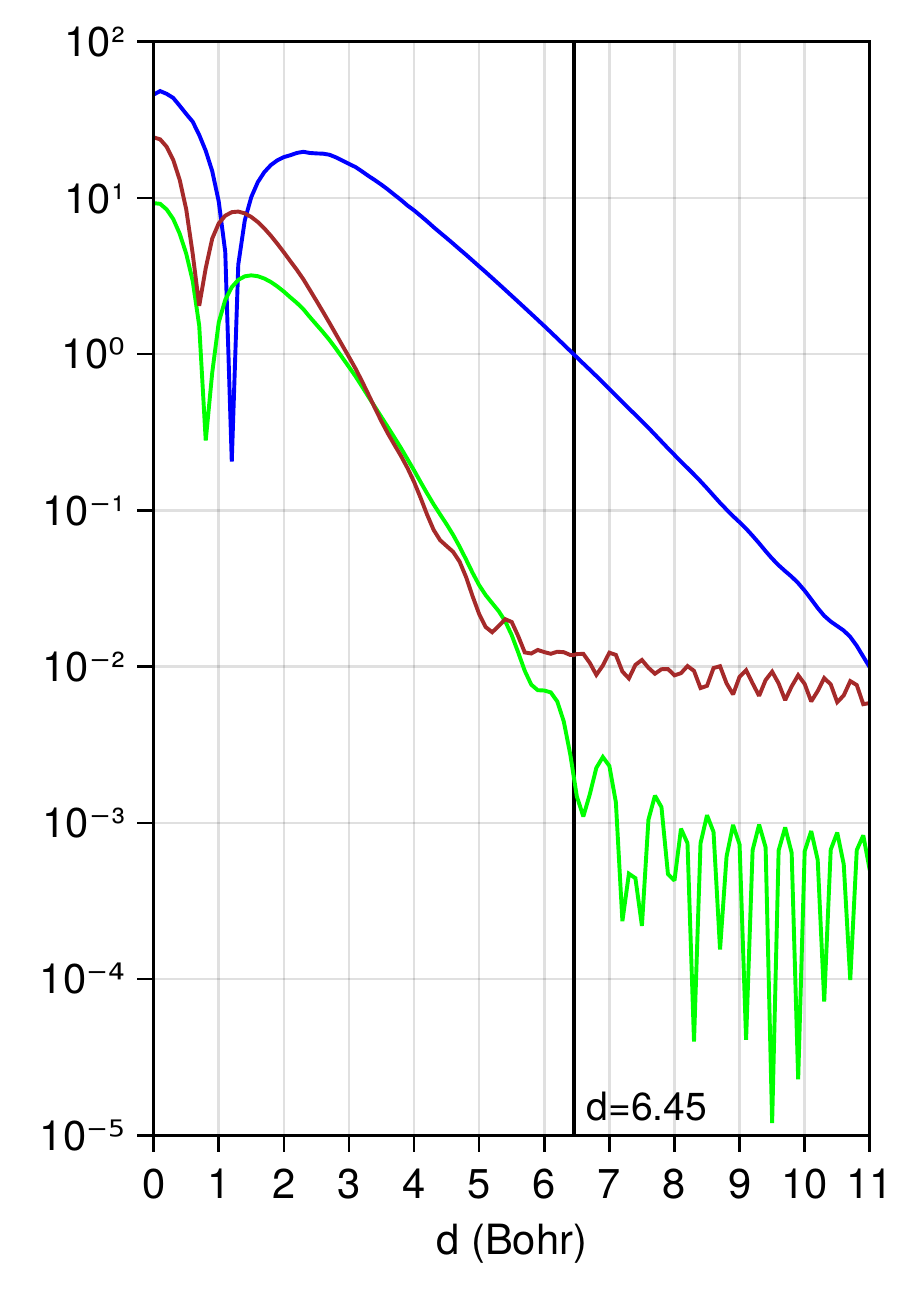}
        \includegraphics[width=0.35\textwidth]{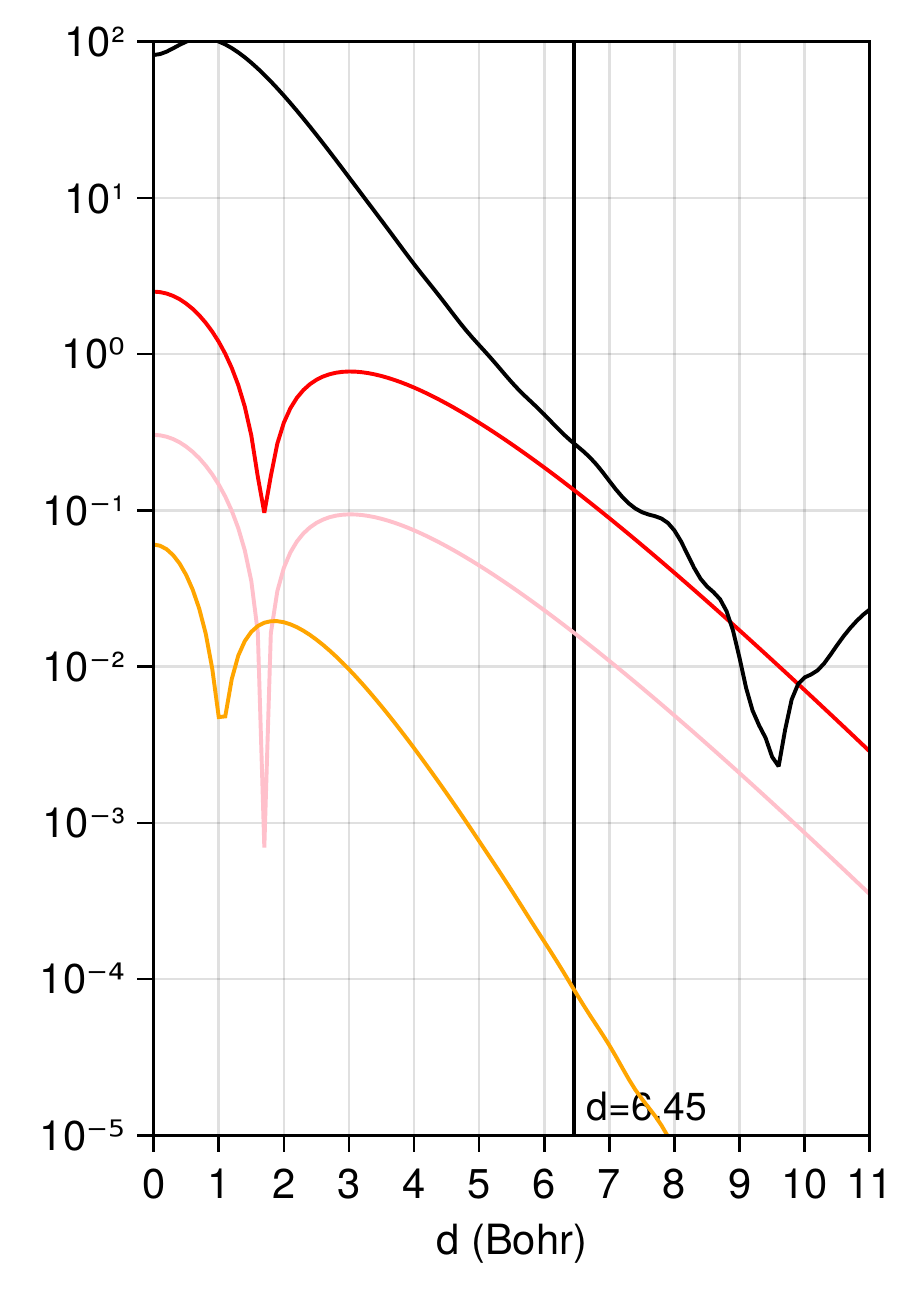}
        \includegraphics[width=0.2\textwidth]{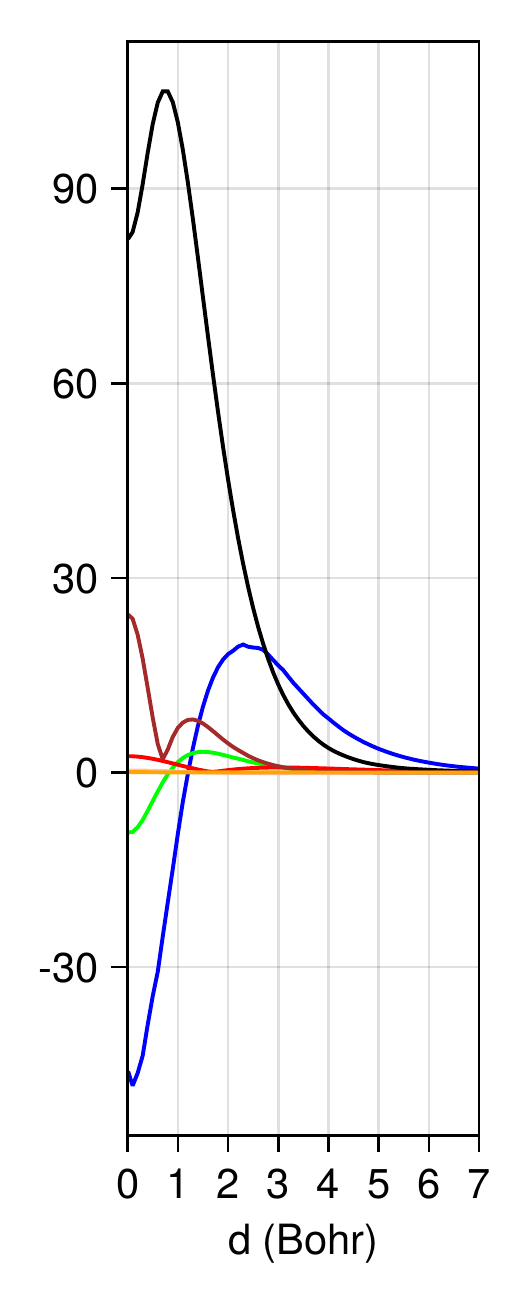}
        \caption{Variations of the reduced model parameters with respect to the interlayer distance $d$. Left: semi-log plots of the functions $d\mapsto w\ind{AA}^d/\wAAd$ (blue), $d\mapsto [V_d u_1 | u_1]_{d,-1,-1}/\wAAd$ (green), $d \mapsto \nor{r_{\bbV,d}}{L^2}/\wAAd$ (brown). The oscillations of the brown and green lines for $d > 5$ are due to numerical noise. Center: semi-log plots in log scale of  $d \mapsto \nor{\bm W_d^+}{L^2}/\wAAd$ (black), $d \mapsto \nor{\na\Sigma_d}{L^2}/v\ind{F}$ (red), $d \mapsto \nor{\Sigma_d}{L^2}$ (pink), and $d \mapsto \nor{r_{\Sigma,d}}{L^2}$ (orange). Right: plots of the seven functions.}\label{fig:dep_d}
    \end{center}
\end{figure}

\end{document}